\documentclass[11pt,a4paper,reqno]{article}
\usepackage{amsmath}
\usepackage{graphicx}
\usepackage{amsfonts}
\usepackage{amssymb}
\usepackage{amsthm}
\usepackage[left=2.5cm, right=2.5cm, top=2.5cm, bottom=2.5cm]{geometry}
\usepackage{indentfirst}
\usepackage[all]{xy}
\usepackage[colorlinks=true,linkcolor=blue]{hyperref}
\usepackage{mathrsfs} 
\usepackage{tikz-cd}
\usetikzlibrary{graphs,decorations.pathmorphing,decorations.markings}
\usepackage{enumitem}
\usepackage[title]{appendix}
\setitemize{noitemsep,topsep=0pt,parsep=0pt,
partopsep=0pt,itemindent=12pt,leftmargin=0pt}
\makeatletter
\newcommand{\subjclass}[2][2020]{%
  \let\@oldtitle\@title%
  \gdef\@title{\@oldtitle\footnotetext{#1 \emph{Mathematics subject classification.} #2}}%
}
\newcommand{\keywords}[1]{%
  \let\@@oldtitle\@title%
  \gdef\@title{\@@oldtitle\footnotetext{\emph{Key words and phrases.} #1.}}%
}
\makeatother


\newtheorem{theorem}{Theorem}[section]
\newtheorem{lemma}[theorem]{Lemma}
\newtheorem{proposition}[theorem]{Proposition}
\newtheorem{corollary}[theorem]{Corollary}

\theoremstyle{definition}
\newtheorem{definition}[theorem]{Definition}

\theoremstyle{remark}
\newtheorem{remark}[theorem]{Remark}

\newcommand{\build}[3]{\mathrel{\mathop{\kern 0pt#1}\limits_{#2}^{#3}}}

\newcommand\SU{{\mathrm{SU}}}
\newcommand\SO{{\mathrm{SO}}}
\newcommand\Sp{{\mathrm{Sp}}}
\newcommand\GL{{\mathrm{GL}}}

\newcommand\U{{\mathrm U}}
\newcommand\Z{{\mathbb Z}}
\newcommand\Q{{\mathbb Q}}
\newcommand\N{{\mathbb N}}

\newcommand\R{\mathbb{R}}
\newcommand\C{\mathbb{C}}
\newcommand\Tr{\mathrm{Tr}}
\newcommand\vol{\mathrm{vol}}
\newcommand\Pbb{\mathbb{P}}
\newcommand\Pfr{\mathscr{P}}
\newcommand\E{\mathbb{E}}
\newcommand\Gfr{\mathscr{G}}
\newcommand\Ufr{\mathscr{U}}

\newcommand\YM{\mathrm{YM}}
\newcommand\Hom{\mathrm{Hom}}

\newcommand\Tbb{\mathbb{T}}
\newcommand\ev{{\mathrm{ev}}}

\title{The central heat trace on large compact classical groups}

\author{Thibaut Lemoine\thanks{Universit\'e de Strasbourg, CNRS, UMR 7501 -- Institut de Recherche Math\'ematique Avanc\'ee, 7 rue Ren\'e Descartes, 67000 Strasbourg, France. thibaut.lemoine@math.unistra.fr}, Myl\`ene Ma\"ida\thanks{Université de Lille, CNRS, Inria, UMR 8524 - Laboratoire Paul Painlevé, F-59000 Lille, France}}

\keywords{Heat kernel on Lie groups, asymptotic representation theory, random partitions, random surfaces, Yang--Mills theory, gauge/string duality}

\subjclass{Primary 05A17, 14H30, 43A75, 58J50, Secondary 81T13, 81T35} 

\begin{document}

\maketitle

\begin{abstract}
We study the large-$N$ asymptotics of the central trace of the heat kernel on compact classical groups. For every classical family $G_N\subset \mathrm{GL}_N(\C)$, we prove a full large-$N$ asymptotic expansion, using a highest weights/partitions correspondence adapted to the large-rank regime, under which the eigenvalues of the Laplace--Beltrami operator stabilize as observables in the algebra of shifted symmetric functions. Then, we prove a random surface representation of the trace in terms of ramified coverings of the torus. We provide two independent applications: an explicit large-rank counting law for the Casimir spectrum, with exponential Hardy--Ramanujan-type growth in contrast with the polynomial behavior of Weyl's law at fixed rank, and a rigorous probabilistic formulation of the Yang--Mills/Hurwitz duality on a two-dimensional torus initiated by Gross and Taylor \cite{GT,GT2}, completing a previous work \cite{LM2} of the authors. We also extend this duality to a Yang--Mills/Gromov--Witten duality by expressing the coefficients of the central heat trace as explicit functionals of the generating function of Gromov--Witten invariants.
\end{abstract}

\tableofcontents

\section{Introduction}

Let $G_N\subset\GL_N(\C)$ be a compact matrix Lie group, such that its Lie algebra $\mathfrak{g}_N$ is endowed with an invariant inner product (the explicit choice will be given in~\eqref{eq:inner_product}). It yields a Riemannian structure on $G_N$, so that the heat kernel $p_t=e^{\frac{t}{2}\Delta_{G_N}}$ is expressed in terms of the corresponding Laplace--Beltrami operator $\Delta_{G_N}.$ By standard results of representation theory, the characters $\{\chi_\lambda,\lambda\in\widehat{G}_N\}$ of the irreducible representations of $G_N$ form an orthonormal family of eigenfunctions of $\Delta_{G_N}$, with associated eigenvalues $\{-c_2(\lambda),\lambda\in\widehat{G}_N\}$:
\[
\Delta_{G_N}\chi_\lambda(g)=-c_2(\lambda)\chi_\lambda(g),\quad \forall \lambda\in\widehat{G}_N \text{ and } \forall g\in G_N.
\]
The number $c_2(\lambda)\geq 0$ is called \emph{Casimir number}, because the Laplace--Beltrami operator $\Delta_{G_N}$ can be identified (up to a sign) with the quadratic Casimir operator on the enveloping algebra of the Lie algebra of $G_N$.

For any fixed time $t>0$, $p_t$ is a convolution operator, acting in particular on two Hilbert spaces: $L^2(G_N)$ on the one hand, and $\mathcal{H}_{G_N}=\{f\in L^2(G_N):f(hgh^{-1})=f(g),\ \forall g,h\in G_N\}$ on the other hand. As an operator on $L^2(G_N)$, its trace is
\[
\Tr_{L^2(G_N)}(p_t)=\sum_{\lambda\in\widehat{G}_N}d_\lambda^2e^{-\frac{t}{2}c_2(\lambda)},
\]
whereas as an operator on $\mathcal{H}_{G_N}$ it is
\begin{equation}\label{eq:hk_trace_ini}
\Tr_{\mathcal{H}_{G_N}}(p_t)=\sum_{\lambda\in\widehat{G}_N}e^{-\frac{t}{2}c_2(\lambda)},\quad \forall t>0.
\end{equation}
In this paper, we will denote by $\Tr(p_t)=\Tr(e^{\frac{t}{2}\Delta_{G_N}})$ the trace over $\mathcal{H}_{G_N}$, which will be the main protagonist of our results (Theorem~\ref{thm:asympt_expansion}, Theorem~\ref{thm:main}, Theorem~\ref{thm:coef_Hurwitz} and Theorem~\ref{thm:Cardy}). To make a distinction with the usual heat trace $\Tr_{L^2(G_N)}(p_t)$, we shall call it the \emph{central heat trace} on $G_N$, because we restrict to square integrable \emph{central} functions.

\subsection{Large-$N$ expansion of the trace}

The trace of the heat kernel acting on $L^2(G_N)$ is known to admit a short-time asymptotic expansion \cite{Feg}, and large deviations when $G_N=\U(N)$ and $N\to\infty$ \cite{LevMai2}, but there are fewer results about the one acting on $\mathcal{H}_{G_N}$: for instance, as far as we are aware, its short-time asymptotics are unknown, and its large-$N$ asymptotics have been mainly studied by theoretical physicists at the level of formal power series \cite{GT,GT2,Dij}. The purpose of this paper is to investigate the large-$N$ asymptotics of the central heat trace at a quantitative level. In the sequel, we shall use the following modified version of Landau's asymptotic notations: if $f,g:(0,\infty)\times \N\to\R$ are functions that depend on a real number $t>0$ and an integer $N$, we write $f(t,N)=O_t(g(t,N))$ as $N\to\infty$ if for any $t$, there exists a constant $C(t)$ \emph{that may depend on} $t$ and an integer $N_0$, such that
\[
\vert f(t,N)\vert \leq C(t)\vert g(t,N)\vert,\quad \forall t>0,\ \forall N\geq N_0.
\]

The fact that $\Tr(e^{\frac{t}{2}\Delta_{G_N}})$ converges to a finite value as $N\to\infty$ has been first conjectured in physics \cite{Gro,Dou}, then proved in \cite{Lem,DL}. Our first result is a far-reaching extension of the results of \cite{DL}, namely a full asymptotic expansion of the trace of the heat kernel for compact classical groups (including unitary, special unitary, special orthogonal and symplectic groups). 

\begin{theorem}[See Theorem~\ref{thm:asympt_expansion_detailed}]\label{thm:asympt_expansion}
Let $(G_N)_{N\geq 1}$ be a sequence of compact classical groups with $G_N\subset\GL_N(\C)$. For any $t>0$, there is an explicit family of coefficients $(a_k(t))_{k\geq 0}$ (depending on the group type) such that for any $p\geq 1$, the following asymptotic expansion holds as $N$ tends to infinity:
\begin{equation}\label{eq:asympt_expansion}
\Tr(e^{\frac{t}{2}\Delta_{G_N}})=a_0(t)+\frac{a_1(t)}{N}+\ldots+\frac{a_p(t)}{N^p}+O_t\left(N^{-p-1}\right).
\end{equation}
\end{theorem}
The case of $\U(N)$ was recently proved by the same authors \cite[Theorem 1.2]{LM2}, and the proof of Theorem~\ref{thm:asympt_expansion} relies on similar tools: a description of $\Tr(e^{\frac{t}{2}\Delta})$ in terms of random partitions, and a precise control of the remainder in terms of tail probabilities with respect to the corresponding measure. Although the expansion involves only even powers of $\frac1N$ in the (special) unitary case, all powers of $\frac1N$ appear for the special orthogonal and symplectic cases. By taking $p=0$ and letting $N\to\infty$, we also obtain that $\Tr(e^{\frac{t}{2}\Delta_{G_N}})\to a_0(t)$, and the limits agree with the results of \cite{Lem,DL}, see Corollary~\ref{cor:lim_pf}. All coefficients of~\eqref{eq:asympt_expansion} are explicitly expressed as expectations with respect to the $q$-uniform measure on partitions, which reveals a clear relationship with integrable hierarchies \cite{BO,OP2}.

We should emphasize the fact that it is only an \emph{asymptotic expansion}, in the sense that there is no hope to let $p$ tend to infinity with fixed $N$; indeed, our estimates only work for $N$ large enough, which is a condition that actually depends on $p$. However, the coefficients of the expansion do not depend on $N$ or $p$. Although Theorem~\ref{thm:asympt_expansion} is stated and proved for the heat kernel, its proof can be readily adapted to study more general measures, such as the central symmetric probability measures induced by infinitely divisible probability measures on $\R$ introduced by Applebaum \cite{App} --- see Remark~\ref{rmk:extension_Levy}.

The asymptotic expansion of Theorem~\ref{thm:asympt_expansion} serves as the analytic backbone of the paper, from which we derive several structural consequences ranging from enumerative geometry to spectral theory. In particular, we will provide:
\begin{enumerate}
\item an interpretation of the heat trace as a model of random surfaces, which is a probabilistic counterpart of the Yang--Mills/Hurwitz duality on a torus, which was  only known at the level of formal power series \cite{GT,GT2} (see Theorems \ref{thm:main} and  \ref{thm:mainbis}),
\item an interpretation of the coefficients of the expansion in terms of the generating function of Gromov--Witten (GW) invariants on an elliptic curve (see Theorem \ref{thm:coef_Hurwitz}),
\item an expression of the counting statistics of Casimir eigenvalues which differs from Weyl's law on compact Riemannian manifolds (see Theorem \ref{thm:Cardy}).
\end{enumerate}
The structure behind the first two consequences can be summarized as follows:
\[
\text{Classical group} \build{\leftrightarrow}{}{} \text{Highest weights} \build{\leftrightarrow}{}{} \text{partitions} \build{\leftrightarrow}{}{} \text{permutations} \build{\leftrightarrow}{}{} \text{coverings} \build{\leftrightarrow}{}{} \text{GW invariants}.
\]


More precisely, the asymptotic expansions are based on representations of highest weights of classical groups in terms of integer partitions that not only stabilize in the large-$N$ limit, but also produce asymptotic results that can be controlled. It can be roughly summarized in the following equation, that we will make more precise in Proposition~\ref{prop:shifted-casimir}: for any sequence $(G_N)_{N\geq 1}$  of compact classical groups with $G_N\subset\GL_N(\C),$ there is a commutative algebra $\mathcal{A}$ and a realization $\kappa_N$ of $-\Delta_{G_N}$ for all $N$ such that $\kappa_N$ admits an expansion of the form
\[
\kappa_N = E +\frac{1}{N}L + O(N^{-2}),
\]
where the algebra $\mathcal{A}$ and the observables $E$ and $L$ depend on the type of the group.  The leading term $E$ corresponds to some \emph{energy operator} while the first-order correction $L$ is governed by the eigenvalue of the well-known \emph{cut-and-join operator} acting on the algebra of symmetric functions \cite{Gou94}. In Hurwitz theory, this operator appears as the generator of the evolution equation satisfied by generating series of branched coverings, and plays a central role in the Bouchard--Mari\~no conjecture and its proof via topological recursion \cite{EMS11}. The appearance of the same content statistic suggests a deeper connection between large-rank harmonic analysis on compact groups and the algebraic structures underlying Hurwitz theory, and our next result is actually a representation of such a connection.

Let $n\geq 1$ and $k\geq 0$ be two fixed integers, and $\Sigma_g$ be a compact connected orientable surface of genus $g\geq 1$. The \emph{Hurwitz space} $\mathcal{H}_g(n,k)$ is the set of equivalence classes of ramified coverings $X\to\Sigma_g$ of degree $n$ with $k$ ramification points, all assumed to be generic. It is finite (its size is the Hurwitz number $H_g(n,k)$), and can be endowed with the uniform probability measure. Associated random models have already been studied, in particular in the large-$n$ limit: for $g=1$ \cite{EO,Agg}, or $g\geq 2$ and $k=0$ \cite{MNP,MagPud}. In the present paper, we shall fix $g=1$, but the novelty compared to the previously mentioned references is that we assume that $n$ and $k$ are random (following respectively a geometric distribution of parameter $1-e^{-\frac{t}{2}}$ and an even Poisson distribution of parameter $t$). Conditionally to $n$ and $k$, and assuming only generic ramifications, the random (equivalence classes of) ramified coverings are taken uniformly in $\mathcal{H}_1(n,k)$. The corresponding measure $\rho_t$ (for a fixed real parameter $t>0$) will be defined and studied in Section~\ref{sec:Hurwitz_integration}.

\begin{theorem}[See Theorem \ref{thm:mainbis}]\label{thm:main}
Let $(G_N)_{N\geq 1}$ be a sequence of compact classical groups with $G_N\subset\GL_N(\C)$. For any $p\geq 1$ and any $t>0$:
\begin{itemize}
\item If $G_N=\U(N)$ or $\SU(N)$, there is an explicit function $\Phi_{t,N}^p$ (that depends on the group type) such that, as $N$ tends to infinity,
\begin{equation}
\Tr(e^{\frac{t}{2}\Delta_{G_N}})=\int \Phi_{t,N}^{p}(X_1,X_2)N^{\chi(X_1)+\chi(X_2)}d\rho_t^{\otimes 2}(X_1,X_2)+O_t(N^{-2p-2}),
\end{equation}
where $\chi(X)$ is the Euler characteristic of the covering $X.$
\item Otherwise, there is an explicit function $\Phi_{t,N}$ (that depends on the group type) such that, as $N$ tends to infinity,
\begin{equation}
\Tr(e^{\frac{t}{2}\Delta_{G_N}})=\int \Phi_{t,N}(X)N^{\chi(X)}d\rho_t(X)+O_t(N^{-p-1}).
\end{equation}
\end{itemize}
\end{theorem}

In the case of $\U(N)$ or $\SU(N)$, $\Phi_{t,N}^{p}$ represents a coupling of coverings which vanishes asymptotically as $N\to\infty$, and it reflects the coupling of random partitions obtained in \cite{LM2} for $\U(N)$\footnote{Note that the random surface representation is new: even in the case of $\U(N)$, it does not appear in \cite{LM2}.}. When $G_N$ is special orthogonal or symplectic, the function $\Phi_{t,N}$ does not depend on $p$ . We refer to the complete statement in Theorem~\ref{thm:mainbis} and the discussion thereafter for more explanations. We will also provide an alternate version of the random surface representation (Theorem~\ref{thm:main_alt}) with slightly different conditions, where the remainder is not a power of $\frac1N$ but exponentially small in $N$, and which is not directly related to the asymptotic expansion.

\subsection{Gauge/string duality on a torus}

Theorem~\ref{thm:main} is a particular case of a physical phenomenon known as gauge/string duality, based on the idea that observables of gauge theory can be expressed in terms of string theory. In the context of two-dimensional Yang--Mills theory, this duality has been first studied by Gross and Taylor \cite{Gro,GT,GT2}. Their work is stated in terms of formal matrix integrals whereas ours truly lies at the level of convergent matrix integrals (see e.g. the preface of \cite{Eyn} for a discussion on this distinction and on the challenge of going from the former to the latter). In Section~\ref{sec:string}, we will develop this point of view and explain more carefully how Theorem~\ref{thm:main} is a rigorous probabilistic counterpart of Gross--Taylor's theory. A natural question is whether this duality also holds for higher genus surfaces, where it is already known that the partition function admits a first-order expansion (see \cite{Lem,DL,Lem3}), or also for surfaces with boundaries. It is  stated  in \cite{GT,GT2}, of course at a formal level and we believe it is the case, at least in terms of a ``sum-over-surfaces" point of view. In this direction, we can mention the refined version of the Weingarten calculus developed by Magee \cite{Mag2} and Dahlqvist\footnote{We mention this paper by Antoine Dahlqvist that came up after ours, but which was partly conducted in parallel to ours.} \cite{Dah26} based on the Koike--Schur--Weyl duality. In some sense, our paper controls the duality at the level of \emph{spectral data} whereas the other ones control it at the level of \emph{topological data}, which are two independent tasks. It also explains why, in the present paper, we do not discuss higher genus --- it is not because it is out of reach, but because it would require combining our main theorems with other independent results, and we try to keep this paper condensed and self-contained. The study of Yang--Mills partition functions on general surface will be carried out in a separate work, mainly as a consequence of this paper and \cite{Dah26}. Let us mention another notion of gauge/string duality, which is more combinatorial: for the discrete Yang--Mills measure on $\Z^d$, Wilson loop expectations can be rewritten as signed sums over combinatorial surfaces \cite{Jaf16,Cha19,CPS,BCSK,Lem26a}. Note that some of the results presented in this paper can also be used in this context \cite{Lem26mf}, but we will not develop this aspect here.

In addition to Gross--Taylor's Yang--Mills/Hurwitz duality, we will also describe a Yang--Mills/Gromov--Witten duality: indeed, all coefficients of the expansion~\eqref{eq:asympt_expansion} can be expressed in terms of the generating function of Gromov--Witten invariants $\langle \tau_{k_1}\ldots\tau_{k_n}\rangle_d^\Tbb$ on a torus, defined for arbitrary parameters $q,t_1,\ldots,t_n$ by
\[
\mathcal{Z}_q^\Tbb(t_1,\ldots,t_n)=\sum_{d=1}^\infty q^d\sum_{k_1,\ldots,k_n}\langle\tau_{k_1}\ldots\tau_{k_n}\rangle_d^\Tbb t_1^{k_1+1}\ldots t_n^{k_n+1}.
\]

\begin{theorem}\label{thm:coef_Hurwitz}
Under the assumptions of Theorem \ref{thm:asympt_expansion}, the coefficients of \eqref{eq:asympt_expansion} are explicit functionals of $\mathcal{Z}_{q_t}^\Tbb$, where $q_t=e^{-\frac{t}{2}}$.
\end{theorem}

This relationship does not seem to appear anywhere in the literature, either on the mathematical or on the physical side. We will prove it in Section~\ref{sec:string}, as well as its application to gauge/string duality (Corollary~\ref{cor:gauge_string}) on a two-dimensional torus, where $\Tr(e^{\frac t2\Delta})$ and $\mathcal{Z}_{q_t}^\Tbb(t_1,\ldots,t_n)$ are respectively identified to the Yang--Mills partition function and a correlation function of a topological sigma-model coupled to gravity. While Hurwitz numbers are known to be related to Gromov--Witten invariants through the work of Okounkov--Pandharipande \cite{OP2}, a direct link between the Yang--Mills partition function on a compact surface and the full generating function of Gromov--Witten invariants on the same surface has not been observed previously. We believe that it might uncover more explicitly the integrable structure that lies beneath Gross--Taylor's original approach. This would lead to a more geometric and less combinatorial counterpart of the gauge/string duality.

\subsection{Asymptotic spectral theory beyond Weyl's law}

From a more analytical viewpoint, the central heat trace is the Laplace transform of the counting measure of eigenvalues of the Casimir operator : let
\[
\mu_N=\sum_{\lambda\in\widehat{G}_N}\delta_{c_2(\lambda)},
\]
so that the central heat trace can be rewritten as the Laplace transform of $\mu_N,$
\[
\Tr(e^{\frac{t}{2}\Delta_{G_N}})=\int_0^\infty e^{-\frac{t}{2}E}d\mu_N(E).
\]

 Recall that for a compact Riemannian manifold $(M,g)$ of dimension $n$, if we denote by $0=\lambda_0<\lambda_1\leq\lambda_2\leq\ldots\to\infty$ the eigenvalues of the Laplace--Beltrami operator, the short-time asymptotic expansion
\[
\Tr(e^{t\Delta})=(4\pi t)^{-n/2}(\vol(M)+\text{higher terms})
\]
of the heat kernel, combined with a Tauberian theorem, yields the asymptotic equivalent, known as \emph{Weyl's law}
\[
\mathcal{N}(\Lambda):=\#\{i:\ \lambda_i\leq \Lambda\}\sim\frac{\pi^{n/2}}{(2\pi)^n\Gamma(n/2+1)}\vol(M)\Lambda^{n/2},\quad \Lambda\to\infty,
\]
A natural question is whether similar results hold when counting the Casimir numbers. In our setting, the manifold is not  fixed as we consider $G_N$ as $N\to\infty.$ It is similar to the regime of spectral theory of high genus hyperbolic surfaces for instance, where a uniform version of Weyl's law \cite{Mon22} has been shown recently. In our case, the asymptotic behavior is  different and is summarized in the following theorem.

\begin{theorem}[Asymptotic counting law for large-rank Casimir spectrum, see Theorem \ref{thm:Cardy-detailed}]\label{thm:Cardy}
Let $(G_N)_{N\geq1}$ be a sequence of compact classical groups  with $G_N\subset\GL_N(\C)$. There exists a locally finite measure $\mu$ on $[0,\infty)$ such that, at every
continuity point $\Lambda$ of $\mu$,
\begin{equation}\label{eq:cv_continuitypoint}
\mu_N([0,\Lambda])\longrightarrow \mu([0,\Lambda]).
\end{equation}
Moreover, there exist $a,b,c >0$ such that
\[ \mathcal{N}_\tau(\Lambda):=\mu_\tau([0,\Lambda]) \sim a \Lambda^{-b} e^{c\sqrt \Lambda},\]
where the constants $a,b,c$ depend on the type $\tau$ of the group.
\end{theorem}

A detailed statement of  Theorem~\ref{thm:Cardy}, with explicit expressions of the constants  will be provided and proved in Section~\ref{sec:Cardy}; it is roughly a consequence of the fact that the large-$N$ limit of the trace, obtained by taking the first coefficient of the expansion from Theorem~\ref{thm:asympt_expansion} is expressed in terms of modular forms, combined with a Tauberian theorem due to Ingham \cite{Ing41} and revisited recently by Bringmann--Jennings-Shaffer--Mahlburg \cite{BJSM23}. Although modularity properties of the heat kernel on Lie groups are known to hold since the work of Fegan \cite{Feg},  the explicit spectral counting law for the large-rank Casimir spectrum has not appeared in the literature yet --- to the best of our knowledge. The exponential growth is reminiscent of Hardy--Ramanujan's estimate of the number $p(n)$ of partitions of $n$ (see \cite{Apo} for instance), and to Cardy's celebrated formula in the context of CFT on a torus \cite{Car86}. This result, or at least its proof, may be familiar to number theorists\footnote{In fact, $Z(t)$ is a particular case of eta--theta quotient, which is described for instance in \cite{FGKST16}.}, but is strikingly unusual from a spectral theory perspective. Indeed, it appears to be a new phenomenon in contrast to the polynomial growth of Weyl's law. Although we did not investigate it in this paper, we believe that a double scaling in $N$ and $t$ might also be of interest, eventually in order to find a phase transition between the regime of finite-$N$ that gives Weyl's law, and the regime of large-$N$ that gives this Cardy growth.\\

{\bf Acknowledgements.} TL acknowledges Elba Garcia-Failde for giving him insights on the cut-and-join equation, as well as Nalini Anantharaman and Jonathan Novak for fruitful discussions. MM  acknowledges support from the Labex CEMPI (ANR-11-LABX-0007-01) and the support of the CDP C2EMPI, together with the French State under the France-2030 programme, the University of Lille, the Initiative of Excellence of the University of Lille, the European Metropolis of Lille for their funding and support of the R-CDP-24-004-C2EMPI project. 

\section{Highest weights, partitions and stable Casimir observables}\label{sec:prelim}

The purpose of this section is to introduce the representation-theoretic coordinates in which the large-rank behavior of the central heat trace becomes tractable. We first recall the normalizations of the Laplace--Beltrami operator and the corresponding Casimir eigenvalues for compact classical groups. We then replace highest weights by partition data adapted to large $N$. In these coordinates, the quadratic Casimir admits a stable expansion whose first two terms are governed by shifted symmetric polynomials. This stable Casimir structure is the algebraic mechanism behind the asymptotic expansion proved in Section~\ref{sec:asympt_exp}.

\subsection{Heat kernel on compact classical groups}

The complex classical groups are special linear groups over the spaces of real, complex or quaternionic numbers. Their compact real forms, called \emph{compact classical groups}, correspond to the \emph{special unitary group}
\[
\SU(N)=\{U\in\GL_N(\C): U^*=U^{-1},\det(U)=1\},
\]
the \emph{special orthogonal group}
\[
\SO(N)=\{O\in\GL_N(\R):O^t=O^{-1},\det(O)=1\},
\]
and the \emph{compact symplectic group}
\[
\Sp(N)=\{S\in\GL_{2N}(\C):S^*=S^{-1},S^tJS=J\},
\]
where $J\in\GL_{2N}(\C)$ is defined by
\[
J=\begin{pmatrix}
0 & I_{N}\\
-I_{N} & 0
\end{pmatrix}.
\]
Following the convention of \cite{DL}, we shall also include the unitary group $\U(N)=\{U\in\GL_N(\C):U^*=U^{-1}\}$ to the list of compact classical groups. 

The corresponding Lie algebras are endowed with the following invariant inner products\footnote{A justification of this choice of inner product is given in \cite[\S 2.1.]{DL}.}:
\begin{equation}\label{eq:inner_product}
\langle X,Y\rangle = \frac{\beta k_N}{2}\Tr(X^*Y),\quad \forall X,Y\in\mathfrak{g}_N,
\end{equation}
where $k_N=N$ for $\U(N),\SU(N),\SO(N)$ and $k_N=2N$ for $\Sp(N)$, and $\beta$ is a parameter depending on the type of Lie algebra (sometimes called the \emph{Dyson index} when referring to Gaussian $\beta$-ensembles in random matrix theory):
\[
\begin{array}{|c|c|c|c|}
\hline \mathfrak{g}_N & \mathfrak{so}(N) & \mathfrak{su}(N),\ \mathfrak{u}(N) & \mathfrak{sp}(N)\\
\hline \beta & 1 & 2 & 4\\
\hline
\end{array}
\]
{\bf Notation.} In the rest of this paper, we shall denote respectively by $A',$ $A,$ $B,$ $C,$ $D$ the type\footnote{We borrow this convention from \cite{DL}, and it mimicks the notations of the underlying root systems.} of unitary groups, special unitary groups, odd orthogonal groups, symplectic groups and even orthogonal groups respectively. It will be convenient to specify several quantities related to the whole family $(G_N)_{N\geq 1}$, such as the coefficients of the expansion \eqref{eq:asympt_expansion}.

Let $\{X_1,\ldots,X_n\}$ be an orthonormal basis of $\mathfrak{g}_N$ for the prescribed inner product; the Laplace--Beltrami operator is defined as
\[
\Delta_{G_N}f(g) = \sum_{i=1}^n \frac{d^2}{dt^2}\bigg\vert_{t=0}f(ge^{tX_i}),\quad \forall f\in\mathscr{C}^\infty(G_N) \text{ and } \forall g\in G_N.
\]
When there is no ambiguity, we will drop the subscript and simply denote by $\Delta$ the Laplace--Beltrami operator. The \emph{heat kernel} on $G_N$ (starting from identity) is the unique solution $p:(0,\infty)\times G_N\to\R_+,(t,g)\mapsto p_t(g)$ of the heat equation\footnote{This is the heat equation used for the Brownian motion, which differs by a factor $\frac12$ from another popular definition of the heat kernel; one can switch between both conventions by a simple rescaling of the time parameter.}
\[
\frac{d}{dt}p_t(g)=\frac12\Delta_{G_N} p_t(g), \quad \lim_{t\to 0^+}p_t(g)dg=\delta_{1_{G_N}}.
\]
By standard results of representation theory, the heat kernel admits a Fourier expansion in terms of irreducible representations:
\begin{equation}
p_t(g)=\sum_{\lambda\in\widehat{G}_N}e^{-\frac{t}{2}c_2(\lambda)}d_\lambda\chi_\lambda(g),
\end{equation}
where $d_\lambda$ is the dimension of the irreducible representation associated to $\lambda$, $\chi_\lambda$ is the character of the representation and $c_2(\lambda)\geq 0$ is the Casimir number, that is, the eigenvalue of $-\Delta_{G_N}$ associated to $\chi_\lambda$. As stated in the introduction, we denote by $\Tr(e^{\frac{t}{2}\Delta_{G_N}})$ the trace of the central heat trace~\eqref{eq:hk_trace_ini}. In order to compute this trace, we need to describe $\widehat{G}_N$ for all compact classical groups, as well as $c_2(\lambda)$ for all $\lambda\in\widehat{G}_N$.

\begin{itemize}
\item The set $\widehat{G}_N$ of irreducible representations has the following form depending on the group:
\[
\widehat{\U}(N)=\{(\lambda_1,\ldots,\lambda_N)\in\Z^N:\lambda_1\geq\ldots\geq\lambda_N\},
\]
\[
\widehat{\SU}(N)=\{(\lambda_1,\ldots,\lambda_N)\in\N^N:\lambda_1\geq\ldots\geq\lambda_N=0\},
\]
\[
\widehat{\SO}(2N+1)=\widehat{\Sp}(N)=\{(\lambda_1,\ldots,\lambda_N)\in\N^N:\lambda_1\geq\ldots\geq\lambda_N\},
\]
\[
\widehat{\SO}(2N)=\{(\lambda_1,\ldots,\lambda_N)\in\N^{N-1}\times \Z:\lambda_1\geq\ldots\geq\lambda_{N-1}\geq\vert\lambda_N\vert\}
\]
\item Given $\lambda\in\widehat{G}_N$, if $G_N$ has type $\tau\in\{A',A,B,C,D\}$, the Casimir number $c_2(\lambda)$ admits an explicit expression (see \cite{DL} for instance, where the same inner product is put on $\mathfrak{g}_N$), as displayed in Table~\ref{tab:Casimirs} below:
\renewcommand\arraystretch{1.5}
\begin{table}[h!]
    \centering
    \begin{tabular}{|c|c|c|}
     \hline $\tau$ & $G_N$ & $c_2(\lambda)$\\
\hline $A'$ & $\U(N)$ & $\frac{1}{N}\sum_{i=1}^N\lambda_i(\lambda_i+N+1-2i)$ \\
\hline
$A$ & $\SU(N)$ & $\frac{1}{N}\sum_{i=1}^N\lambda_i(\lambda_i+N+1-2i)-\frac{1}{N^2}\left(\sum_{i=1}^N\lambda_i\right)^2$\\ \hline
$B$ & $\SO(2N+1)$ & $\frac{1}{2N+1}\sum_{i=1}^N\lambda_i(\lambda_i+2N+1-2i)$\\ \hline
$C$ & $\Sp(N)$ & $\frac{1}{2N}\sum_{i=1}^N\lambda_i(\lambda_i+2N+2-2i)$\\ \hline
$D$ & $\SO(2N)$ & $\frac{1}{2N}\sum_{i=1}^N\lambda_i(\lambda_i+2N-2i)$.\\ 
\hline
    \end{tabular}
    \caption{Casimir numbers for all group types.}
    \label{tab:Casimirs}
\end{table}
\end{itemize}

\subsection{Random partitions}\label{sec:random_partitions}

An integer partition is a finite family $\alpha=(\alpha_1\geq\ldots\geq\alpha_r)$ of nonincreasing positive integers. Its size is $\vert\alpha\vert=\sum_i\alpha_i$, its length is $\ell(\alpha)=r$, and its parts are the coefficients $\alpha_i$. By convention, we denote by $\varnothing$ the empty partition, which has zero size and length. If $n\geq 0$ is an integer, we write $\alpha\vdash n$ if $\alpha$ is a partition such that $\vert\alpha\vert=n$. Besides, we denote by $\Pfr_n$ the set of partitions $\alpha\vdash n$, and $\Pfr=\bigsqcup_{n\geq 0}\Pfr_n$ the set of all integer partitions. The generating function of the numbers $p(n)$ of partitions of $n$ is given by
\[
\sum_{n\geq 0}p(n)q^n=\prod_{m\geq 1}(1-q^m)^{-1},
\]
which is absolutely convergent for $\vert q\vert <1$ and corresponds to the inverse of Euler's function $\phi:q\mapsto\prod_{m\geq 1} (1-q^m)$ (see \cite[\S 14.3]{Apo}).

Partitions are in bijection with Young diagrams, which are left-justified tables, such that the $i$-th line of the Young diagram of $\alpha$ contains $\alpha_i$ cells. See Fig.~\ref{fig:YD} for an illustration.
\begin{figure}[h!]
    \centering
    \includegraphics{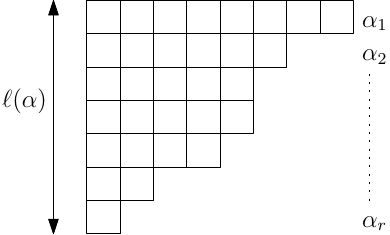}
    \caption{The Young diagram representing the partition $\alpha=(8,6,5,5,4,2,1)$.}
    \label{fig:YD}
\end{figure}
Integer partitions are ubiquitous in mathematics, and random partitions have been used to solve various problems ranging from the distribution of the longest increasing subsequence of a random permutation \cite{BOO} to the asymptotics of volumes of moduli spaces of holomorphic differentials \cite{EO}. Let $\Ufr(q)$ be the $q$-\emph{uniform measure} on $\Pfr$ for $q\in(0,1)$, which is defined by
\begin{equation}
\label{def:qunif}
\Pbb(\alpha)=\phi(q)q^{\vert\alpha\vert},\quad \forall \alpha\in\Pfr.
\end{equation}
Its construction is analogous to the Poissonized Plancherel measure \cite{BOO}, which has been more extensively studied.

Given any partition $\lambda\vdash n$, we assign to any cell $\square\in\lambda$ of its Young diagram a \emph{content} $c(\square)=j(\square)-i(\square)$, which is the difference of its coordinates in the diagram. The \emph{total content} of $\lambda$ is defined by
\[
K(\lambda)=\sum_{\square\in\lambda}c(\square).
\]
It can be expressed explicitly as a function of the length of the rows of the diagram:
\begin{equation}\label{eq:content_explicit}
K(\lambda) = \frac{1}{2} \sum_{i=1}^{\ell(\lambda)} \lambda_i\left(\lambda_i +1-2i\right).
\end{equation}
The typical observables that we will consider in the present paper are polynomials of $K(\lambda)$ and $\vert\lambda\vert$.
It will be useful to compare the total content to the size and length of a given partition.
\begin{proposition}\label{prop:bound_content}
For any $\alpha\in\Pfr$,
\begin{equation}\label{eq:bound_content_1}
\vert K(\alpha)\vert\leq \vert\alpha\vert^2,
\end{equation}
and
\begin{equation}\label{eq:bound_content_2}
-\ell(\alpha)\vert\alpha\vert\leq 2K(\alpha) \leq \alpha_1\vert\alpha\vert.
\end{equation}
\end{proposition}

\begin{proof}
On the one hand, by~\eqref{eq:content_explicit},
\[
2K(\alpha)=\sum_{i=1}^{\ell(\alpha)}\alpha_i(\alpha_i+1-2i) =\alpha_1 \sum_i (\alpha_i+1-2i)\leq \alpha_1\vert\alpha\vert.
\]
On the other hand, the sum of contents of the column $j$ of $\alpha$ is
\[
(j-1)+(j-2)+\ldots+(j-n_j)\geq n_j\left(j-\frac12(n_j+1)\right),
\]
where $n_j\leq \ell(\alpha)$ is the height of the column $j$. It follows that
\[
K(\alpha)= \sum_j n_j\left(j-\frac12-\frac12n_j\right)\geq -\sum_j \frac{n_j}{2}\ell(\alpha)\geq -\frac12\vert\alpha\vert\ell(\alpha),
\]
which proves \eqref{eq:bound_content_2}. Then, \eqref{eq:bound_content_1} can be deduced from the fact that $\ell(\alpha)\leq \vert\alpha\vert$ and $\alpha_1\leq\vert\alpha\vert$.
\end{proof}

The final result of this subsection is that only even powers of the total content contribute to the expectations of the observables. 

\begin{proposition}\label{prop:symmetry_partitions}
Let $\alpha,\beta$ two independent $q$-uniform random partitions and $P\in\R[X,Y]$ be a bivariate polynomial. For any integers $k_1,k_2,k_3,k_4\geq 0$, if $k_1$ or $k_2$ is odd, then
\begin{equation}
\E[K(\alpha)^{k_1}K(\beta)^{k_2}P(\vert\alpha\vert,\vert\beta\vert)]=0.
\end{equation}
\end{proposition}

\begin{proof}
By symmetry of the roles of $\alpha$ and $\beta$, we can assume that $k_1$ is odd. By definition,
\begin{align*}
\E[K(\alpha)^{k_1}K(\beta)^{k_2}&P(\vert\alpha\vert,\vert\beta\vert)]= \phi(q)^2\sum_{n_1,n_2\geq 0}q^{n_1+n_2}P(n_1,n_2)\sum_{\alpha\vdash n_1}K(\alpha)^{k_1}\sum_{\beta\vdash n_2}K(\beta)^{k_2}.
\end{align*}
However, for any $\alpha\vdash n_1$, $K(\alpha)=-K(\alpha')$, where $\alpha'$ is the diagram obtained by transposing the diagram of $\alpha$, so that $\sum_\alpha K(\alpha)^{k_1}=0$ if $k_1$ is odd. The result follows.
\end{proof}

\subsection{Highest weights/partitions correspondence}\label{sec:hw_partition_corresp}

For any compact classical group $G_N$, we want to describe $\widehat{G}_N$ in terms of partitions instead of highest weights, because it is more adapted to the study of the large-$N$ asymptotics. Indeed, the expression of $c_2(\lambda)$ for $\lambda\in\widehat{G}_N$ has a nontrivial dependence on $N$ that becomes fully explicit in this new description, and it is the key of the asymptotic expansion of $\Tr(e^{\frac{t}{2}\Delta})$. As we will see below, the relationship between highest weights of special orthogonal and symplectic groups is quite straightforward, but we shall explain more carefully the procedure for unitary groups. There is a well-established construction \cite{Sta,Koi} that produces a highest weight of $\GL_N(\C)$ from a couple of partitions. Namely, if $\alpha$ and $\beta$ are two arbitrary partitions, for any $N\geq \ell(\alpha)+\ell(\beta)$ one can construct the highest weight
\[
[\alpha,\beta]_N=(\alpha_1,\ldots,\alpha_{\ell(\alpha)},0,\ldots,0,-\beta_{\ell(\beta)},\ldots,-\beta_1).
\]
It corresponds to a rational representation of $\GL_N(\C)$ but also an irreducible representation of $\U(N)$. It has the convenience to be well-defined for any $N$ large enough, and it is suitable for large-$N$ asymptotics, hence it is named \emph{stable representation}. What we shall consider is a slightly modified version of this construction, because the map $(\alpha,\beta)\mapsto [\alpha,\beta]_N$ does not produce all highest weights: we will rather consider a shift of these rational representations by an arbitrary integer $n\in\Z$, and introduce the map $\lambda_N:\Pfr\times\Pfr\times\Z\to\widehat{\U}(N)$ defined by
\[
\lambda_N(\alpha,\beta,n)=(\alpha_1+n,\ldots,\alpha_{\ell(\alpha)}+n,n,\ldots,n,n-\beta_{\ell(\beta)},\ldots,n-\beta_1).
\]
Now this map is surjective for any fixed $N$, but not injective, as explained in \cite[\S 2.2]{Lem}. In order to get injectivity, we have to force some kind of ``discontinuity" between both partitions, by setting $\lambda_k=n$ for a particular $1\leq k\leq N$, and imposing that $\ell(\alpha)\leq k-1$ and $\ell(\beta)\leq N-k$. The choice of this $k$ is not unique, but not completely arbitrary either: as long as $k$ is of order $N/2,$ everything will work fine. We choose $k=\lfloor(N+1)/2\rfloor$, so that the cutoffs for lengths of $\alpha$ and $\beta$ are respectively
\[
A_N=\lfloor(N+1)/2\rfloor-1,\quad B_N=N-\lfloor(N+1)/2\rfloor.
\]

\begin{proposition}\label{prop:dual_sets}
Let $G_N$ be a compact classical group of type $\tau\in\{A',A,B,C,D\}$. Its dual $\widehat{G}_N$ can be described in terms of integer partitions. More precisely,
\[
\widehat{\U}(N)\simeq \Lambda_N^{(1)}:=\{(\alpha,\beta,n)\in\Pfr\times\Pfr\times\Z:\ell(\alpha)\leq A_N,\ell(\beta)\leq B_N\},
\]
\[
\widehat{\SU}(N)\simeq \Lambda_N^{(2)}:= \{(\alpha,\beta)\in\Pfr\times\Pfr:\ell(\alpha)\leq A_N,\ell(\beta)\leq B_N\},
\]
\[
\widehat{\SO}(2N+1)\simeq\widehat{\Sp}(N)\simeq \Lambda_N^{(3)}:=\{\mu\in\Pfr:\ell(\mu)\leq N\},
\]
\[
\widehat{\SO}(2N)\simeq\Lambda_N^{(4)}:=\{(\mu,m,n)\in\Pfr\times\N\times\Z: \ell(\mu)\leq N-2, \vert n\vert\leq m\}.
\]
\end{proposition}

\begin{proof}
The cases of $\U(N)$ and $\SU(N)$ are treated in \cite{Lem,LM2}, and are respectively described by bijections $\lambda_N^{A'}:\Lambda_N^{(1)}\to\widehat{\U}(N)$, $\lambda_N^A:\Lambda_N^{(2)}\to\widehat{\SU}(N)$ given by
\[
\lambda_N^{A'}(\alpha,\beta,n)= (\alpha_1+n,\ldots,\alpha_{\ell(\alpha)}+n,n,\ldots,n,n-\beta_{\ell(\beta)},\ldots,n-\beta_1),
\]
\[
\lambda_N^A(\alpha,\beta)=(\alpha_1+\beta_1,\ldots,\alpha_{\ell(\alpha)}+\beta_1,\beta_1,\ldots,\beta_1,\beta_1-\beta_{\ell(\beta)},\ldots,\beta_1-\beta_2,0).
\]
The parameters $A_N$ and $B_N$ ensure that $\lambda_{A_N+1}=n$ is independent of $\alpha$ and $\beta$, and that $N\geq \ell(\alpha)+\ell(\beta)+1$, so that there is a unique way to recover $\alpha,\beta,n$ in the case of $\U(N)$, or $\alpha,\beta$ in the case of $\SU(N)$. For $\SO(2N+1)$, there is a much more straightforward bijection $\lambda_N^B:\Lambda_N^{(3)}\to\widehat{\SO}(2N+1)$ given by
\[
\lambda_N^B(\mu)=(\mu_1,\ldots,\mu_{\ell(\mu)},0,\ldots,0),
\]
and the same happens for $\Sp(N)$: $\lambda_N^C:\Lambda_N^{(3)}\to\Sp(N)$. Finally, in the case of $\SO(2N)$, the bijection $\lambda_N^D:\Lambda_N^{(4)}\to\SO(2N)$ is defined by
\[
\lambda_N^D(\mu,m,n)=(\mu_1+m,\ldots,\mu_{N-2}+m,m,n).
\]
In the formula above, we have set $\mu_i=0$ for all $i>\ell(\mu)$ by convention. From this, it is clear that we have a bijection, whose inverse is given by $n_\lambda=\lambda_N$, $m_\lambda=\lambda_{N-1}$ and $\mu_i=\lambda_i-m$ for all $i\leq N-1$.
\end{proof}

Let us introduce several intermediary functions that we shall use to describe the Casimir numbers of all group types (cf. Table~\ref{tab:Casimirs}) more conveniently:
\[
F^{A'}(\alpha,\beta,n)=K(\alpha)+K(\beta)+n(\vert\alpha\vert-\vert\beta\vert),\quad \forall(\alpha,\beta,n)\in\Lambda_N^{(1)},
\]
\[
F_1^{A}(\alpha,\beta)=K(\alpha)+K(\beta),\quad F_2^A(\alpha,\beta)=(\vert\alpha\vert-\vert\beta\vert)^2,\quad \forall (\alpha,\beta)\in\Lambda_N^{(2)},
\]
\[
F^B(\mu)=K(\mu)-\frac12\vert\mu\vert,\quad  F^C(\mu)=K(\mu)+\frac12\vert\mu\vert,\quad \forall \mu\in\Lambda_N^{(3)},
\]
\[
F_1^D(\mu)=K(\mu)-\frac12\vert \mu\vert,\quad  F_{2,N}^D(\mu,m,n)=\frac{N-1}{2}m+\frac{N-1}{2N}m^2+\frac{\vert\mu\vert}{N}m+\frac{n^2}{2N},\quad \forall (\mu,m,n)\in\Lambda_N^{(4)}.
\]
\begin{lemma}\label{lem:casimirs}
The Casimir numbers of the compact classical groups can be rewritten in terms of integer partitions and the previously defined functions:
\begin{align*}
c_2(\lambda_N^{A'}(\alpha,\beta,n))&=  \vert\alpha\vert+\vert\beta\vert+n^2+\frac{2}{N}F^{A'}(\alpha,\beta,n),\quad \forall (\alpha,\beta,n)\in\Lambda_N^{(1)},\\
c_2(\lambda_N^A(\alpha,\beta))&=  \vert\alpha\vert+\vert\beta\vert+\frac{2}{N}F_1^A(\alpha,\beta)-\frac{1}{N^2}F_2^A(\alpha,\beta),\quad \forall(\alpha,\beta)\in\Lambda_N^{(2)},\\
c_2(\lambda_N^B(\mu))&=  \vert\mu\vert+\frac{2}{2N+1}F^B(\mu),\quad \forall\mu\in\Lambda_N^{(3)},\\
c_2(\lambda_N^C(\mu))&=\vert\mu\vert+\frac{2}{2N}F^C(\mu),\quad\forall\mu\in\Lambda_N^{(3)},\\
c_2(\lambda_N^D(\mu,m,n))&=\vert\mu\vert+\frac{2}{2N}F_1^D(\mu) +F_{2,N}^D(\mu,m,n),\quad \forall (\mu,m,n)\in\Lambda_N^{(4)}.
\end{align*}

\end{lemma}

\begin{proof}
The proofs in the cases of types $A$ and $A'$ can be found in \cite{Lem}. For $\mu\in\Lambda_N^{(3)}$,
\[
c_2(\lambda_N^B(\mu))=\frac{1}{2N+1}\sum_{i=1}^{\ell(\mu)}\mu_i(\mu_i+2N+1-2i)=\vert\mu\vert + \frac{1}{2N+1}\sum_{i=1}^{\ell(\mu)}\mu_i(\mu_i-2i),
\]
and the result follows from~\eqref{eq:content_explicit}. The computation for $c_2$ is exactly the same, now let us treat the case of $c_2$. Let $(\mu,m,n)\in\Lambda_N^{(4)}$. Recall that
\[
\lambda_N^D(\mu,m,n)=(\mu_1+m,\ldots,\mu_{N-2}+m,m,n),
\]
so that
\[
c_2(\lambda_N^D(\mu,m,n))=\frac{1}{2N}\sum_{i=1}^{N-1}(\mu_i+m)(\mu_i+m+2N-2i) + \frac{1}{2N}n^2.
\]
The result follows from simple identifications.
\end{proof}

The preceding formulas have two consequences. First, they turn the central heat trace into an expectation with respect to explicit measures on partitions. Second, they reveal the stable
large-rank structure of the Casimir. We start with
the probabilistic reformulation of the trace. If we combine Proposition~\ref{prop:dual_sets} and Lemma~\ref{lem:casimirs}, we obtain a new formula for the central heat trace, involving $q_t$-uniform random partitions as defined in \eqref{def:qunif}, for $q_t = e^{-t/2}$. In the unitary case, they also involve the integer Gaussian distribution $\Gfr(q_t)$ on $\Z$, given by $\Pbb(n)=\theta(q_t)e^{-\frac{t}{2}n^2}$ for $n\in\Z$, with $\theta(q)=\sum_{n\in\Z}q^{n^2}$ (see \cite{LM2}).

\begin{corollary}\label{cor:heat_trace_expectation}
For all $N$ and all classical groups $G_N$, the central heat trace on $G_N$ can be rewritten as an expectation over random partitions.
\begin{enumerate}
\item If $G_N=\U(N)$ or $\SU(N)$, let $(\alpha,\beta,n)$ be independent random variables such that $\alpha,\beta\sim\Ufr(q_t)$ and $n\sim\Gfr(q_t)$. We have
\begin{align}
\begin{split}
\Tr(e^{\frac{t}{2}\Delta_{\U(N)}}) & =\frac{\theta(q_t)}{\phi(q_t)^2}\E\left[e^{-\frac{t}{N}F^{A'}(\alpha,\beta,n)}\mathbf{1}_{\Lambda_N^{(1)}}(\alpha,\beta,n)\right]\\
\Tr(e^{\frac{t}{2}\Delta_{\SU(N)}}) & = \frac{1}{\phi(q_t)^2}\E\left[e^{-\frac{t}{N}F_1^A(\alpha,\beta)+\frac{t}{2N^2}F_2^A(\alpha,\beta)}\mathbf{1}_{\Lambda_N^{(2)}}(\alpha,\beta)\right].
\end{split}
\end{align}
\item If $G_N=\SO(2N+1),$ $\Sp(N)$ or $\SO(2N)$, let $\mu$ be a single random partition such that $\mu\sim\Ufr(q_t)$. We have
\begin{align}
\begin{split}
\Tr\left(e^{\frac{t}{2}\Delta_{\SO(2N+1)}}\right) & = \frac{1}{\phi(q_t)}\E\left[e^{-\frac{t}{2N+1}F^B(\mu)}\mathbf{1}_{\{\ell(\mu)\leq N\}}\right],\\
\Tr\left(e^{\frac{t}{2}\Delta_{\Sp(N)}}\right) & = \frac{1}{\phi(q_t)}\E\left[e^{-\frac{t}{2N}F^C(\mu)}\mathbf{1}_{\{\ell(\mu)\leq N\}}\right],\\
\Tr(e^{\frac{t}{2}\Delta_{\SO(2N)}}) & =\frac{1}{\phi(q_t)}\E\left[e^{-\frac{t}{2N}F_1^D(\mu)}\sum_{m=0}^\infty\sum_{n=-m}^m e^{-\frac{t}{2}F_{2,N}^D(\mu,m,n)}\mathbf{1}_{\{\ell(\mu)\leq N-2\}}\right].
\end{split}
\end{align}
\end{enumerate}
\end{corollary}

In order to perform an asymptotic expansion of the formulas from Corollary~\ref{cor:heat_trace_expectation}, which will be essentially the mechanism behind the proof of Theorem~\ref{thm:asympt_expansion}, we need two important lemmas: a suitable lower bound on Casimir numbers (or more precisely on certain terms involved in the expression of the Casimir numbers), and deviation inequalities for $q_t$-uniform random partitions.

\begin{lemma}\label{lem:lb-casimirs}
Let $G_N$ be a compact classical group of type $\tau\in\{A',A,B,C,D\}$. We have the lower bounds
\[
c_2(\lambda_N^{A'}(\alpha,\beta,n)) \geq \frac12(\vert\alpha\vert+\vert\beta\vert)+\left(n+\frac{\vert\alpha\vert-\vert\beta\vert}{N}\right)^2,\quad \forall (\alpha,\beta,n)\in\Lambda_N^{(1)},
\]
\[
\frac{2}{N}F_1^A(\alpha,\beta)-\frac{1}{N^2}F_2^A(\alpha,\beta)\geq -\frac12(\vert\alpha\vert+\vert\beta\vert),\quad  \forall (\alpha,\beta)\in\Lambda_N^{(2)},
\]
\[
\frac{2}{2N+1}F^B(\mu)\geq -\frac{N+1}{2N+1}\vert\mu\vert \geq -\frac{2}{3}\vert\mu\vert ,\quad \forall \mu\in\Lambda_N^{(3)},
\]
\[
\frac{1}{N}F^C(\mu)\geq-\frac12\vert\mu\vert,\quad \forall\mu\in\Lambda_N^{(3)},
\]
\[
\frac{1}{N}F_1^D(\mu)\geq-\frac12\vert\mu\vert,\quad \forall \mu\in\Pfr\ \text{s.t.}\ \ell(\mu)\leq N-2.
\]
\end{lemma}

\begin{proof}
The case of $\U(N)$ is done in \cite[Lemma 4.1]{LM2}. Let us treat the case of $\SU(N)$ (the result is stated in \cite[Lemma 2.9]{Lem} but the proof is incomplete due to an error of sign, hence we provide the right proof here). Let $\alpha$ be a partition of length $\ell(\alpha)\leq A_N\leq N/2$. A quick computation shows that
\[
2K(\alpha)=\sum_{i=1}^{\ell(\alpha)}\alpha_i(\alpha_i+1-2i)=\sum_{i=1}^{\ell(\alpha)}\alpha_i^2+\sum_{1\leq i<j\leq \ell(\alpha)}(\alpha_i-\alpha_j)-\ell(\alpha)\vert\alpha\vert.
\]
However, $\alpha_i\geq \alpha_j$ for all $i<j$ and
\[
\sum_{i=1}^{\ell(\alpha)}\alpha_i^2-\frac1N\vert\alpha\vert^2 \geq \sum_{i=1}^{\ell(\alpha)}\alpha_i^2-\frac{1}{\ell(\alpha)}\vert\alpha\vert^2,
\]
because $\ell(\alpha)\leq N$, and the RHS is nonnegative by Cauchy--Schwarz inequality. Hence,
\[
2K(\alpha)-\frac1N\vert\alpha\vert^2 \geq -\ell(\alpha)\vert\alpha\vert\geq-\frac{N}{2}\vert\alpha\vert.
\]
The same holds for $\beta$ such that $\ell(\beta)\leq B_N\leq N/2$, therefore
\begin{align*}
\frac2N F_1^A(\alpha,\beta)-\frac{1}{N^2}F_2^A(\alpha,\beta) = & \frac1N\left(2K(\alpha)-\frac1N\vert\alpha\vert^2 + 2K(\beta)-\frac1N\vert\beta\vert^2+\frac{2}{N}\vert\alpha\vert\vert\beta\vert\right)\\
\geq & -\frac12(\vert\alpha\vert+\vert\beta\vert).
\end{align*}
All other cases follow from the estimate
\begin{equation}\label{eq:content-size}
2K(\mu)-\vert\mu\vert\geq -(\ell(\mu)+1)\vert\mu\vert,
\end{equation}
which is a consequence of \eqref{eq:bound_content_2}.
\end{proof}

\begin{lemma}\label{lem:dev_ineq_bis}
For any integers $n\in\Z$ and $p\geq 1$ and any random $q_t$-uniform $\alpha,\beta$ with $t>0$, we have as $N\to\infty$
\[
\Pbb((\alpha,\beta,n)\notin\Lambda_N^{(1)})=O_t(N^{-p}).
\]
\[
\Pbb((\alpha,\beta)\notin\Lambda_N^{(2)})=O_t(N^{-p}).
\]
\[
\Pbb(\alpha\notin\Lambda_N^{(i)})=O_t(N^{-p}),\quad i\in\{3,4\}.
\]
Furthermore, for any $\gamma\in(0,\infty)$,
\[
\Pbb(\vert\alpha\vert>N^\gamma)=O_t(e^{-\frac{t}{4}N^\gamma}).
\]
\end{lemma}

\begin{proof}
All cases follow easily from \cite[Proposition 2.5]{LM2}. For instance,
\[
\Pbb((\alpha,\beta,n)\notin\Lambda_N^{(1)})\leq \Pbb(\ell(\alpha)>A_N)+\Pbb(\ell(\beta)>B_N)\leq 2 C_1(p,q)N^{-p},
\]
by checking that $A_N$ and $B_N$ are greater or equal to $\lfloor (N-1)/2\rfloor$.
\end{proof}

\subsection{Stable realization of the Casimir}

Let us explain how the Casimir number $c_2(\lambda)$ becomes a stable observable in the algebra of shifted symmetric functions. Recall first the operator which underlies the terminology. Let $G_N$ be one of the compact classical groups considered above, and let $\mathfrak g_N$ be its real Lie algebra, endowed with the invariant inner product fixed in~\eqref{eq:inner_product}. If $(X_i)_i$ is an orthonormal basis of $\mathfrak g_N$, then the Laplace--Beltrami operator is the left-invariant differential operator associated with
\[
        \sum_i X_i^2 .
\]
Equivalently, after complexifying the Lie algebra, this element belongs to the center of the universal enveloping algebra $Z\bigl(\mathcal U(\mathfrak g_{N,\C})\bigr),$ where $\mathfrak g_{N,\C}=\mathfrak g_N\otimes_{\mathbb R}\C$ is the complexification of $\mathfrak{g}_N$. Here $\mathcal U(\mathfrak g_{N,\C})$ denotes the associative algebra generated by $\mathfrak g_{N,\C}$, with relations
\[
        XY-YX=[X,Y],\qquad X,Y\in\mathfrak g_{N,\C}.
\]
Through the standard identification between bi-invariant differential operators on $G_N$ and $Z(\mathcal U(\mathfrak g_{N,\C}))$ \cite[Chap.~10]{FegBook}, the negative Laplacian $-\Delta_{G_N}$ is represented by the positive quadratic Casimir.

If $(\rho_\lambda,V_\lambda)$ is an irreducible representation of $G_N$, Schur's lemma implies that every element of $Z(\mathcal U(\mathfrak g_{N,\C}))$ acts on $V_\lambda$ by a scalar. In particular, the quadratic Casimir acts by the scalar $c_2(\lambda)$.

The Harish--Chandra isomorphism says that the eigenvalues of central elements are polynomial functions of shifted highest weights. In type $A$, after the usual identification of highest weights with partitions, this leads to the algebra of shifted symmetric functions. We use the normalization of Okounkov--Olshanski \cite{OO97,OO98} and Okounkov--Pandharipande \cite[Section~0.4.3]{OP2}. For a partition $\lambda$, define, for $k\geq 1$,

\begin{equation}\label{eq:pk}
\mathbf{p}_k:\lambda\in\Pfr\mapsto\sum_{i=1}^\infty\left((\lambda_i+\frac12-i)^k-(\frac12-i)^k\right)+(1-2^{-k})\zeta(-k).
\end{equation}
The functions $\mathbf{p}_k$ are zeta-regularizations\footnote{An enlightening explanation of the necessity of this regularization can be found in \cite[\S 0.4.3]{OP2}.} of the power sums
\[
p_k(\lambda)=\sum_{i=1}^{\ell(\lambda)}\left(\lambda_i+\frac12-i\right)^k,
\]
and they can be described by means of exponential generating functions
\begin{equation}\label{eq:exp_pk}
\mathbf{e}(\lambda,z)=\sum_{i=1}^\infty e^{z(\lambda_i+\frac12-i)} =\sum_{k=0}^\infty\frac{1}{k!}\mathbf{p}_k(\lambda)z^k.
\end{equation}
For any $\lambda\in\Pfr$, we have
\begin{equation}\label{eq:p_1p_2}
\mathbf{p}_1(\lambda)=\vert\lambda\vert,\quad \mathbf{p}_2(\lambda)=\sum_{i=1}^{\ell(\lambda)}\lambda_i(\lambda_i+1-2i)=2K(\lambda).
\end{equation}

Let us first recall the constructions:
\[
\lambda_N^{A'}(\alpha,\beta,n)=[\alpha,\beta]_N+(n,\ldots,n),\quad \forall (\alpha,\beta,n)\in\Pfr\times\Pfr\times\Z: \ell(\alpha)\leq A_N,\ell(\beta)\leq B_N.
\]
\[
\lambda_N^A(\alpha,\beta)=[\alpha,\beta]_N+(\beta_1,\ldots,\beta_1),\quad \forall (\alpha,\beta)\in\Pfr\times\Pfr,\ \ell(\alpha)\leq A_N,\ \ell(\beta)\leq B_N.
\]
\[
\lambda_N^B(\mu)=\lambda_N^C(\mu)=(\mu_1,\ldots,\mu_{\ell(\mu)},0,\ldots,0),\quad \forall \mu\in\Pfr:\ell(\mu)\leq N.
\]
\[
\lambda_N^D(\mu,m,n)=(\mu_1+m,\ldots,\mu_{N-2}+m,m,n),\quad \forall (\mu,m,n)\in\Pfr\times\N\times\Z, \ell(\mu)\leq N-2,\vert n\vert\leq m.
\]

We then define the commutative algebra $\mathcal{A}^\tau$ in terms of the algebra $\Lambda^*=\C[\mathbf{p}_1,\mathbf{p}_2,\ldots]$ of shifted symmetric functions \cite{OO97,OO98}, generated by the shifted symmetric power sums defined in~\eqref{eq:pk}:
\begin{itemize}
\item In type $A'$, one has $\mathcal A^{A'}=\Lambda^*\otimes \Lambda^*\otimes \C[n]$;
\item In type $A$, one has $\mathcal A^{A}=\Lambda^*\otimes \Lambda^*$;
\item In types $B,C$, and in type $D$ restricted to the sector $m=n=0$, one has $\mathcal A^\tau=\Lambda^*$.\\
\end{itemize}

The quadratic Casimir operator on $G_N$ then defines an element $\kappa_N^\tau\in \mathcal A^\tau,$ in the sense that $\kappa_N^\tau(\lambda_N^\tau) = c_2(\lambda_N^\tau).$ Then Lemma \ref{lem:casimirs} can be rewritten as follows~:

\begin{proposition}\label{prop:shifted-casimir}
For any type $\tau\in\{A',A,B,C,D\}$, the operator $\kappa_N^\tau$ admits an expansion
\begin{equation}
\label{eq:kappa-expansion}
\kappa_N^\tau
=
E^\tau+\frac{1}{N}L^\tau+O(N^{-2}),
\qquad N\to\infty,
\end{equation}
where $E^\tau,L^\tau\in \mathcal A^\tau$ are explicit polynomials in the shifted power sums $\mathbf{p}_1,\mathbf{p}_2$.

More precisely

\begin{itemize}
\item In type $A'$,  the expansion is exact, in the sense that $\kappa_N^{A'}=E^{A'}+\frac{1}{N}L^{A'},$  with
\[
E^{A'}(\alpha,\beta,n)=\mathbf{p}_1(\alpha)+\mathbf{p}_1(\beta)+n^2, \quad L^{A'}(\alpha,\beta,n)
=
\mathbf{p}_2(\alpha)+\mathbf{p}_2(\beta)
+2n\bigl(\mathbf{p}_1(\alpha)-\mathbf{p}_1(\beta)\bigr),
\]
\item In type $A$, the expansion is again exact, with
\[
E^{A}(\alpha,\beta)=\mathbf{p}_1(\alpha)+\mathbf{p}_1(\beta),
\qquad
L^{A}(\alpha,\beta)=\mathbf{p}_2(\alpha)+\mathbf{p}_2(\beta),
\]
and there exists an explicit polynomial $Q^A$ in $\mathbf{p}_1$ such that
\[
\kappa_N^{A}
=
E^{A}+\frac{1}{N}L^{A}+\frac{1}{N^2}Q^{A}.
\]

\item In types $B$ and $C$,
\[
E^\tau(\mu)=\mathbf{p}_1(\mu),\qquad
L^{B}(\mu)=\mathbf{p}_2(\mu)-\mathbf{p}_1(\mu),\qquad
L^{C}(\mu)=\mathbf{p}_2(\mu)+\mathbf{p}_1(\mu).
\]
\item In type $D$, the expansion \eqref{eq:kappa-expansion} holds on the sector $m=n=0$; outside this sector, the Casimir contains a branch of order $mN$ which does not admit a stable limit, although it does not contribute in the end, as we have seen in the proof of the expansion.\\
\end{itemize}
\end{proposition}

The leading term $E^\tau$ can be interpreted as an energy observable, while the first correction $L^\tau$ is governed by $p_2=2K$, up to the type-dependent linear terms displayed above. As we will see later, this is the algebraic bridge between the spectral problem studied here and the Hurwitz-theoretic expansions developed below. One can also reformulate Corollary  \ref{cor:heat_trace_expectation} under the following form~:

\begin{proposition}\label{prop:heat-shifted}
Under the assumptions of Theorem~\ref{thm:asympt_expansion_detailed} we have
\[
\Tr\left(e^{\frac t2\Delta_{G_N}}\right)=\mathbb E\left[\exp\left(-\frac t2 E^\tau-\frac{t}{2N}L^\tau+O(N^{-2})\right)\mathbf 1_{\Lambda_N}\right].
\]
In particular, the large-$N$ asymptotic expansion of the central heat trace is generated by the
stable observable $L^\tau$, while the leading term is controlled by the energy observable $E^\tau$.
\end{proposition}

\section{Asymptotic expansion of the heat trace}\label{sec:asympt_exp}

We now turn from the algebraic structure of the Casimir to the analytic expansion of the central heat trace. The starting point is the expectation formula of Corollary~\ref{cor:heat_trace_expectation}. The proof consists in expanding the exponential perturbation of the leading energy term, while controlling the remainders uniformly by the lower bounds and tail estimates of Section~\ref{sec:prelim}. Let us start with a more detailed version of Theorem~\ref{thm:asympt_expansion}, which provides a first expression of the coefficients $a_k^\tau (t).$ Note that in the following statement, and other asymptotic statements in the sequel, when we consider $(G_N)_{N \geq 1}$ be a sequence of compact classical groups of type $\tau$ such that $G_N\subset\GL_N(\C),$ this makes sense for type $B$ only if $N$ is odd, so we can write $N=2k+1$ and $G_N=\SO(2k+1)$, and for types $C$ and $D$ only if $N$ is even, so we can write $N=2k$ and $G_N=\SO(2k)$ or $G_N=\Sp(k)=\Sp(N/2)$. The important thing to keep in mind is that the coefficients of the expansion depend on the type of $G_N$ but not on $N$.

\begin{theorem}
 \label{thm:asympt_expansion_detailed}
 Let $(G_N)_{N \geq 1}$ be a sequence of compact classical groups of type $\tau \in \{A, A', B, C, D\},$ such that $G_N\subset\GL_N(\C)$. Then $\Tr(e^{\frac{t}{2}\Delta_{G_N}})$ admits the following asymptotic expansion as $N$ tends to infinity~: let $(\alpha, \beta,n)$ be independent random variables such that $\alpha, \beta \sim \Ufr(q_t)$ and $n \sim  \Gfr(q_t),$ we have
 \begin{itemize}
  \item if $G_N = \U(N)$ or $\SU(N),$ ($\tau \in \{A', A\}$ respectively)
  \[ \Tr(e^{\frac{t}{2}\Delta_{G_N}}) = \sum_{k=0}^p \frac{a_{2k}^\tau(t)}{N^{2k}} + O_t(N^{-2p-2}), \]
  with, in the case of $\U(N),$
  \[ a_{2k}^{A'}(t) = \frac{\theta(q_t) t^{2k}}{\phi(q_t)^2 (2k)!} \mathbb E\left[\left(K(\alpha) + K(\beta) +n (|\alpha|-|\beta|)\right)^{2k}\right], \]
  and, in the case of $\SU(N),$
   \[ a_{2k}^A(t) = \frac{1}{\phi(q_t)^2} \sum_{k_1+k_2+k_3 = k} t^{2k-k_3}\frac{\mathbb E\left[K(\alpha)^{2k_1} K(\beta)^{2k_2}(|\alpha|-|\beta|)^{2k_3}\right]}{2^{k_3}(2k_1)!(2k_2)!k_3!}. \]
 \item if $G_N = \SO(N)$ ($\tau \in\{B,D\}$) or $G_N=\Sp(N/2)$ ($\tau=C$),
\[ \Tr(e^{\frac{t}{2}\Delta_{G_N}}) = \sum_{k=0}^p \frac{a_{k}^\tau(t)}{N^{k}} + O_t(N^{-p-1}), \]
with, in the case of $\SO(N)$,
\[  a_{k}^{B}(t) = a_k^D(t)= \frac{(-t)^{k}}{\phi(q_t) k!} \mathbb E\left[\left(K(\alpha) - \frac 1 2 |\alpha|\right)^{k}\right],\]
and, in the case of $\Sp(N/2)$,
\[  a_{k}^{C}(t) = \frac{(-t)^{k}}{\phi(q_t) k!} \mathbb E\left[\left(K(\alpha) + \frac 1 2 |\alpha|\right)^{k}\right].\]
\end{itemize}
\end{theorem}

The proof of Theorem \ref{thm:asympt_expansion_detailed} will be adapted from the case of $\U(N)$ given in \cite{LM2}, where the case of  $\U(N)$ was treated in detail, resulting with the expression of the coefficients $a_{2k}^{A'}(t)$ as above.

\begin{proof}[Proof of Theorem~\ref{thm:asympt_expansion_detailed} for $\SU(N)$]\mbox{}\\

\emph{Step 1.} From Corollary~\ref{cor:heat_trace_expectation}, one can write  $\Tr(e^{\frac{t}{2}\Delta_{\SU(N)}})$ as an expectation of an exponential function of $\frac t N$ under the $q_t$-uniform measure:
\[
\Tr(e^{\frac{t}{2}\Delta_{\SU(N)}})=\frac{1}{\phi(q_t)^2}\E[e^{- \frac t N F_{N}(\alpha,\beta)}\mathbf{1}_{\Lambda_N^{(2)}}(\alpha,\beta)], \text{ with } F_{N}(\alpha,\beta)=F_1^A(\alpha,\beta)-\frac{1}{2N}F_2^A(\alpha,\beta).
\]

\emph{Step 2.} We perform a Taylor expansion of the exponential term for fixed $\alpha$ and $\beta$  and show that the remainder is negligible.
The key point will be to use the lower bound on $F_N$ obtained  Lemma \ref{lem:lb-casimirs}. In this case,
\[
e^{- \frac t N F_{N}(\alpha,\beta)}=\sum_{k=0}^{2p+1} \frac{(-t)^k F_{N}(\alpha,\beta)^k}{k!N^k}+ \frac{F_N(\alpha, \beta)^{2p+2}}{N^{2p+2}(2p+1)!} \int_0^t (t-s)^{2p+1} e^{-\frac{s}{N} F_{N}(\alpha,\beta)}ds.
\]
We have to show that $\forall s \leq t,$
\[
I_{N,p}(s):=\mathbb E\left(F_{N}(\alpha,\beta)^{2p+2} e^{-\frac{s}{N} F_{N}(\alpha,\beta)}\mathbf{1}_{\Lambda_N^{(2)}}(\alpha,\beta)\right) =O_t(1),
\]
in the sense that it is bounded by a constant depending on $t$ and $p$ but not depending on $N.$ We have the following upper bounds:
\begin{eqnarray*}
 I_{N,p}(s) & \leq & \mathbb E\left(F_{N}(\alpha,\beta)^{2p+2}\left( 1 +  e^{-\frac{t}{N} F_{N}(\alpha,\beta)}\right)\mathbf{1}_{\Lambda_N^{(2)}}(\alpha,\beta)\right) \\
 & \leq& \mathbb E\left(F_{N}(\alpha,\beta)^{2p+2}\right) +  \mathbb E\left(F_{N}(\alpha,\beta)^{6p+6}\right)^{1/3} \mathbb E \left(  e^{-\frac{3t}{2N} F_{N}(\alpha,\beta)}  \mathbf{1}_{\Lambda_N^{(2)}}(\alpha,\beta) \right)^{2/3}.
\end{eqnarray*}

From Proposition \ref{prop:bound_content}, we have that
\[
 |F_{N}(\alpha,\beta)|  \leq |K(\alpha)| + |K(\beta)| + |\alpha|^2 + |\beta|^2   \leq 2 |\alpha|^2 + 2 |\beta|^2
\]
and it is easy to check that, for any $t >0$ and any $k\in\mathbb N,$ $\sum_{\alpha \in \Pfr} |\alpha|^{k} q_t^{|\alpha|} <\infty$. This leads to a control, depending only on $p,$ on the two quantities $\mathbb E\left(F_{N}(\alpha,\beta)^{2p+2}\right)$ and $\mathbb E\left(F_{N}(\alpha,\beta)^{6p+6}\right)$. Using the lower bound from Lemma \ref{lem:lb-casimirs} which is valid for any $(\alpha, \beta) \in \Lambda_N^{(2)},$ we get
\[   \mathbb E \left(  e^{-\frac{3t}{2N} F_{N}(\alpha,\beta)}  \mathbf{1}_{\Lambda_N^{(2)}}(\alpha,\beta) \right) \leq \mathbb E \left(  e^{\frac{3t}{8}(|\alpha|+|\beta|)}   \right) = \frac{1}{\phi(q_t)^2}\left(\sum_{\alpha \in \Pfr} (q_{t/4})^{|\alpha|}\right)^2.\]
The latter is finite and independent of $N.$

\emph{Step 3.} We write $\mathbf{1}_{\Lambda_N^{(2)}} = 1 - \mathbf{1}_{(\Lambda_N^{(2)})^c}$ in each term of the sum and then use deviation inequalities to discard the terms involving $\mathbf{1}_{(\Lambda_N^{(2)})^c}.$

For each $k \leq2p+1,$  we have
\[
\left\vert\E[F_{N}(\alpha,\beta)^k(1-\mathbf{1}_{\Lambda_N^{(2)}}(\alpha,\beta))]\right\vert\leq \E[F_{N}(\alpha, \beta)^{2k}]^{1/2}   \Pbb((\alpha,\beta)\notin\Lambda_N^{(2)}).
\]
By Lemma~\ref{lem:dev_ineq_bis}, we therefore have that
\[
\E[F_{N}(\alpha,\beta)^k\mathbf{1}_{\Lambda_N^{(2)}}(\alpha,\beta)]=\E[F_{N}(\alpha,\beta)^k]+O_t(N^{-2p-2}).
\]
Altogether,
\[
\Tr(e^{\frac{t}{2}\Delta_{\SU(N)}})=\sum_{k=0}^{2p+1}\frac{ (-t)^k  }{N^k k!\phi(q_t)^2}  \E[F_{N}(\alpha,\beta)^k]+O_t(N^{-2p-2}).
\]
Let us expand the powers of $F_{N}^k$ in order to extract the right powers of $N$ and check that, for any $k\geq 0,$ $\E[F_{N}(\alpha,\beta)^{2k+1}] =0.$ We have
\begin{align*}
\sum_{k=0}^{2p+1}\frac{(-t)^k\E[F_{N}(\alpha,\beta)^k]}{k!N^k}= &  \sum_{k_1+ 2k_2\leq 4p+2} \frac{1}{k_1!k_2!}\frac{(-1)^{k_1}t^{k_1+k_2}}{2^{k_2}N^{k_1+2k_2}}\E[(K(\alpha)+K(\beta))^{k_1}(\vert\alpha\vert-\vert\beta\vert)^{2k_2}]\\
= & \sum_{k_1+2k_2\leq 2p+1}\frac{1}{k_1!k_2!}\frac{(-1)^{k_1}t^{k_1+k_2}}{2^{k_2}N^{k_1+2k_2}}\\
&\times \E[(K(\alpha)+K(\beta))^{k_1}(\vert\alpha\vert-\vert\beta\vert)^{2k_2}] + O_t(N^{-2p-2})\\
= & \sum_{\ell_1+\ell_2+2k_2\leq 2p+1}\frac{(-1)^{\ell_1+\ell_2}t^{\ell_1+\ell_2+k_2}}{\ell_1!\ell_2!k_2!2^{k_2}N^{\ell_1+\ell_2+2k_2}}\\
&\times \E[K(\alpha)^{\ell_1}K(\beta)^{\ell_2}(\vert\alpha\vert-\vert\beta\vert)^{2k_2}]+O_t(N^{-2p-2}).
\end{align*}
In the second equality, we used the fact that the terms such that $2p+2 \leq k_1+2k_2 \leq4p+2$ are absorbed in the $O_t(N^{-2p-2})$, and in the last equality we used the binomial formula. Finally, by Proposition~\ref{prop:symmetry_partitions}, the terms vanish unless both $\ell_1$ and $\ell_2$ are even, and we obtain the statement of the theorem.
\end{proof}

The proofs for $G_N=\SO(2N+1)$ or $G_N=\Sp(N)$ go along the same steps, although the computations are a bit different, and we omit them for the sake of conciseness. However, we will detail the case of $\SO(2N)$, which is a bit more involved and requires an additional step.

\begin{proof}[Proof of Theorem~\ref{thm:asympt_expansion_detailed} for $\SO(2N)$]
\emph{Step 1.} By Corollary~\ref{cor:heat_trace_expectation},
\[
\Tr(e^{\frac{t}{2}\Delta_{\SO(2N)}})=   \frac{1}{\phi(q_t)}\E\left[e^{-\frac{t}{2N}F_1^D(\mu)} L_N(\mu)\mathbf{1}_{\{\ell(\mu)\leq N-2\}}\right],
\]
with
\begin{equation}\label{def:LN}
F_1^D(\mu) = K(\mu)-\frac{1}{2}\vert\mu\vert,\ \text{and}\ L_N(\mu) = \sum_{m=0}^\infty\sum_{n=-m}^m q_t^{\frac{N-1}{2} m + \frac{m^2}{2} + \frac{m|\mu|}{N} -\frac{m^2}{2N}+ \frac{n^2}{2N}}.
\end{equation}

\emph{Step 1 bis.} As announced above, we need a preliminary step, which is to prove that $L_N$ does not contribute to the asymptotic expansion. The contribution of the term $m=0$ is 1, so that $L_N(\mu)\geq 1$ for all $\mu,$ and
\[
0 \leq L_N(\mu) -1 \leq \sum_{m=1}^\infty \sum_{n=-m}^m q_t^{\frac{N-1}{2}m} \leq\sum_{m=1}^\infty(2m+1)q_t^{\frac{N-1}{2}m}.
\]
In particular, as $\frac{N-1}{2}m\geq\frac{N-1}{4}+\frac{N-1}{4}m\geq\frac{N-1}{4}+\frac14m$, for all $m\geq 1$ and $N\geq 2$, we have
\[
0\leq L_N(\mu)-1\leq q_t^{\frac{N-1}{4}}\sum_{m=0}^\infty(2m+1)q_t^{\frac{1}{4}m}=q_t^{\frac{N-1}{4}}\frac{1+q_t^{1/4}}{(1-q_t^{1/4})^2}=O_t(N^{-p-1}),
\]
uniformly in $\mu$. We finally deduce that
\[
\Tr(e^{\frac{t}{2}\Delta_{\SO(2N)}})=  \frac{1}{\phi(q_t)}\E\left[ e^{-\frac{t}{2N}(K(\mu)-\frac{1}{2}\vert\mu\vert)} \mathbf{1}_{\{\ell(\mu)\leq N-2\}}\right] + O_t(N^{-p-1}).
\]
From there, \emph{Step 2.} and \emph{Step 3.} can be proved by identical arguments as the case of $\SU(N)$.
\end{proof}

\begin{remark}\label{rmk:extension_Levy}
If one wishes to replace the heat kernel by the distribution of a more general L\'evy process, here is what can be done: consider a symmetric infinitely divisible probability measure $\rho$ on $\R$, with L\'evy--Khinchine formula
\[
\int_\R e^{i\xi x}d\rho(x)=e^{-\eta(\xi)},\quad \eta(\xi)=\frac12\sigma^2 \xi^2+\int_{\R^*}(1-\cos(\xi y))d\nu(y),
\]
where $\nu$ is the associated L\'evy measure. Recall that for any measure $\mu$ on a compact group $G$, its Fourier transform $\widehat{\mu}:\widehat{G}\to\bigoplus_{\lambda\in\widehat{G}}\mathrm{End}(V_\lambda)$ is defined by taking the integral of any representation $(\pi_\lambda,V_\lambda)$ associated to $\lambda\in\widehat{G}$:
\[
\widehat{\mu}(\lambda)=\int_G\pi_\lambda(g)d\mu(g)\in\mathrm{End}(V_\lambda).
\]
According to \cite{App}, a central probability measure $\mu$ on a compact group $G$ is induced by $\rho$ if
\[
\widehat{\mu}(\lambda)=e^{-\eta(\sqrt{c_2(\lambda)})}I_\lambda,\quad \forall \lambda\in\widehat{G}.
\]
The case of $p_t$ corresponds to $\eta_t(\xi)=\frac{t}{2}\xi^2$. The central trace becomes
\[
\Tr(\mu)=\sum_{\lambda\in\widehat{G}}e^{-\eta(\sqrt{c_2(\lambda)})}.
\]
If $\eta$ is regular enough, then combining its Taylor expansion with the proof of Theorem~\ref{thm:asympt_expansion} yields the asymptotic expansion of $\Tr(\mu)$.
\end{remark}

As stated in the introduction, Theorem~\ref{thm:asympt_expansion} enables us to deduce the following, simply by taking $p=1$ and $N\to\infty$.

\begin{corollary}[\cite{DL}, Theorem 3.1, (1)]\label{cor:lim_pf}
For any $t>0$ and any sequence $(G_N)$ of compact classical groups of the same type, the trace of the heat kernel on $G_N$ converges to a finite value given by Table~\ref{tab:limits} below.
\end{corollary}

\begin{table}[h!]
    \centering
    \begin{tabular}{|c|c|c|c|}
    \hline $\tau$ & $A'$ & $A$ & $B, C, D$\\
     \hline $G_N$ & $\U(N)$ & $\SU(N)$ & $\SO(N), \Sp(N)$ \\
\hline $\lim_N \Tr(e^{\frac{t}{2}\Delta_{G_N}})$ & $\frac{\theta(q_t)}{\phi(q_t)^2}$ & $\frac{1}{\phi(q_t)^2}$ & $\frac{1}{\phi(q_t)}$ \\
\hline
    \end{tabular}
    \caption{Limits of central heat traces of compact classical groups.}
    \label{tab:limits}
\end{table}

\section{Random surface representation}

The coefficients obtained in Theorem~\ref{thm:asympt_expansion_detailed} have a natural geometric interpretation. The reason is that powers of the total content $K(\lambda)$ enumerate ramified coverings of the torus via Hurwitz numbers. In this section we turn the moment expansion of Section~\ref{sec:asympt_exp} into an integral representation over spaces of ramified coverings. This yields the random surface formulation announced in Theorem~\ref{thm:main}.

\subsection{Integration over Hurwitz spaces}\label{sec:Hurwitz_integration}

Let $\Sigma_g$ be a compact connected Riemann surface of genus $g\geq 0$. A \emph{ramified covering} of $\Sigma_g$ is the data of a Riemann surface $X$ together with a non-constant holomorphic map $\pi:X\to\Sigma_g$. By standard results of complex analysis, $\pi^{-1}(x)\subset X$ is discrete for any $x\in\Sigma_g$. The surface $\Sigma_g$ is called the \emph{base space}, or the \emph{target space}. Locally around a point $p\in X$, $\pi$ is conjugated (using holomorphic coordinate changes) to a monomial $z\mapsto z^k$ on $\C$, with $k\geq 1$, and $e_p=k$ is called the \emph{ramification index} of $X$ at $p$. The point $p\in X$ is said to be a \emph{ramification point} if $e_p>1$, and we denote by $R=\{\pi(p),p\in X\ \text{s.t.}\ e_p>1\}\subset \Sigma_g$ the \emph{ramification locus} of $X\to\Sigma_g$, which is also a discrete set. There is an integer $n\geq 1$ such that for any $x\in\Sigma_g\setminus R$, its preimage $\pi^{-1}(x)$ is a finite set of cardinal $n$, and the integer $n$ is denoted by $\deg(X)$ and called the \emph{degree} of the covering. Note that for any $x\in\Sigma_g$ we have the identity
\[
n=\deg(X)=\sum_{p\in\pi^{-1}(x)}e_p,
\]
so that the preimage $\pi^{-1}(x)$ has less than $n$ points for $x\in R$. The Riemann--Hurwitz formula states that
\begin{equation}\label{eq:Riemann-Hurwitz}
\chi(X)=\deg(X)\chi(\Sigma_g)-\sum_{p\in X}(e_p-1),
\end{equation}
where $\chi(X)=2-2g(X)$ is the Euler characteristic of the surface $X$.

Thanks to the monodromy map, for any fixed $x\in\Sigma_g\setminus R$, ramified coverings $\pi:X\to\Sigma_g$ of degree $n$ with ramification locus $R$ are in bijection with the set $\Hom(\pi_1(\Sigma_g\setminus R,x),S_n)$ of homomorphisms from the fundamental group of $\Sigma_g\setminus R$ based in $x$ to the permutation group $S_n.$ Equivalence classes of ramified coverings of $\Sigma_g$ are therefore in bijection with
\[
\Hom(\pi_1(\Sigma_g\setminus R,x),S_n)/S_n.
\]
We denote by $[X]$ the equivalence class of a ramified covering $X.$ Note that the degree and Euler characteristic of a ramified covering only depends on its equivalence class : we can denote it by $\deg(X) = \deg([X])$ and $\chi(X) = \chi([X])$ respectively.
If $R$ is empty, then it is the space of representations of surface groups, and it corresponds to the set of equivalence classes of \emph{unramified coverings}. Random models in this setting have been extensively studied for compact surfaces of genus $g\geq 2$ by Magee, Naud and Puder \cite{MagPud,MN,MNP}. 

In this work, we consider the case of a complex torus $\Tbb$, i.e. $g=1$. Given conjugacy classes $C_1,\ldots,C_k$ of $S_n$, the subset of ramified coverings $X\to\Tbb$ with ramification locus $R=\{x_1,\ldots,x_k\}$ and with monodromy in the class $C_i$ around $x_i$ is the set
\[
\{(\sigma_1,\sigma_2,\tau_1,\ldots,\tau_k)\in S_n^2\times C_1\times\ldots\times C_k\ \vert \ [\sigma_1,\sigma_2]\tau_1\ldots\tau_k=1\}.
\]
Note that $C_1,\ldots,C_k$ are assumed to be nontrivial conjugacy classes, otherwise there would be no ramification. We say that there is a \emph{generic ramification} over $x\in R$ if the corresponding monodromy is conjugated to a transposition. Let $\mathbb{X}_1(n,k)$ be the set of ramified coverings $X\to\Tbb$ with degree $n$ and with $k$ generic ramifications. Its quotient by the action of $S_n$ is the \emph{Hurwitz space} $\mathcal{H}_1(n,k)$, whose cardinal is the \emph{Hurwitz number} $H_1(n,k)$.

The Riemann--Hurwitz formula~\eqref{eq:Riemann-Hurwitz} gives
\[
\chi(X)=-2k,\quad \forall [X]\in\mathcal{H}_1(n,2k),
\]
therefore the Euler characteristic of $X$ does not depend on its degree, which is specific of the case where the target surface is a torus.

In general, Hurwitz numbers can be expressed in terms of irreducible representations of the symmetric group \cite{LZ}, and it is in particular the case of those we care about:
\begin{equation}\label{eq:Hurwitz_torus}
H_1(n,k)=\sum_{\alpha\vdash n} K(\alpha)^k.
\end{equation}
Using Proposition~\ref{prop:symmetry_partitions} we find that $H_1(n,k)=0$ for all odd $k$. The asymptotic behaviour of the number $H'_1(n,k)$ of \emph{connected}\footnote{For a correspondence between enumeration of connected and disconnected coverings, see \cite[\S 2.2]{EO} for instance.} ramified coverings has been investigated in the groundbreaking paper \cite{EO} by Eskin and Okounkov (recently generalized by Aggarwal \cite{Agg}), who used them to compute asymptotics of Masur--Veech volumes of strata of moduli spaces of abelian differentials, or translation surfaces.

As opposed to the aforementioned works, we do not exactly consider a uniform counting of equivalence classes of ramified coverings, because we also randomize the degree and number of ramification points of the covering: for any real parameter $t>0$, let $\rho_t$ be the measure on the countably infinite space
\[
\mathcal{R}=\bigsqcup_{n\geq 1}\bigsqcup_{k\geq 0}\mathcal{H}_1(n,2k),
\]
defined by
\[
d\rho_t([X])=\sum_{n\geq 1}q_t^n\sum_{k\geq 0}\frac{t^{2k}}{(2k)!}\sum_{[Y]\in\mathcal{H}_1(n,2k)}\delta_{[Y]}([X]).
\]

The measure $\rho_t$ can be interpreted as a randomization of the counting measure on $\mathcal{H}_1(n,2k)$, when $n$ is taken to follow a geometric distribution of parameter $1-q_t$ and $m=2k$ a Poisson distribution of parameter $t$ conditioned to take even values. It is the measure-theoretic equivalent of (a specialization of) the two-parameter generating series $\mathcal{F}_1\in\C[[q,X]]$ of Hurwitz numbers on the torus with varying degrees and generic ramifications \cite{LZ}:
\[
\mathcal{F}_1(q,X)=\sum_{n\geq 1}\sum_{k\geq 0}q^n\frac{X^{2k}}{(2k)!}H_1(n,2k).
\]
In particular, we will be interested in integrals of functions of the degree and Euler characteristic of the coverings: for any function $f:\N\times\Z\to\R$ we have the identity
\begin{equation}\label{eq:conversion_integral_Hurwitz}
\int_{\mathcal{R}}f(\deg(X),\chi(X))d\rho_t([X])=\sum_{n\geq 1}q_t^n\sum_{k\geq 0}\frac{t^{2k}}{(2k)!}f(n,-2k)H_1(n,2k).
\end{equation}

\begin{remark}
The measure $\rho_t$ has infinite mass (see Lemma~\ref{lem:infinite_mass} for more details), so it cannot be normalized into a probability measure. In that sense, it is a slight language abuse to speak about ``random surfaces". However, if we replace the counting measure on $\mathcal{H}_1(n,k)$ by the uniform measure and we apply the proper normalization to the distributions of $\deg(X)$ and $\chi(X)$ we obtain a true probability measure
\[
d\mu_t([X])=\sum_{n\geq 1}\frac{q_t^n}{1-q_t}\sum_{k\geq 0}\frac{t^{2k}}{(2k)!\cosh(t)}\sum_{[Y]\in\mathcal{H}_1(n,2k)}\frac{1}{H_1(n,2k)}\delta_{[Y]}([X]).
\]
This normalization is inessential here, and using it would rather complicate the formulas, therefore we will stick with the initial measure $\rho_t$, while keeping the terminology of random surfaces.
\end{remark}

\begin{remark}
Our definition of the measure on the space of ramified coverings is slightly different from the one we considered in \cite{LM2}. The main difference in the approach of the present paper is that the ramification locus is not specified, only the number of ramification points; this enables a simpler definition, and the measures considered are in fact equivalent, because the measure $\rho_t$ only depends on the number of ramification points, not on their positions. We believe that it avoids superfluous -- yet interesting -- topological and measure-theoretic developments. The interested readers might find such developments in the last chapter of \cite{Lev10}.
\end{remark}

As $\rho_t$ has infinite mass, any integral representation with respect to $\rho_t$ needs a priori to have a cutoff, be it on the Euler characteristic (cutoff on $k$) or on the degree (cutoff on $n$). As pointed out above, in the case of the torus, both quantities are independent, therefore we have to make a choice. We will explore both approaches in this paper, because we believe that both have their merits and demerits: imposing a cutoff on the Euler characteristic of on the degree will lead us to Theorem~\ref{thm:mainbis} and~\ref{thm:main_alt} respectively.

\subsubsection{Cutoff on  Euler characteristic}

Let us start with a formula which is close to the asymptotic expansion of Theorem~\ref{thm:asympt_expansion}: we need to introduce a few notations.
\begin{enumerate}
\item For any $([X_1],[X_2])\in \mathcal{R}^2$, any integers $N,p\geq 1$, and any real number $t>0$:
\begin{equation}\label{eq:Phi_Aprime}
\Phi_{t,N}^{p,A'}(X_1,X_2)=\sum_{k\in I_p(X_1,X_2)}\sum_{n\in\Z}q_t^{n^2}\frac{[nt(\deg(X_1)-\deg(X_2))]^{2k}}{N^{2k}(2k)!},
\end{equation}
\begin{equation}\label{eq:Phi_A}
\Phi_{t,N}^{p,A}(X_1,X_2)=\sum_{k\in I_p(X_1,X_2)}\frac{t^k(\deg(X_1)-\deg(X_2))^{2k}}{2^{k}N^{2k}k!}
\end{equation}
where $I_p(X_1,X_2)=\{0,1,\ldots, p+\frac12\chi(X_1)+\frac12\chi(X_2)\}$.\\
\item For any $[X]\in \mathcal{R}$, any integer $N\geq 1$ and any real number $t>0$:
\begin{equation}\label{eq:Phi_BD}
\Phi_{t,N}^{B}(X)=\Phi_{t,N}^{D}(X)=e^{\frac{t}{2N}\deg(X)},
\end{equation}
\begin{equation}\label{eq:Phi_C}
\Phi_{t,N}^{C}(X)=e^{-\frac{t}{2N}\deg(X)}.
\end{equation}
\end{enumerate}
We can now state the result:

\begin{theorem}[Cutoff on Euler characteristic]\label{thm:mainbis}
Let $(G_N)_{N \geq 1}$ be a sequence of compact classical groups of type $\tau\in\{A',A,B,C,D\}$, with $G_N\subset\GL_N(\C)$. For any $p\geq 1$ and any $t>0$, we have as $N\to\infty$:
\begin{enumerate}
\item If $\tau\in\{A',A\}$, 
\begin{align}\label{eq:integral_A}
\begin{split}
\Tr(e^{\frac{t}{2}\Delta})= & \int_{\mathcal{R}^2}\mathbf{1}_{\{\chi(X_1)+\chi(X_2)\geq -2p\}} \Phi_{t,N}^{2p,\tau}(X_1,X_2)\\
&\times N^{\chi(X_1)+\chi(X_2)}d\rho_t^{\otimes 2}([X_1],[X_2])+O_t(N^{-2p-2}).
\end{split}
\end{align}
\item If $\tau\in\{B,C,D\}$, 
\begin{equation}\label{eq:integral_BCD}
\Tr(e^{\frac{t}{2}\Delta})=\int_{\mathcal{R}}\mathbf{1}_{\{\chi(X)\geq -p\}}\Phi_{t,N}^{\tau}(X)N^{\chi(X)}d\rho_t([X])+O_t(N^{-p-1}).
\end{equation}
\end{enumerate}
\end{theorem}

For $\tau\in\{A',A\}$, the functionals $\Phi_{t,N}^{p,\tau}$ are polynomial functions of the degrees of the ramified coverings, and converge to a holomorphic function $\Phi_{t,N}^{\tau}$ as $p\to\infty$:
\begin{equation}
\label{def:phiA'}
\Phi_{t,N}^{A'}(X_1,X_2)=\lim_{p\to\infty}\Phi_{t,N}^{p,A'}(X_1,X_2)=\sum_{n\in\Z}e^{-\frac{t}{2}n^2}\cosh\left(\frac{t}{N} n(\deg(X_1)-\deg(X_2))\right),
\end{equation}
\begin{equation}
\label{def:phiA}
\Phi_{t,N}^{A}(X_1,X_2)=\lim_{p\to\infty}\Phi_{t,N}^{p,A}(X_1,X_2)=e^{\frac{t}{2N^2}(\deg(X_1)-\deg(X_2))^2}.
\end{equation}
It would have been more satisfying to replace $\Phi_{t,N}^{p,\tau}$ by $\Phi_{t,N}^\tau$ in \eqref{eq:integral_A}, just like in \eqref{eq:integral_BCD}, but unfortunately, $\Phi_{t,N}^{A}$ and $\Phi_{t,N}^{A'}$ are not integrable with respect to $\rho_t^{\otimes 2}$ (see Proposition~\ref{prop:non_integrability_PhiA} in the Appendix). In any case, the large-$N$ behavior of $\Phi_{t,N}^{p,\tau}$ is quite trivial (and does not depend on $p$):
\[
\lim_{N\to\infty}\Phi_{t,N}^{p,A'}(X_1,X_2)=\theta(q_t),\quad \forall p\geq 1,
\]
\[
\lim_{N\to\infty}\Phi_{t,N}^{p,A}(X_1,X_2)=1,\quad \forall p\geq 1,
\]
\[
\lim_{N\to\infty}\Phi_{t,N}^{\tau}(X)=1,\quad \forall \tau\in\{B,C,D\}.
\]

\subsubsection{Cutoff on degree}

Now, let us rather consider a cutoff on the degree: from the point of view of random partitions, it means to put a restriction on the sizes of partitions. Such cutoff has already been successfully used in similar contexts: \cite{Lem,Mag,Lem3}. In particular, in \cite{Lem,Lem3}, a growth condition that depends on some power $N^\gamma$ with a small $\gamma>0$ has been allowed, leading to precise estimates of some observables of random partitions and random highest weights. We will apply the same settings here.

\begin{theorem}[Cutoff on degree]\label{thm:main_alt}
Let $(G_N)_{N\geq 1}$ be a sequence of compact classical groups of type $\tau\in\{A',A,B,C,D\}$, with $G_N\subset\GL_N(\C)$. For any $t>0$ and any $\gamma\in(0,1)$, we have as $N\to\infty$:
\begin{enumerate}
\item If $\tau\in\{A',A\},$ with $\Phi_{t,N}^\tau$ as defined in \eqref{def:phiA'} and \eqref{def:phiA} respectively,
\begin{align}
\begin{split}
\Tr(e^{\frac{t}{2}\Delta}) = & \int_{\mathcal{R}^2}\mathbf{1}_{\{\deg(X_1),\deg(X_2)\leq N^\gamma\}}\Phi_{t,N}^\tau(X_1,X_2)\\
&\times N^{\chi(X_1)+\chi(X_2)}d\rho_t^{\otimes 2}([X_1],[X_2])+O_t(e^{-\frac{t}{8}N^\gamma}),
\end{split}
\end{align}
\item If $\tau\in\{B,C,D\}$, with $\Phi_{t,N}^\tau$ as defined in \eqref{eq:Phi_BD} and \eqref{eq:Phi_C} respectively,
\begin{equation}
\Tr(e^{\frac{t}{2}\Delta})=\int_{\mathcal{R}}\mathbf{1}_{\{\deg(X)\leq N^\gamma\}}\Phi_{t,N}^\tau(X)N^{\chi(X)}d\rho_t([X])+O_t(e^{-\frac{t}{12}N^\gamma}).
\end{equation}
\end{enumerate}

\end{theorem}

This time, the set of integration depends on $N$, we will see in the proof of Theorem \ref{thm:main_alt} that it is the price to pay for getting an exponentially small remainder.

\subsection{Proof of Theorems~\ref{thm:mainbis} and~\ref{thm:main_alt}}

We will now prove Theorems~\ref{thm:mainbis} and~\ref{thm:main_alt} for all groups. Let us start with Theorem~\ref{thm:mainbis}: the general outline of the proof is to start from the asymptotic expansion obtained in Theorem~\ref{thm:asympt_expansion}, and to reorder the terms appropriately. In the case of special orthogonal and symplectic groups, we will also add negligible terms to the initial asymptotic expansion in order to get the right integrands.

\begin{proof}[Proof of Theorem \ref{thm:mainbis} for $\U(N)$]

According to Theorem~\ref{thm:asympt_expansion_detailed}, we have for any $p\geq 0$
\begin{align*}
\Tr(e^{\frac{t}{2}\Delta_{\U(N)}})= & \sum_{k=0}^p\frac{t^{2k}}{N^{2k}(2k)!}\sum_{n\in\Z}q_t^{n^2}\sum_{n_1,n_2\geq 1} q_t^{n_1+n_2}\\
&\times \sum_{\alpha\vdash n_1,\beta\vdash n_2}(K(\alpha)+K(\beta)+n(n_1-n_2))^{2k}+O_t(N^{-2p-2})\\
= & \sum_{k=0}^p\left(\frac{t}{N}\right)^{2k}\sum_{k_1+k_2+k_3=2k}\frac{1}{k_1!k_2!k_3!}\sum_{n\in\Z}q_t^{n^2}n^{k_3} \sum_{n_1,n_2\geq 1}q_t^{n_1+n_2}\\
& \times (n_1-n_2)^{k_3}\sum_{\alpha\vdash n_1}K(\alpha)^{k_1}\sum_{\beta\vdash n_2}K(\beta)^{k_2}+O_t(N^{-2p-2}).
\end{align*}
Note that by symmetry, if $k_3$ is odd, $\sum_{n\in\Z}q_t^{n^2}n^{k_3}=0$ and we can restrict to even values of $k_3$. The same happens with $k_1$ and $k_2$ thanks to Proposition~\ref{prop:symmetry_partitions}. We get
\begin{align*}
\Tr(e^{\frac{t}{2}\Delta_{\U(N)}})= & \sum_{k_1+k_2+k_3\leq p}\left(\frac{t^2}{N^2}\right)^{k_1+k_2+k_3}\frac{1}{(2k_1)!(2k_2)!(2k_3)!}\sum_{n\in\Z}q_t^{n^2}n^{2k_3} \\
& \times \sum_{n_1,n_2\geq 1}q_t^{n_1+n_2} (n_1-n_2)^{2k_3}H_1(n_1,2k_1)H_1(n_2,2k_2)+O_t(N^{-2p-2}).
\end{align*}
We can put the sums over $n_1,n_2$ outside, and slice the (finite) sum over $k_1+k_2+k_3$ on a first sum over $k_1+k_2$ and a second one over $k_3$:
\begin{align*}
\Tr(e^{\frac{t}{2}\Delta_{\U(N)}})= &\sum_{n_1,n_2\geq 1}q_t^{n_1+n_2} \sum_{k_1+k_2\leq p}\left(\frac{t^2}{N^2}\right)^{k_1+k_2}\frac{H_1(n_1,2k_1)H_1(n_2,2k_2)}{(2k_1)!(2k_2)!}\\
& \times \sum_{k_3\leq p-(k_1+k_2)}\frac{1}{(2k_3)!}\sum_{n\in\Z}q_t^{n^2}\left(\frac{tn(n_1-n_2)}{N}\right)^{2k_3} +O_t(N^{-2p-2}).
\end{align*}
Thanks to \eqref{eq:conversion_integral_Hurwitz}, the sum can be rewritten as an integral over $\mathcal{R}^2$ that is exactly \eqref{eq:integral_A}, by making the following identifications: $n_1=\deg(X_1)$, $n_2=\deg(X_2)$, $2k_1=-\chi(X_1)$, $2k_2=-\chi(X_2)$.
\end{proof}

The proof of Theorem~\ref{thm:mainbis} for $\SU(N)$ is pretty much identical and shall be omitted, although the computations slightly differ. In the case of $\SO(N)$, there is an additional argument that justifies to reproduce the whole proof.

\begin{proof}[Proof of Theorem~\ref{thm:mainbis} for $\SO(N)$]
According to Theorem~\ref{thm:asympt_expansion_detailed}, we have for any $p\geq 0$
\begin{align*}
\Tr(e^{\frac{t}{2}\Delta_{\SO(N)}})= & \sum_{k=0}^p\frac{(-t)^k}{N^kk!}\sum_{n\geq 1}q_t^n\sum_{\mu\vdash n}\left(K(\mu)-\frac12n\right)^k+O_t(N^{-p-1})\\
= & \sum_{k_1+k_2\leq p}\left(\frac{-t}{N}\right)^{k_1+k_2}\frac{(-1)^{k_2}}{2^{k_2}k_1!k_2!}\sum_{n\geq 1}q_t^nn^{k_2}\sum_{\mu\vdash n}K(\mu)^{k_1}+O_t(N^{-p-1})\\
= & \sum_{2k_1+k_2\leq p}\frac{t^{2k_1+k_2}}{N^{2k_1+k_2}2^{k_2}(2k_1)!k_2!}\sum_{n\geq 1}q_t^nn^{k_2}H_1(n,2k_1)+O_t(N^{-p-1}).
\end{align*}
Let us split the sum into a sum over $2k_1\leq p$ and a sum over $k_2\leq p-2k_1$:
\begin{align*}
\Tr(e^{\frac{t}{2}\Delta_{\SO(N)}})= & \sum_{2k_1\leq p}\frac{t^{2k_1}}{N^{2k_1}(2k_1)!}\sum_{n\geq 1}q_t^nH_1(n,2k_1) \sum_{k_2=0}^{p-2k_1}\frac{(tn)^{k_2}}{(2N)^{k_2}k_2!}+O_t(N^{-p-1}).
\end{align*}
Recall that for any $x\in\R$ and any $k_0\geq 1$,
\begin{equation}\label{eq:exp_Taylor}
\exp(x)=\sum_{k=0}^{k_0}\frac{x^k}{k!}+R_{k_0}(x),
\end{equation}
where
\begin{equation}\label{eq:exp_Taylor2}
R_{k_0}(x)=\frac{x^{k_0+1}}{k_0!}\int_0^1(1-s)^{k_0}\exp(sx)ds\leq \frac{\vert x\vert^{k_0+1}\exp(\vert x\vert)}{(k_0+1)!}.
\end{equation}
We apply this to $x=\frac{tn}{2N}$ and $k_0=p-2k_1$, so that
\[
\sum_{k_2=0}^{p-2k_1}\frac{(tn)^{k_2}}{(2N)^{k_2}k_2!}=e^{\frac{tn}{2N}}-R_{p-k_1}\left(\frac{tn}{2N}\right).
\]
We put it back in the expansion, using
\begin{equation}\label{eq:bound_integral_SO}
\tilde{R}_{N}(p,t):=\sum_{2k_1\leq p}\frac{t^{2k_1}}{N^{2k_1}(2k_1)!}\sum_{n\geq 1}q_t^n H_1(n,2k_1)R_{p-k_1}\left(\frac{tn}{2N}\right)=O_t(N^{-p-1}),
\end{equation}
and we obtain \eqref{eq:integral_BCD} by appropriate identifications. It remains to prove \eqref{eq:bound_integral_SO}: by \eqref{eq:exp_Taylor2}, the RHS of \eqref{eq:bound_integral_SO} is bounded in absolute value by
\[
\vert \tilde{R}_{N}(p,t)\vert\leq \sum_{2k_1\leq p}\frac{t^{2k_1}}{N^{2k_1}(2k_1)!}\sum_{n\geq 1}q_t^n H_1(n,2k_1)\frac{(tn)^{p-2k_1+1}\exp(\frac{tn}{2N})}{(2N)^{p-2k_1+1}(p-2k_1+1)!}.
\]
Applying Lemma \ref{prop:bound_Hurwitz}, we get now
\begin{align*}
\vert \tilde{R}_{N}(p,t)\vert\leq & \sum_{2k_1\leq p}\frac{t^{p+1}}{N^{p+1}(2k_1)!(p-2k_1+1)!}\sum_{n\geq 1}q_t^{n(1-1/N)} p(n)n^{4k_1+p+1}\\
\leq & \frac{t^{p+1}}{N^{p+1}(p+1)!}\sum_{2k_1\leq p}\binom{p+1}{2k_1}\sum_{n\geq 1}n^{3p+1}e^{-\frac{t}{2}n(1-1/N)+\pi\sqrt{2n/3}},
\end{align*}
where the last equality uses the inequality $p(n)\leq e^{\pi\sqrt{2n/3}}$ for all $n$ (see \cite[Theorem 14.5]{Apo}). Now \eqref{eq:bound_integral_SO} becomes obvious.
\end{proof}

The proof for $\Sp(N)$ is exactly the same as $\SO(N)$, simply by replacing $N$ by $2N$ and $(K(\mu)-n/2)^k$ by $(K(\mu)+n/2)^k$, thus we omit it.\\

Now let us turn to Theorem~\ref{thm:main_alt}: the proof will not depend on the asymptotic expansion, but starts from Corollary~\ref{cor:heat_trace_expectation}. From there, we introduce a new cutoff on random partitions, based on their size rather than on their length. We show that the remainder is exponentially small, and that the partial sum is directly the announced integral. In this case, the proof is identical\footnote{There is only a slight additional argument for $\SO(2N)$ which is the same as for Theorem~\ref{thm:asympt_expansion_detailed}, and which is that the term $L_N$ defined in~\eqref{def:LN} is equal to 1, up to an exponentially small correction, which we already proved.} for all groups, therefore we only detail it for $\U(N)$.

\begin{proof}[Proof of Theorem~\ref{thm:main_alt} for $\U(N)$]
Let us start with the expression from Corollary~\ref{cor:heat_trace_expectation}. For $\gamma\in(0,1)$, let
\[
\tilde{\Lambda}_N^{(1)}(\gamma)=\{(\alpha,\beta,n)\in\Pfr^2\times\Z:\vert\alpha\vert,\vert\beta\vert\leq N^\gamma\}
\]
be a new subset with a cutoff on the size of partitions instead of their lengths. For $N$ large enough, $\tilde{\Lambda}_N^{(1)}(\gamma)\subset\Lambda_N^{(1)}$. In this case, \begin{align*}
\Tr(e^{\frac{t}{2}\Delta_{\U(N)}})= & \frac{\theta(q_t)}{\phi(q_t)^2}\E\left[e^{-\frac{t}{N}F^{A'}(\alpha,\beta,n)}\mathbf{1}_{\tilde{\Lambda}_N^{(1)}(\gamma)}(\alpha,\beta,n)\right]+\sum_{(\alpha,\beta,n)\in\Lambda_N^{(1)}\setminus\tilde{\Lambda}_N^{(1)}(\gamma)} e^{-\frac{t}{2}c_2(\lambda_N^{A'}(\alpha,\beta,n))}.
\end{align*}
Let us put the first term aside and control the second one. According to Lemma \ref{lem:lb-casimirs},
\begin{align*}
\sum_{(\alpha,\beta,n)\in\Lambda_N^{(1)}\setminus\tilde{\Lambda}_N^{(1)}(\gamma)} e^{-\frac{t}{2}\left(c_2(\lambda_N^{A'}(\alpha,\beta,n))\right)}
 \leq & \sum_{(\alpha,\beta,n)\in\Lambda_N^{(1)}\setminus\tilde{\Lambda}_N^{(1)}(\gamma)} q_t^{\frac12\vert\alpha\vert+\frac12\vert\beta\vert+\left(n+\frac{\vert\alpha\vert-\vert\beta\vert}{N}\right)^2}\\
\leq & \sum_{\vert\alpha\vert>N^\gamma}\sum_{\beta\in\Pfr}\sum_{n\in\Z}q_t^{\frac12\vert\alpha\vert+\frac12\vert\beta\vert+\left(n+\frac{\vert\alpha\vert-\vert\beta\vert}{N}\right)^2}\\
& + \sum_{\vert\beta\vert>N^\gamma}\sum_{\alpha\in\Pfr}\sum_{n\in\Z}q_t^{\frac12\vert\alpha\vert+\frac12\vert\beta\vert+\left(n+\frac{\vert\alpha\vert-\vert\beta\vert}{N}\right)^2}.
\end{align*}
On the one hand, for any $x\in\R$ and any $t>0$,
\[
\sum_{n\in\Z}q_t^{(n+x)^2}\leq \sup_{y\in[0,1]}\sum_{n\in\Z}q_t^{(n+y)^2}\leq C_1(t),
\]
where $C_1(t)>0$ is a constant that only depends on $t$ (due to the fact that the sum in the right defines a continuous function of $y\in[0,1]$), and on the other hand, by Lemma~\ref{lem:dev_ineq_bis},
\[
\sum_{\vert\alpha\vert>N^\gamma}q_t^{\frac12\vert\alpha\vert}=
\sum_{n > N^\gamma} p(n) q_{t/2}^{n} \leq C_2(t)e^{-\frac{t}{8}N^\gamma}.
\]
Therefore,
\[
\sum_{(\alpha,\beta,n)\in\Lambda_N^{(1)}\setminus\tilde{\Lambda}_N^{(1)}(\gamma)} e^{-\frac{t}{2}c_2(\lambda_N^{A'}(\alpha,\beta,n))}\leq 2C_1(t)\phi(q_{t/2})^{-1}C_2(t)e^{-\frac{t}{8}N^\gamma}.
\]
From there,  it is easy to get
\[
\Tr(e^{\frac{t}{2}\Delta_{\U(N)}})= \frac{\theta(q_t)}{\phi(q_t)^2}\E\left[e^{-\frac{t}{N}F^{A'}(\alpha,\beta,n)}\mathbf{1}_{\tilde{\Lambda}_N^{(1)}(\gamma)}(\alpha,\beta,n)\right]+O_t\left(e^{-\frac{t}{8}N^\gamma}\right).
\]
Rewriting the expectation as a triple sum over $\alpha$ yields
\begin{align*}
\Tr(e^{\frac{t}{2}\Delta_{\U(N)}})= & \sum_{n_1,n_2\leq N^\gamma}q_t^{n_1+n_2}\sum_{n\in\Z}q_t^{n^2}e^{-\frac{t}{N}n(n_1-n_2)} \sum_{\alpha\vdash n_1}e^{-\frac{t}{N}K(\alpha)}\sum_{\beta\vdash n_2}e^{-\frac{t}{N}K(\beta)}+O_t(e^{-\frac{t}{8}N^\gamma}),
\end{align*}
and the end of the proof consists in expanding the exponentials of contents in power series and identifying the measures.
\end{proof}

\section{Applications: spectral asymptotics and gauge/string duality}

We now present two families of applications of the asymptotic analysis that have been performed in Section \ref{sec:asympt_exp} with the help of Section \ref{sec:prelim}.
In Section \ref{sec:Cardy}, we gather some spectral consequences~: we first explain that Lemma \ref{lem:casimirs} can be reformulated as a spectral gap result for the Casimir operator and then go to  more involved application. The large-$N$ limit of the central heat trace determines a limiting counting measure for the Casimir spectrum, whose growth differs sharply from Weyl's law. Section \ref{sec:string} is more of physical and geometric flavour: the random surface representation gives a rigorous large-$N$ form of the Yang--Mills/Hurwitz duality on the torus, and the coefficient formulas of Theorem~\ref{thm:asympt_expansion_detailed} further imply a Yang--Mills/Gromov--Witten correspondence.

\subsection{Large-$N$ Casimir eigenvalues counting}\label{sec:Cardy}

Let us start with the asymptotic behavior of the Laplacian, or, equivalently, the Casimir operator $-\Delta_{G_N}$. As announced, we derive here two applications of our formulas.

\paragraph{Asymptotic spectral gap}

The first one is an estimate of the spectral gap. It is a direct reformulation of Lemma \ref{lem:casimirs}, but we believe it is worth mentioning by analogy with the current active research around spectral gaps of large random geometric objects (see for instance the recent breakthroughs \cite{HM23,AM1,AM2,HMT,HMT2} for spectral gaps of random surfaces). Recall that the lowest eigenvalue of $-\Delta_{G_N}$ is always 0 and the only highest weight that contributes to this eigenvalue is the trivial highest weight $(0,\ldots,0)$. The spectral gap is by definition the second smallest eigenvalue, that is, the smallest possible Casimir $c_2(\lambda)$ for $\lambda\neq(0,\ldots,0)$.

\begin{proposition}[Asymptotic spectral gap]\label{prop:spectral_gap}
Let $(G_N)_{N\geq 1}$ be a sequence of compact classical groups of type $\tau\in\{A',A,B,C,D\}$.
As $N$ tends to infinity, the Casimir numbers of all highest weights concentrate around integer values (or half-integer values if $G_N=\SO(2N)$). More specifically, let $\alpha,\beta,\mu$ be integer partitions and $m\in\N$ and $n\in\Z$ be integers, then for $N$ large enough,
\begin{equation}\nonumber
c_2(\lambda_N^{A'}(\alpha,\beta,n))=\vert\alpha\vert+\vert\beta\vert+n^2+O\left(\frac1N\right),
\end{equation}
\begin{equation}\nonumber
c_2(\lambda_N^{A}(\alpha,\beta))=\vert\alpha\vert+\vert\beta\vert+O\left(\frac1N\right),
\end{equation}
\begin{equation}\nonumber
c_2(\lambda_N^{B}(\mu))=\vert\mu\vert+O\left(\frac1N\right),
\end{equation}
\begin{equation}\nonumber
c_2(\lambda_N^{C}(\mu))=\vert\mu\vert+O\left(\frac1N\right),
\end{equation}
\begin{equation}\nonumber
c_2(\lambda_N^{D}(\mu,m,n))=\vert\mu\vert+\frac12m(N+m-1)+O\left(\frac1N\right).
\end{equation}
In particular, the negative Laplacian $-\Delta_{G_N}$ has an asymptotic spectral gap of 1 as $N\to\infty$.
\end{proposition}

\begin{proof}
The expressions of Casimirs follow either from Lemma~\ref{lem:casimirs}, by isolating the dependence in $N$ in each formula, or from Proposition~\ref{prop:shifted-casimir} by noticing that $\mathbf{p}_1(\alpha)=\vert\alpha\vert$ and $\mathbf{p}_2(\alpha)=2K(\alpha)$ for any $\alpha\in\Pfr.$ For the spectral gap, we need to use the bijections from Proposition~\ref{prop:dual_sets} and express the highest weights in terms of partitions, then let $N$ tend to infinity. We see that in this regime, the lowest nonzero Casimirs are given by the partitions with smallest sizes: for $\U(N)$ we have either $(\alpha,\beta,n)=((1),\varnothing,0)$ or $(\varnothing,(1),0)$, for $\SU(N)$ we have $(\alpha,\beta)=((1),\varnothing)$ or $(\varnothing,(1))$, and for the other cases we have $\mu=(1)$, with $m=n=0$ for the case of $\SO(2N)$. In all these cases, the Casimir converges to 1, so that asymptotically there is no eigenvalue of $-\Delta$ in $(0,1)$.
\end{proof}

\paragraph{Asymptotic counting law}

The second application is Theorem~\ref{thm:Cardy}, whose detailed version is stated below.

\begin{theorem}\label{thm:Cardy-detailed}
For any classical type $\tau\in\{A',A,B,C,D\}$, let $(G_N)_{N\geq1}$ be a sequence of compact classical groups of type $\tau$ with $G_N\subset\GL_N(\C)$.
 There exists a locally finite measure $\mu_\tau$ on $[0,\infty)$ such that, at every
continuity point $\Lambda$ of $\mu_\tau$,
\begin{equation}\label{eq:cv_continuitypoint}
\mu_N([0,\Lambda])\build{\longrightarrow}{N\to\infty}{} \mu_\tau([0,\Lambda]).
\end{equation}
Moreover, if we set $\mathcal{N}_\tau(\Lambda):=\mu_\tau([0,\Lambda]),$ then, as $\Lambda\to\infty$,
\begin{equation}\label{eq:Cardy1}
\mathcal{N}_{A'}(\Lambda)\sim \frac{1}{4\pi}\Lambda^{-1/2}
e^{2\pi\sqrt{\frac{\Lambda}{3}}},
\end{equation}
\begin{equation}\label{eq:Cardy2}
\mathcal{N}_A(\Lambda)\sim \frac{1}{4\pi3^{1/4}}\Lambda^{-3/4}
e^{2\pi\sqrt{\frac{\Lambda}{3}}},
\end{equation}
and
\begin{equation}\label{eq:Cardy3}
\mathcal{N}_B(\Lambda)\sim \mathcal{N}_C(\Lambda)\sim \mathcal{N}_D(\Lambda)\sim
\frac{1}{2\pi\sqrt{2}}\Lambda^{-1/2}
e^{\pi\sqrt{\frac{2\Lambda}{3}}}.
\end{equation}
\end{theorem}

The main additional input that we require is a general Tauberian estimate obtained recently by Bringmann--Jennings-Shaffer--Mahlburg \cite{BJSM23}, which refines a landmark theorem by Ingham \cite{Ing41}.

\begin{theorem}[\cite{BJSM23}, Theorem 1.1]\label{thm:BJSM}
Suppose that $B(q)=\sum_{n\geq 0}b_n q^n$ is a power series with non-negative real coefficients and radius of convergence $\geq 1$. If there exist $\lambda,\alpha,\beta,\gamma\in\R$ with $\gamma>0$ such that
\begin{equation}
B(e^{-t})\sim \lambda\log\left(\frac1t\right)^\alpha t^\beta e^{\frac{\gamma}{t}},\quad t\downarrow 0, t \in \mathbb R^+,
\end{equation}
\begin{equation}
B(e^{-z})\ll \log\left(\frac{1}{\vert z\vert}\right)^\alpha \vert z\vert^\beta e^{\frac{\gamma}{\vert z\vert}},\quad \vert z\vert\downarrow 0,
\end{equation}
with $z=x+iy$, $y>0$, in each region\footnote{These regions are often named Stolz angles, or Stolz sectors, see e.g. \cite{Dur70}.} of the form $\vert y\vert\leq \Delta x$ for $\Delta>0$, then
\begin{equation}
\sum_{n=0}^\Lambda b_n\sim \frac{\lambda\gamma^{\frac{\beta}{2}-\frac14}\log(\Lambda)^\alpha}{2^{\alpha+1}\sqrt{\pi}\Lambda^{\frac{\beta}{2}+\frac14}}e^{2\sqrt{\gamma \Lambda}},\quad \Lambda\to\infty.
\end{equation}
\end{theorem}

\begin{proof}[Proof of Theorem~\ref{thm:Cardy-detailed}]
For each $N$, set
\[
Z_N(t):=\Tr\left(e^{\frac t2\Delta_{G_N}}\right)
=\int_{[0,\infty)} e^{-tE/2}d\mu_N(E),
\qquad
\mu_N=\sum_{\lambda\in\widehat G_N}\delta_{c_2(\lambda)}.
\]
By Corollary~\ref{cor:lim_pf}, for every $t>0,$
\[
\int_{[0,\infty)} e^{-tE/2}\,d\mu_N(E)=Z_N(t)\xrightarrow[N \rightarrow\infty]{} Z_\tau(t).
\]
Therefore, by Feller's extended continuity theorem for Laplace transforms \cite[Ch. XIII, \S 1, Theorem 2a]{FellerII}, one can write 
\[
Z_\tau(t)=\int_{[0,\infty)}e^{-tE/2}\,d\mu_\tau(E),
\]
with $\mu_\tau$ a locally finite measure on $[0, \infty)$ and  the convergence of the Laplace transforms implies
\[
\mu_N([0,\Lambda])\xrightarrow[N \rightarrow\infty]{} \mu_\tau([0,\Lambda])
\]
at every continuity point $\Lambda$ of $\mu_\tau$, which is exactly~\eqref{eq:cv_continuitypoint}.\\

Moreover, from Corollary~\ref{cor:lim_pf}, we have that
\[  Z_\tau(t) = B_\tau(e^{-t/2}),\]
with 
\[ B_\tau(q) = \theta(q)^{\varepsilon_\tau}\phi(q)^{-\delta_\tau},\]
where $\varepsilon_\tau, \delta_\tau \geq 0.$

From there, one can easily check that 
\[
B_\tau(q)=\sum_{n\geq 0} a_n^{(\tau)} q^n,
\]
where  the coefficients $a_n^{(\tau)}$ are nonnegative and the radius of convergence of the series is at least~1.

We also have that the measure $\mu_\tau$ can be written
\[
\mu_\tau=\sum_{n\geq 0} a_n^{(\tau)}\delta_n,
\]
so that
\[
\mathcal{N}_\tau(\Lambda):=\mu_\tau([0,\Lambda])=\sum_{n\leq\Lambda} a_n^{(\tau)},
\]
and we can apply   Theorem~\ref{thm:BJSM} to study the asymptotics as $\Lambda$ grows to infinity.

As $s\downarrow 0$, the modular asymptotics of the Jacobi theta and Euler functions~\cite{FlaSed09} yield
\[
\theta(e^{-s})\sim \sqrt{\frac{\pi}{s}},
\qquad
\phi(e^{-s})\sim \sqrt{\frac{2\pi}{s}}e^{-\pi^2/(6s)}.
\]
Hence
\[
B_\tau(e^{-s})
\sim \lambda_\tau s^{(\delta_\tau-\varepsilon_\tau)/2}
\exp\left(\frac{\gamma_\tau}{s}\right),
\]
where
\[
\lambda_\tau=\pi^{\varepsilon_\tau/2}(2\pi)^{-\delta_\tau/2},
\qquad
\gamma_\tau=\frac{\delta_\tau\pi^2}{6}.
\]
Moreover, the same modular formulas are uniform in every Stolz angle
$z=x+iy$, $|y|\leq\Delta x$, so that
\[
B_\tau(e^{-z})\ll |z|^{(\delta_\tau-\varepsilon_\tau)/2}
\exp\left(\frac{\gamma_\tau}{|z|}\right),
\qquad |z|\downarrow 0.
\]
Therefore, Theorem~\ref{thm:BJSM}  with
\[
\alpha=0,\qquad
\beta=\frac{\delta_\tau-\varepsilon_\tau}{2},\qquad
\gamma=\gamma_\tau,\qquad
\lambda=\lambda_\tau,
\]
 gives
\[
\mathcal{N}_\tau(\Lambda)
\sim
\frac{\lambda_\tau\gamma_\tau^{\frac{\delta_\tau-\varepsilon_\tau-1}{4}}}
{2\sqrt{\pi}}
\Lambda^{-\frac{\delta_\tau-\varepsilon_\tau+1}{4}}
\exp\left(2\sqrt{\gamma_\tau\Lambda}\right).
\]
Substituting the values of $(\varepsilon_\tau,\delta_\tau)$ yields~\eqref{eq:Cardy1}, \eqref{eq:Cardy2} and~\eqref{eq:Cardy3}, which proves the theorem.
\end{proof}

\subsection{Gauge/string dualities}\label{sec:string}

We now turn to the second application of the expansion, namely its interpretation in two-dimensional Yang--Mills theory. While the preceding subsection used only the Laplace transform interpretation of the central heat trace, the next one uses the identification of the same trace with the genus-one Yang--Mills partition function.

Euclidean quantum Yang--Mills theory on a manifold $M$ endowed with a volume form corresponds to path integrals with respect to the formal measure
\begin{equation}\label{eq:wrong_YM}
d\mu_{M,G}(A)=\frac{1}{Z}e^{-\frac{1}{2T}S_{\YM}(A)}\mathcal{D}A,
\end{equation}
defined on the space $\mathcal{A}/\mathcal{G}$ of connections modulo gauge transformations of a $G$-principal bundle over $M$, where $G$ is a compact Lie group named the structure group. The group of gauge transformations is $\mathcal{G}=\mathscr{C}^\infty(\Sigma,G)$ and acts on the space of connections by
\[
g\cdot A = g^{-1}A g + g^{-1}dg,\quad \forall g\in\mathcal{G},\ \forall A\in\mathcal{A}.
\]
Let us describe the different parts of \eqref{eq:wrong_YM}:
\begin{itemize}
\item The measure $\mathcal{D}A$ is the push-forward of a formal Lebesgue measure\footnote{It is formal because there exists no such measure in the sense of Lebesgue integration.} on $\mathcal{A}$ by the quotient map $\mathcal{A}\to\mathcal{A}/\mathcal{G}$,
\item The functional $S_{\YM}:\mathcal{A}\to\R$ is a gauge-invariant action called the \emph{Yang--Mills action}, defined by
\[
S_\YM(A)=\langle F_A\wedge\star F_A\rangle,
\]
where $F_A=dA+A\wedge A$ is the curvature of the connection, 
\item $T>0$ is a real parameter called the \emph{coupling constant},
\item The quantity $Z$ is a normalization constant called the \emph{partition function}, ensuring that the total mass of $\mu_{M,G}$ is 1.
\end{itemize}
The partition function on a closed oriented surface of genus $g\geq 0$ and total area $t>0$ admits the following Fourier expansion \cite{Rus,Wit}:
\begin{equation}\label{eq:Migdal}
Z_G(g,t)=\sum_{\lambda\in\widehat{G}}d_\lambda^{2-2g}e^{-\frac{t}{2}c_2(\lambda)}.
\end{equation}
From this formula, it appears that Theorems~\ref{thm:asympt_expansion} and~\ref{thm:main} describe the asymptotic behavior of the partition function $Z_{G_N}(1,t)$ on a two-dimensional torus with structure group $G_N$. See \cite{Lem4} for a recent review of the construction of the Yang--Mills measure and the computation of the partition function.

We will show in this section that the main results of the present paper admit important consequences for the Yang--Mills partition function $Z_{G_N}(1,t)$.

\paragraph{Yang--Mills/Hurwitz duality}

Since the breakthrough work by 't Hooft \cite{Hoo74} and further developments by Gross and Taylor \cite{Gro,GT,GT2} in two dimensions, the observables of $\YM_2$ are conjectured to be related to string theory; although many descriptions have been given in physics for this duality in the last 30 years (this two-dimensional gauge/string duality remains an active topic in physics as of today \cite{PS,AKS,BenTro}), essentially formal results have been obtained. Many other gauge/string dualities are conjectured: the most famous might be the AdS/CFT correspondence proposed by Maldacena \cite{Mal}, and more recent dualities seem to be related to the geometric Langlands conjecture \cite{KW}, but most dualities in dimension higher than 2 are assuming supersymmetry (as far as we know). It is worth mentioning that the techniques of random partitions that we apply have also been helpful for a particular case of four-dimensional supersymmetric gauge theory called \emph{Seiberg--Witten theory} \cite{NO}. In the case of two-dimensional Yang--Mills theory, from the perspective of Gross and Taylor, the dual string theory corresponds to a theory of random surfaces, more specifically random ramified coverings of the base space (Hurwitz theory). Most of the literature treats the case of $\SU(N)$, but other classical groups also attracted a bit of interest \cite{Ram2}; in the case of $\SU(N)$, a coupling of random ramified coverings appears in the expressions and is interpreted as a coupling of two \emph{chiral sectors} of the theory. We shall not elaborate on this notion of chirality, but mention that the \emph{chiral partition function} is obtained by replacing the sum over couples of partitions by a sum over a single partition $\lambda$ such that $\ell(\lambda)\leq N$. The idea is that the chiral partition function is easier to understand in the large-$N$ limit, both in terms of asymptotic expansion and in terms of string representation (see \cite{Nov2024}), but it fails to capture the complete asymptotic behavior of the full partition function in general. This is why we developed the tools described in \S\ref{sec:hw_partition_corresp} in order to find the correct correspondence between partitions and highest weights.

A first step towards a mathematical proof of gauge/string duality on compact surfaces has been achieved for $g=1$ and $G=\U(N)$ by the authors in \cite[Theorem 1.2]{LM2}: they proved that $Z_{\U(N)}(1,t)$ admits an asymptotic expansion which can be described in terms of Hurwitz numbers $H_1(n,2k)$. Theorem~\ref{thm:asympt_expansion} extends this result to all compact classical groups, and in particular for $\SU(N)$ we recover the asymptotic expansion, in genus 1, corresponding to the so-called \emph{Gross--Taylor asymptotic series} \cite[Eq. (2.2)]{CMR}. A random surface interpretation in the spirit of Gross--Taylor was found in \cite[Corollary 6.2]{LM2} for $\U(N)$ for the chiral partition function, but the case of the full partition function remained unsolved -- and supposedly harder. Theorems~\ref{thm:mainbis} and~\ref{thm:main_alt} fill this gap, and not only bring an answer for the full partition with gauge group $\U(N)$, but also for the other groups. In particular, in the large-$N$ limit, their statements imply \cite[Theorem 1]{Dij}, and Equation \eqref{eq:integral_A} is a mathematically rigorous equivalent of \cite[Eq. (3.6)]{GT2} in genus 1.

In the cases of $\U(N)$ and $\SU(N)$ the couplings of random coverings from Theorems~\ref{thm:mainbis} and~\ref{thm:main_alt} correspond to the couplings of chiral sectors mentioned earlier, and the absence of coupling in the cases of $\SO(N)$ and $\Sp(N)$ can be interpreted as an intrinsic chirality of the theory for these gauge groups. The factor $N^{\chi(X)}$ in the random surface representation already appeared in the foundational paper \cite{Hoo74} and can be found in other mathematical works more or less related to gauge/string duality \cite{PPSY,CPS}.

\paragraph{Yang--Mills/Gromov--Witten duality}

As an alternative to the theory of random ramified coverings initiated by Gross and Taylor, we shall now introduce a new gauge/string duality, where the dual string model is a topological sigma-model coupled to gravity. Initially developed by Witten in the late '80s \cite{Wit88,Wit91b}, it is based on similar path integrals from conformal field theory. Given a Riemannian manifold $X$ (playing the role of spacetime), one considers path integrals over the space $\mathcal{M}_X$ of maps $\phi:\Sigma\to X$, where $\Sigma$ is a closed Riemann surface, with respect to a measure
\[
\frac{1}{\mathcal{Z}^X}e^{-S(\phi)}\mathcal{D}\phi.
\]
The measure $\mathcal{D}\phi$ is the formal Lebesgue measure on $\mathcal{M}_X$, which is ill-defined as in the case of Yang--Mills theory, and the functional $S$ is a conformally invariant Lagrangian. The partition function $\mathcal{Z}^X$ (we use a calligraphic $\mathcal{Z}$ to avoid confusion with the Yang--Mills partition function $Z$) is defined by
\[
\mathcal{Z}^X=\int_{\mathcal{M}_X}e^{-S(\phi)}\mathcal{D}\phi.
\]
This formal definition of path integrals can be made rigorous under further assumptions on $X$, with the help of Gromov--Witten invariants
\[
\langle\tau_{k_1}(\omega)\ldots\tau_{k_n}(\omega)\rangle_d^X.
\]
These invariants correspond to integrals of some cohomology classes $\tau_{k_i}(\omega)$, called $\tau$-classes, with respect to the moduli space of stable maps $\pi:C\to X$ of degree $d$. We defer to Appendix~\ref{sec:GW} the precise definition of these invariants and the associated moduli space, as we do not need them directly. The $n$-point correlation function of the physical model, which also corresponds to the partition function with $n$ fixed boundary conditions, is defined as the generating function
\begin{equation}\label{eq:GW_PF}
\mathcal{Z}_q^X(z_1,\ldots,z_n)=\sum_{d\geq 0}q^d\sum_{k_1,\ldots,k_n}\langle \tau_{k_1}(\omega)\ldots\tau_{k_n}(\omega)\rangle_d^Xz_1^{k_1+1}\ldots z_n^{k_n+1},
\end{equation}
where $q\in(0,1)$ is a coupling parameter and $z_1,\ldots,z_n$ are arbitrary coupling constants.

Okounkov and Pandharipande developed a  Gromov--Witten/Hurwitz correspondence, which can be summarized in the following statement: Gromov--Witten invariants with target a nonsingular algebraic curve can be expressed combinatorially through shifted symmetric functions $\mathbf p_k$ defined in~\eqref{eq:pk}. The precise formula of the correspondence for a curve $X$ of genus $g$ is \cite[eq. (0.25)]{OP2}:
\begin{equation}\label{eq:GWH_bracket}
\langle\tau_{k_1}(\omega)\ldots\tau_{k_n}(\omega)\rangle_d^X=\sum_{\lambda\vdash d}\left(\frac{f_\lambda}{d!}\right)^{2-2g}\frac{\mathbf{p}_{k_1+1}(\lambda)}{(k_1+1)!}\ldots \frac{\mathbf{p}_{k_n+1}(\lambda)}{(k_n+1)!},
\end{equation}
where $f_\lambda$ is the dimension of the representation of the symmetric group $S_d$ indexed by $\lambda.$

If the target is $\Tbb$, which is an elliptic curve, it has genus $g=1$ and the factor involving $f_\lambda/d!$ vanishes. Moreover, by~\eqref{eq:p_1p_2} is easy to check that the Hurwitz numbers enumerating ramified coverings on the torus with simple ramifications are particular Gromov--Witten invariants:
\[
\left\langle\tau_1(\omega)^{2m}\right\rangle_{n}^{\mathbb{T}}=H_1(n,2m).
\]

The new gauge/string duality we propose is the following: instead of writing the partition function as an integral over a Hurwitz space (Theorem~\ref{thm:main}), we show that its asymptotic expansion obtained in Theorem~\ref{thm:asympt_expansion} can be expressed in terms of the partition function $\mathcal{Z}_{q_t}^\Tbb$ defined in~\eqref{eq:GW_PF}.

\begin{corollary}[Gauge/string duality for YM on a torus]\label{cor:gauge_string}
Let $(G_N)$ be a sequence of compact classical groups of the same type $\tau\in\{A',A,B,C,D\}$ with $G_N\subset\GL_N(\C)$. For any $t>0$, there are explicit functionals $\{\Psi_{k,t}^\tau, k\geq 0\}$ of the partition function $\mathcal{Z}_{q_t}^\Tbb$ such that for any $p\geq 1$, as $N\to\infty$,
\begin{equation}
Z_{G_N}(1,t) = \sum_{k=0}^p\frac{\Psi_{k,t}^\tau(\mathcal{Z}_{q_t}^\Tbb)}{N^k}+O_t(N^{-p-1}).
\end{equation}
\end{corollary}

The duality is summarized in the following diagram, where the horizontal arrows describe approximate correspondences up to $O_t(N^{-p-1})$:
\[
\begin{tikzcd}
Z_{G_N}(1,t) \arrow[rr, leftrightsquigarrow, "Cor.~\ref{cor:gauge_string}"] \arrow[d, equal, "\text{Fourier}"'] & & \sum_{k=0}^p\frac{\Psi_{k,t}^\tau(\mathcal{Z}_{q_t}^\Tbb)}{N^k} \arrow[d, equal, "\text{GW/H}"]\\
\Tr(e^{\frac{t}{2}\Delta_{G_N}}) \arrow[rr, leftrightsquigarrow, "Thm.~\ref{thm:asympt_expansion}"] & & \sum_{k=0}^p\frac{a_k^\tau(t)}{N^k}
\end{tikzcd}
\]

If we notice that \eqref{eq:Migdal} identifies $Z_{G_N}(1,t)$ to $\Tr(e^{\frac{t}{2}\Delta_{G_N}})$, then Corollary~\ref{cor:gauge_string} is a mere rewriting of Theorem~\ref{thm:coef_Hurwitz}, whose proof is given at the end of the subsection. We conjecture that it admits a generalization to any genus $g\geq 2$; however, it is no longer only related to the expansion of the trace of the heat kernel, but also on the expansion of the Witten zeta function $\zeta_{G_N}$ \cite{Wit}, for which only partial results have been obtained \cite{Mag,Mag2,Lem3,Dah26}.

To conclude, let us note that physicists also conjectured that this duality goes beyond the partition function, and holds for Wilson loop expectations, which are integrals of traces of holonomies along loops in the base surface with respect to the Yang--Mills measure \cite{GT,GT2,CMR}. Up to date, the best rigorous results on Wilson loops have been obtained for the continuous Yang--Mills measure in the plane \cite{PPSY}, or the discrete measure in $\Z^d$ \cite{CPS}. The gauge/string duality is much less explicit, because the string side is described in terms of ``combinatorial" surfaces (i.e. maps) rather than ramified coverings or stable maps. The correspondence between the approaches of these works and ours has yet to be explored, but it is quite unlikely that these techniques of random maps can be easily extended to the continuous Yang--Mills measure on surfaces other than the plane, mainly because of nontrivial topological effects of the compact surfaces on all observables.

A first step is to express the coefficients of the asymptotic expansion~\eqref{eq:asympt_expansion} in terms of the generating function of Hurwitz numbers of base $\Tbb$ with $k$ generic ramifications $\mathcal{F}_{1,k}$. It is defined by
\begin{equation}
\mathcal{F}_{1,k}(q)=\sum_{n\geq 0}H_1(n,k)q^n.
\end{equation}
Let $\mathcal{D}=q\frac{d}{dq}$ be a homogeneous first-order differential operator on power series in $q$. Many observables of $q$-uniform random partitions can be rewritten as derivatives of this generating function: if $\alpha$ is a $q$-uniform random partition, for any $k\geq 0$ and any $m\geq 1$,
\begin{equation}\label{eq:GF_Hurwitz}
\phi(q)^{-1}\E[K(\alpha)^k\vert\alpha\vert^m]=\sum_{n\geq 0}n^mH_1(n,k)q^n=\mathcal{D}^m\mathcal{F}_{1,k}(q).
\end{equation}
In particular, $\mathcal{F}_{1,0}(q)=\sum_{n\geq 0}p(n)q^n=\phi(q)^{-1}$. The key observation to get the gauge/string duality is that the coefficients of~\eqref{eq:asympt_expansion} can be expressed in terms of $\mathcal{F}_{1,k}$. As we shall see, the proof of Theorem~\ref{thm:coef_Hurwitz} will follow.

\begin{lemma}\label{lem:coef_Hurwitz}
For any type $\tau\in\{A',A,B,C,D\}$ and any $t>0$, the coefficients $(a_k^\tau(t))$ of the expansion of the central heat trace can be expressed in terms of generating functions of Hurwitz numbers:
\begin{enumerate}
\item For $\tau\in\{A',A\}$,
\begin{equation}
a_{2k}^{\tau}(t) = \sum_{k_1+k_2+2k_3+2k_4=2k} \kappa_{k_1,k_2,k_3,k_4}^{\tau}(t)\mathcal{D}^{k_1}\mathcal{F}_{1,2k_3}(q_t)\mathcal{D}^{k_2}\mathcal{F}_{1,2k_4}(q_t),
\end{equation}
\item For $\tau\in\{B,C,D\}$,
\begin{equation}
a_k^\tau(t)=\sum_{k_1+2k_2=k}\kappa_{k_1,k_2}^\tau(t)\mathcal{D}^{k_1}\mathcal{F}_{1,2k_2}(q_t).
\end{equation}
\end{enumerate}
The coefficients $\kappa$ are explicit:
\begin{equation}
\kappa_{k_1,k_2,k_3,k_4}^{A'}(t)= \frac{(-1)^{k_2}t^{k_1+k_2+2k_3+2k_4}\mathcal{D}^{\frac{k_1+k_2}{2}}\theta(q_t)}{2^{k_1+k_2}k_1!k_2!(2k_3)!(2k_4)!},
\end{equation}
\begin{equation}
\kappa_{k_1,k_2,k_3,k_4}^A(t)= \frac{(k_1+k_2)!(-1)^{k_2}t^{\frac{k_1+k_2}{2}+2k_3+2k_4}}{2^{\frac{k_1+k_2}{2}}\frac{k_1+k_2}{2}!k_1!k_2!(2k_3)!(2k_4)!},
\end{equation}
\begin{equation}
\kappa_{k_1,k_2}^B(t) = \kappa_{k_1,k_2}^D(t) =  \frac{t^k}{2^{k_1}k_1!(2k_2)!},
\end{equation}
\begin{equation}
\kappa_{k_1,k_2}^C(t) = \frac{(-1)^{k_1}t^k}{2^{k_1}k_1!(2k_2)!}.
\end{equation}
\end{lemma}

\begin{proof}
The result follows from the expression of the coefficients obtained in the proof of Theorem~\ref{thm:asympt_expansion_detailed}, combined with \eqref{eq:GF_Hurwitz}. Let us detail the proof for $a_k^B(t)$ and leave the others to the reader. By the binomial theorem,
\begin{align*}
a_k^B(t)= & \sum_{k_1+k_2=k}\frac{(-t)^{k}}{2^kk_1!k_2!}\sum_{n=0}^\infty q_t^n(-n)^{k_1}\sum_{\mu\vdash n}(2K(\mu))^{k_2}\\
= & \sum_{k_1+k_2=k}\frac{(-1)^{k_2}t^{k}}{2^{k_1}k_1!k_2!}\sum_{n=0}^\infty q_t^n n^{k_1}H_1(n,k_2).
\end{align*}
As Hurwitz numbers $H_1(n,k_2)$ vanish for odd values of $k_2$, we can even restrict to $k_1+2k_2=k$ and replace $k_2$ by $2k_2$ in the expression. Then, using \eqref{eq:GF_Hurwitz}, we finally obtain
\[
a_k^B(t)= \sum_{k_1+2k_2=k}\frac{t^{k}}{2^{k_1}k_1!(2k_2)!}\mathcal{D}^{k_1}\mathcal{F}_{1,2k_2}(q_t),
\]
which is the expected formula.
\end{proof}

\begin{proof}[Proof of Theorem~\ref{thm:coef_Hurwitz}]
By Lemma~\ref{lem:coef_Hurwitz} we know that the coefficients are explicit functionals of generating functions of Hurwitz numbers. We conclude by using the fact that the generating function $\mathcal{F}_{1,k}$ can be easily expressed in terms of $\mathcal{Z}_q^\Tbb$:
\begin{equation}\label{eq:Hurwitz_GW_gf}
\mathcal{F}_{1,k}(q) = \frac{1}{2^{k}}\left(\frac{d}{dz_1}\right)^2\ldots \left(\frac{d}{dz_{k}}\right)^2\mathcal{Z}_q^\Tbb(z_1,\ldots,z_{k})\vert_{z_1=\ldots=z_{k}=0}.
\end{equation}
\end{proof}

\begin{remark}
All the results of this subsection can also be expressed in terms of another generating function, which appears in several papers about random partitions \cite{BO,EO,Oko} under the name $n$-point function. The interested reader can find a discussion in Appendix~\ref{sec:GF_appendix} on how this other generating function is related to $\mathcal{Z}_q^\Tbb$.
\end{remark}

\appendix

\section{Estimates on Hurwitz numbers and applications}

In this appendix, we give simple yet useful bounds on Hurwitz numbers $H_1(n,k)$ and their consequences in terms of integration, in particular the fact that $\rho_t$ has infinite mass and that the functions $\Phi_{t,N}^A$ and $\Phi_{t,N}^{A'}$ are not integrable with respect to $\rho_t$.

\begin{lemma}\label{prop:bound_Hurwitz}
For any $n\geq 1$ and $k\geq 0$,
\begin{equation}\label{eq:bound_Hurwitz}
4^{-k}(n-1)^{4k}\leq H_1(n,2k)\leq p(n)n^{4k}.
\end{equation}
\end{lemma}

\begin{proof}
Let us fix $n\geq 1$. For any $k\geq 0$,
\[
H_1(n,2k)=\sum_{\alpha\vdash n}K(\alpha)^{2k},
\]
and we apply \eqref{eq:bound_content_1} to get the upper bound. If we denote by $(n)$ the partition corresponding to a Young diagram with only one row of length $n$, we have $K((n))=n(n-1)/2,$ and we find that
\[
H_1(n,2k)\geq K((n))^{2k}=\left(\frac{n(n-1)}{2}\right)^{2k},
\]
which implies the lower bound.
\end{proof}

The following lemma states that the measure $\rho_t$ has infinite mass:

\begin{lemma}\label{lem:infinite_mass}
For any $p\geq 1$ and any $t>0$, the following sums are finite:
\[
\sum_{n=1}^p q_t^n\sum_{k=0}^\infty \frac{t^{2k}}{(2k)!}H_1(n,2k),\quad \sum_{k=0}^p \frac{t^{2k}}{(2k)!} \sum_{n=1}^\infty q_t^nH_1(n,2k).
\]
However, the double series
\[
\sum_{n\geq 1,k\geq 0}q_t^n\frac{t^{2k}}{(2k)!}H_1(n,2k)
\]
diverges.
\end{lemma}

\begin{proof}
The first assertion boils down to prove that if we fix $n$ or $k$, the sum over the other variable converges. For any $n\geq 1$, by the upper bound of Lemma \ref{prop:bound_Hurwitz},
\[
A_n:=\sum_{k\geq 0}\frac{t^{2k}}{(2k)!}H_1(n,2k)\leq p(n)\sum_{k\geq 0}\frac{(tn^2)^{2k}}{(2k)!}=p(n)\cosh(tn^2)<\infty.
\]
In fact, by using \eqref{eq:Hurwitz_torus} and Fubini's theorem, we can even compute the exact value of $A_n$:
\[
A_n=\sum_{k=0}^\infty\sum_{\lambda\vdash n}\frac{t^{2k}}{(2k)!}K(\lambda)^{2k}=\sum_{\lambda\vdash n}\cosh(tK(\lambda)).
\]
Likewise, if we fix $k\geq 0$, then
\[
B_k:=\sum_{n\geq 1}q_t^nH_1(n,2k)\leq \sum_{n\geq 1}e^{-\frac{t}{2}n}p(n)n^{4k},
\]
and the RHS also converges because $p(n)\leq e^{\pi\sqrt{2n/3}}$ for $n$ large enough (see \cite{Apo} for instance).

However, as $n$ tends to infinity, we see that $q_t^nA_n$ diverges. By the lower bound of Proposition \ref{prop:bound_Hurwitz},
\[
A_n = \sum_{k\geq 0}\frac{t^{2k}}{(2k)!}H_1(n,2k)\geq \sum_{k\geq 0}\frac{1}{(2k)!}\left(\frac{t}{4}\left(n-1\right)^2\right)^{2k}=\cosh\left(\frac{t}{4}(n-1)^2\right),
\]
and
\[
e^{-\frac{t}{2}n}A_n \geq \frac12 e^{-\frac{t}{2}n}\left(e^{\frac{t}{4}(n-1)^2}+e^{-\frac{t}{4}(n-1)^2}\right)\geq \frac12 e^{\frac{t}{4}(n^2-4n+1)}.
\]
The RHS does not converge to zero (it even diverges as $n\to\infty$), therefore it is trivially not summable.
\end{proof}

We now wonder about integrability the coupling functions appearing in the random surface representation in the unitary and special unitary cases:

\begin{proposition}[Non-integrability of $\Phi_{t,N}^{A'}$ and $\Phi_{t,N}^A$]\label{prop:non_integrability_PhiA}
For any $p,N\geq 1$ and any $\tau\in\{A',A\}$, with $\Phi_{t,N}^{\tau}$ defined in \eqref{def:phiA'}  and \eqref{def:phiA}, the following integral diverges:
\begin{align*}
I_p^\tau(N)=\int_{\mathcal{R}^2}\mathbf{1}_{\{\chi(X_1)+\chi(X_2)\geq -2p\}} \Phi_{t,N}^{\tau}(X_1,X_2)N^{\chi(X_1)+\chi(X_2)}d\rho_t^{\otimes 2}([X_1],[X_2]).
\end{align*}
\end{proposition}

\begin{proof}
{\bf Case of $\U(N)$.}  Let us fix $p\geq 1$. Using \eqref{eq:bound_Hurwitz}, we get
\begin{align*}
I_p^{A'}(N)= & \frac{1}{2}\sum_{k_1+k_2\leq p} \sum_{n_1,n_2\geq 1}e^{-\frac{t}{2}(n_1+n_2)} \sum_{n \in \Z} e^{-\frac{t}{2}n^2}\left(e^{ \frac{tn}{N}(n_1-n_2)} + e^{-\frac{tn}{N}(n_1-n_2)}\right)\\
& \times N^{-2k_1-2k_2}H_1(n_1,2k_1)H_1(n_2,2k_2)\\
\geq & \sum_{k_1+k_2\leq p} \sum_{n_1 \geq n_2\geq 1}e^{-\frac{t}{2}(n_1+n_2)}e^{-\frac{tN^2}{2} + t(n_1 - n_2)} N^{-2k_1-2k_2}4^{-k_1-k_2}(n_1-1)^{4k_1}(n_2-1)^{4k_2}.
\end{align*}
and it is clear that for any fixed $k_1,k_2,N$,
\[ \sum_{n_1 \geq n_2\geq 1}e^{-\frac{t}{2}(n_1+n_2)}e^{-\frac{tN^2}{2} + t(n_1 - n_2)}  \times N^{-2k_1-2k_2}4^{-k_1-k_2}(n_1-1)^{4k_1}(n_2-1)^{4k_2}\]
diverges, so that it is obviously the case for $I_p^{A'}(N)$.

{\bf Case of $\SU(N)$.} Let us fix $p\geq 1$. Similarly, using \eqref{eq:bound_Hurwitz}, we get
\begin{align*}
I_p^A(N)= & \sum_{k_1+k_2\leq p} \sum_{n_1,n_2\geq 1}e^{-\frac{t}{2}(n_1+n_2)+\frac{t^2}{2N^2}(n_1-n_2)^2} N^{-2k_1-2k_2}H_1(n_1,2k_1)H_1(n_2,2k_2)\\
\geq & \sum_{k_1+k_2\leq p} \sum_{n_1,n_2\geq 1}e^{-\frac{t}{2}(n_1+n_2)+\frac{t^2}{2N^2}(n_1-n_2)^2} N^{-2k_1-2k_2}4^{-k_1-k_2}(n_1-1)^{4k_1}(n_2-1)^{4k_2}.
\end{align*}
However, it is clear that for any fixed $k_1,k_2,N$,
\begin{align*}
\sum_{n_1,n_2\geq 1}e^{-\frac{t}{2}(n_1+n_2)+\frac{t^2}{2N^2}(n_1-n_2)^2}(n_1-1)^{4k_1}(n_2-1)^{4k_2}
\end{align*}
diverges, so that it is obviously the case for $I_p^A(N)$.
\end{proof}

\section{Gromov--Witten invariants}\label{sec:GW}

This appendix is aimed to provide a self-contained definition of Gromov--Witten invariants, which are well-known objects for algebraic geometers but not necessarily to other readers who might be interested in the paper. We will keep it short, and hope that it could serve as a reading guide for the already overwhelming literature on the subject.

\begin{definition}
Let $X$ be a smooth projective curve, and $g,n\geq 0$ and $d>0$ be fixed integers. The \emph{moduli space of stable maps} $\overline{\mathcal{M}}_{g,n}(X,d)$ is the set of isomorphism classes of pointed algebraic maps $\pi:C\to X$ of degree $d$, where $C$ is a nonsingular curve of genus $g$ with $n$ distinct marked points $p_1,\ldots,p_n$.
\end{definition}

This moduli space has been introduced by Kontsevich \cite{Kon95}, and we refer to \cite{HKKPRV} for a general overview in the context of mirror symmetry. In the setting we care about, $X$ is a compact Riemann surface, $\pi:C\to X$ is a ramified covering of $X$ of degree $d$ with ramification locus $\{\pi(p_1),\ldots,\pi(p_n)\}$, and the isomorphism is described by conjugation of monodromies around the points $\pi(p_i)$. We will denote the stable map by $(C,p_1,\ldots,p_n,\pi)$. There are natural morphisms
\[
\ev_i:\left\lbrace\begin{array}{ccc}
\overline{\mathcal{M}}_{g,n}(X,d) & \longrightarrow & X\\
(C,p_1,\ldots,p_n,\pi) & \longmapsto & \pi(p_i),
\end{array}\right.
\]
which are the evaluations of the map at the marked points. In the setting of ramified coverings, they correspond to the projections onto the points of the ramification locus of $X$. The cotangent line bundles $L_1,\ldots,L_n\to\overline{\mathcal{M}}_{g,n}(X,d)$ are line bundles over the moduli space where the fiber of $L_i$ above the point $(C,p_1,\ldots,p_n,\pi)$ is the cotangent line $T_{p_i}C$. The most important cohomology classes in Gromov--Witten theory are $\psi$-classes and $\tau$-classes:
\begin{itemize}
\item we denote the $\psi$-classes by $\psi_i\in H^2(\overline{\mathcal{M}}_{g,n}(X,d),\Q)$ the first Chern class of the cotangent line bundle $L_i\to\overline{\mathcal{M}}_{g,n}(X,d)$:
\[
\psi_i = c_1(L_i).
\]
\item we denote the $\tau$-classes by $\tau_{k_1},\ldots,\tau_{k_n}$ defined by  $\tau_{k_i}=\psi_i^{k_i}\cup\ev_i^*(\gamma)$, for any $\gamma\in H^*(X,\Q)$, where $\cup$ denotes the cup product of cohomology classes. 
\end{itemize}
For the sake of simplicity, we shall rather write cup products by simple products in the sequel.

\begin{definition}
Let $X$ be a smooth projective curve, and $g,n\geq 0$ and $d>0$ be fixed integers. For any $k_1,\ldots,k_n\geq 0$, the \emph{Gromov--Witten invariant} $\langle\tau_{k_1}(\omega)\ldots\tau_{k_n}(\omega)\rangle_d^X$ is the integral\footnote{The integration is done with respect to the virtual fundamental class $[\overline{\mathcal{M}}_{g,n}(X,d)]^{\mathrm{vir}}$, which is the canonical equivalent of the fundamental class of a manifold. See \cite[\S 26.1]{HKKPRV} for further details.}
\begin{equation}\label{eq:stationary_GW}
\left\langle \tau_{k_1}(\omega)\ldots\tau_{k_n}(\omega)\right\rangle_d^X := \int_{\left[\overline{\mathcal{M}}_{g,n}(X,d)\right]^{\mathrm{vir}}}\prod_{i=1}^n \psi_i^{k_i}\ev_i^*(\omega),
\end{equation}
where $\omega\in H^2(X,\Q)$ is the Poincar\'e dual of the class of a point.
\end{definition}

\begin{remark}
In our definitions, we do not require that the curve $C$ is connected. See \cite{OP2} for a discussion about the connected vs disconnected Gromov--Witten theories.
\end{remark}

\section{Generating functions of integer partitions and the Bloch--Okounkov theorem}\label{sec:GF_appendix}

For any formal parameter $q$, let $F_q(t_1,\ldots,t_n)$ be the following generating function, called $n$-\emph{point function}:
\begin{equation}\label{eq:npoint}
F_q(t_1,\ldots,t_n)=\sum_{n\geq 1}q^n\sum_{\lambda\vdash n}\prod_{k=1}^n\sum_{i=1}^\infty t_k^{\lambda_i+\frac12-i}.
\end{equation}
It converges if $1<\vert t_i\vert$ and $\vert t_1\ldots t_n\vert<\frac1q$, and admits a meromorphic extension on $\C^n$. A first easy result is a correspondence between $F_q$ and $\mathcal{Z}_q^\Tbb$.

\begin{proposition}\label{prop:corresp_GF}
For any $q\in(0,1)$ and $z_1,\ldots,z_n\in\R$,
\begin{equation}
F_q(e^{z_1},\ldots,e^{z_n})=\mathcal{Z}_q^\Tbb(z_1,\ldots,z_n).
\end{equation}
\end{proposition}

\begin{proof}
Fix $q\in(0,1)$ and $z_1,\ldots,z_n\in\R$. We have
\[
F_{q_t}(e^{z_1},\ldots,e^{z_n})=\sum_{d\geq 0}q_t^d\sum_{\lambda\vdash d}\mathbf{e}(\lambda,z_1)\ldots\mathbf{e}(\lambda,z_n).
\]
For any $d\geq 0$, combining \eqref{eq:GWH_bracket} and \eqref{eq:exp_pk} gives
\[
\sum_{k_1,\ldots,k_n}\langle\tau_{k_1}(\omega)\ldots\tau_{k_n}(\omega)\rangle_d^\Tbb z_1^{k_1+1}\ldots z_n^{k_n+1}=\sum_{\lambda\vdash d}\mathbf{e}(\lambda,z_1)\ldots\mathbf{e}(\lambda,z_n),
\]
and the result follows.
\end{proof}
The generating function $\mathcal{F}_{1,k}$, which is expressed in terms of $\mathcal{Z}_q^\Tbb$, should also have an expression in terms of $F_q$. In fact, according to \cite[0.18]{BO}, we have
\begin{equation}
\mathbf{p}_k(\lambda)-(2^{-k}-1)\zeta(-k)=\left(t\frac{d}{dt}\right)^k\left.\left(\sum_{i=1}^\infty t^{\lambda_i+\frac12-i}-\frac{1}{\log t}\right)\right\vert_{t=1}
\end{equation}
In particular,
\[
\mathbf{p}_2(\lambda)=\left(t\frac{d}{dt}\right)^2\left.\left(\sum_{i=1}^\infty t^{\lambda_i+\frac12-i}-\frac{1}{\log t}\right)\right\vert_{t=1},
\]
because $-2k$ is a trivial zero of the Riemann zeta function, and we have
\begin{align*}
H_1(n,k)= & 2^{-k}\sum_{\lambda\vdash n}\mathbf{p}_{2}(\lambda)\ldots \mathbf{p}_{2}(\lambda)\\
= & 2^{-k}\left(t_1\frac{d}{dt_1}\right)^2\ldots\left(t_m\frac{d}{dt_k}\right)^2\left.\sum_{\lambda\vdash n}\left(\prod_{\ell=1}^k\sum_{i=1}^\infty t_\ell^{\lambda_i+\frac12-i}-\frac{1}{\log t_\ell}\right)\right\vert_{t_1=\ldots=t_k=1}.
\end{align*}
Expanding the product over $\ell\in\{1,\ldots,k\}$ and putting the result in the generating function $\mathcal{F}_{1,k}(q)$ clearly produces an expression of the generating function of Hurwitz numbers in terms of $F_q$, but the exact formula is rather unpleasant compared to~\eqref{eq:Hurwitz_GW_gf}. There is, however, an important merit of using $F_q$ instead of $\mathcal{Z}_q^\Tbb$: according to the Bloch--Okounkov theorem \cite[Theorem 6.1]{BO}, there is an explicit determinantal formula to compute $F_q(t_1,\ldots,t_k)$ in terms of derivatives of the Theta function
\[
\Theta(x;q)=\eta^{-3}(q)\sum_{n\in\Z}(-1)^nq^{\frac{(n+1/2)^2}{2}}x^{n+1/2}.
\]
It ensures that we are able to compute recursively any coefficient of the asymptotic expansion of Theorem~\ref{thm:asympt_expansion}. It also has an interpretation in terms of quasimodular forms, see \cite{Dij,OP2}.

\bibliographystyle{plain}
\bibliography{heat_trace}

@article {Agg,
    AUTHOR = {Aggarwal, Amol},
     TITLE = {Large genus asymptotics for volumes of strata of {A}belian
              differentials},
      NOTE = {With an appendix by Anton Zorich},
   JOURNAL = {J. Amer. Math. Soc.},
  FJOURNAL = {Journal of the American Mathematical Society},
    VOLUME = {33},
      YEAR = {2020},
    NUMBER = {4},
     PAGES = {941--989},
      ISSN = {0894-0347},
   MRCLASS = {32G15 (05A16 37D40 37P45)},
  MRNUMBER = {4155217},
MRREVIEWER = {Zongliang Sun},
       DOI = {10.1090/jams/947},
       URL = {https://doi.org/10.1090/jams/947},
}

@article {AKS,
    AUTHOR = {Aharony, Ofer and Kundu, Suman and Sheaffer, Tal},
     TITLE = {A string theory for two dimensional {Y}ang-{M}ills theory.
              {P}art {I}},
   JOURNAL = {J. High Energy Phys.},
  FJOURNAL = {Journal of High Energy Physics},
      YEAR = {2024},
    NUMBER = {7},
     PAGES = {Paper No. 63, 56},
      ISSN = {1126-6708},
   MRCLASS = {81T13 (81T30 81T40)},
  MRNUMBER = {4771420},
       DOI = {10.1007/jhep07(2024)063},
       URL = {https://doi.org/10.1007/jhep07(2024)063},
}

@book {Apo,
    AUTHOR = {Apostol, Tom M.},
     TITLE = {Introduction to analytic number theory},
    SERIES = {Undergraduate Texts in Mathematics},
 PUBLISHER = {Springer-Verlag, New York-Heidelberg},
      YEAR = {1976},
     PAGES = {xii+338},
   MRCLASS = {10-01 (10AXX 10HXX)},
  MRNUMBER = {434929},
MRREVIEWER = {E. Grosswald},
}

@article {App,
    AUTHOR = {Applebaum, David},
     TITLE = {Infinitely divisible central probability measures on compact
              {L}ie groups---regularity, semigroups and transition kernels},
   JOURNAL = {Ann. Probab.},
  FJOURNAL = {The Annals of Probability},
    VOLUME = {39},
      YEAR = {2011},
    NUMBER = {6},
     PAGES = {2474--2496},
      ISSN = {0091-1798},
   MRCLASS = {60B15 (35K08 43A30 47D07 60E07 60G51)},
  MRNUMBER = {2932674},
MRREVIEWER = {Wilfried Hazod},
       DOI = {10.1214/10-AOP604},
       URL = {https://doi.org/10.1214/10-AOP604},
}

@article{AM1,
      title={Friedman-{R}amanujan functions in random hyperbolic geometry and application to spectral gaps}, 
      author={Nalini Anantharaman and Laura Monk},
      year={2025},
      eprint={2304.02678},
      archivePrefix={arXiv},
      primaryClass={math.SP},
      url={https://arxiv.org/abs/2304.02678}, 
}

@article{AM2,
      title={Friedman-{R}amanujan functions in random hyperbolic geometry and application to spectral gaps {II}}, 
      author={Nalini Anantharaman and Laura Monk},
      year={2025},
      eprint={2502.12268},
      archivePrefix={arXiv},
      primaryClass={math.MG},
      url={https://arxiv.org/abs/2502.12268}, 
}

@article{BenTro,
    author = "Benizri, Lior and Troost, Jan",
    title = "{The string dual to two-dimensional Yang-Mills theory}",
    eprint = "2502.02662",
    archivePrefix = "arXiv",
    primaryClass = "hep-th",
    doi = "10.1007/JHEP08(2025)017",
    journal = "JHEP",
    volume = "08",
    pages = "017",
    year = "2025"
}

@misc{BCSK,
      title={Surface sums for lattice {Y}ang-{M}ills in the large-{$N$} limit}, 
      author={Jacopo Borga and Sky Cao and Jasper Shogren-Knaak},
      year={2024},
      eprint={2411.11676},
      archivePrefix={arXiv},
      primaryClass={math.PR},
      url={https://arxiv.org/abs/2411.11676}, 
}

@article {BO,
    AUTHOR = {Bloch, Spencer and Okounkov, Andrei},
     TITLE = {The character of the infinite wedge representation},
   JOURNAL = {Adv. Math.},
  FJOURNAL = {Advances in Mathematics},
    VOLUME = {149},
      YEAR = {2000},
    NUMBER = {1},
     PAGES = {1--60},
      ISSN = {0001-8708},
   MRCLASS = {11F22 (17B66)},
  MRNUMBER = {1742353},
MRREVIEWER = {Mirko Primc},
       DOI = {10.1006/aima.1999.1845},
       URL = {https://doi.org/10.1006/aima.1999.1845},
}

@article {BOO,
    AUTHOR = {Borodin, Alexei and Okounkov, Andrei and Olshanski, Grigori},
     TITLE = {Asymptotics of {P}lancherel measures for symmetric groups},
   JOURNAL = {J. Amer. Math. Soc.},
  FJOURNAL = {Journal of the American Mathematical Society},
    VOLUME = {13},
      YEAR = {2000},
    NUMBER = {3},
     PAGES = {481--515},
      ISSN = {0894-0347},
   MRCLASS = {05E10 (33C10 60C05)},
  MRNUMBER = {1758751},
MRREVIEWER = {A. N. Philippou},
       DOI = {10.1090/S0894-0347-00-00337-4},
       URL = {https://doi.org/10.1090/S0894-0347-00-00337-4},
}

@article {BJSM23,
    AUTHOR = {Bringmann, Kathrin and Jennings-Shaffer, Chris and Mahlburg,
              Karl},
     TITLE = {On a {T}auberian theorem of {I}ngham and {E}uler-{M}aclaurin
              summation},
   JOURNAL = {Ramanujan J.},
  FJOURNAL = {Ramanujan Journal. An International Journal Devoted to the
              Areas of Mathematics Influenced by Ramanujan},
    VOLUME = {61},
      YEAR = {2023},
    NUMBER = {1},
     PAGES = {55--86},
      ISSN = {1382-4090,1572-9303},
   MRCLASS = {11M45 (11B68 11P82)},
  MRNUMBER = {4584516},
MRREVIEWER = {Driss\ Essouabri},
       DOI = {10.1007/s11139-020-00377-5},
       URL = {https://doi.org/10.1007/s11139-020-00377-5},
}

@article{CPS,
      title={Random surfaces and lattice {Y}ang-{M}ills}, 
      author={Sky Cao and Minjae Park and Scott Sheffield},
      year={2025},
      journal={Comm. Amer. Math. Soc.},
      note={to appear}
}

@article {Car86,
    AUTHOR = {Cardy, John L.},
     TITLE = {Operator content of two-dimensional conformally invariant
              theories},
   JOURNAL = {Nuclear Phys. B},
  FJOURNAL = {Nuclear Physics. B. Theoretical, Phenomenological, and
              Experimental High Energy Physics. Quantum Field Theory and
              Statistical Systems},
    VOLUME = {270},
      YEAR = {1986},
    NUMBER = {2},
     PAGES = {186--204},
      ISSN = {0550-3213,1873-1562},
   MRCLASS = {17B65 (81E08 82A25)},
  MRNUMBER = {845940},
MRREVIEWER = {Andrew\ Pressley},
       DOI = {10.1016/0550-3213(86)90552-3},
       URL = {https://doi.org/10.1016/0550-3213(86)90552-3},
}

@article {Cha19,
    AUTHOR = {Chatterjee, Sourav},
     TITLE = {Rigorous solution of strongly coupled {$SO(N)$} lattice gauge
              theory in the large {$N$} limit},
   JOURNAL = {Comm. Math. Phys.},
  FJOURNAL = {Communications in Mathematical Physics},
    VOLUME = {366},
      YEAR = {2019},
    NUMBER = {1},
     PAGES = {203--268},
      ISSN = {0010-3616,1432-0916},
   MRCLASS = {81T25 (81T13)},
  MRNUMBER = {3919447},
MRREVIEWER = {Svetoslav\ Zahariev},
       DOI = {10.1007/s00220-019-03353-3},
       URL = {https://doi.org/10.1007/s00220-019-03353-3},
}

@article {CMR,
    AUTHOR = {Cordes, Stefan and Moore, Gregory and Ramgoolam, Sanjaye},
     TITLE = {Large {$N$} {$2$}{D} {Y}ang-{M}ills theory and topological
              string theory},
   JOURNAL = {Comm. Math. Phys.},
  FJOURNAL = {Communications in Mathematical Physics},
    VOLUME = {185},
      YEAR = {1997},
    NUMBER = {3},
     PAGES = {543--619},
      ISSN = {0010-3616},
   MRCLASS = {58D29 (55N91 58D27 81T40 81T70)},
  MRNUMBER = {1463054},
MRREVIEWER = {Martin Schlichenmaier},
       DOI = {10.1007/s002200050102},
       URL = {https://doi.org/10.1007/s002200050102},
}

@article {DL,
    AUTHOR = {Dahlqvist, Antoine and Lemoine, Thibaut},
     TITLE = {Large {N} limit of {Y}ang--{M}ills partition function and
              {W}ilson loops on compact surfaces},
   JOURNAL = {Probab. Math. Phys.},
  FJOURNAL = {Probability and Mathematical Physics},
    VOLUME = {4},
      YEAR = {2023},
    NUMBER = {4},
     PAGES = {849--890},
      ISSN = {2690-0998},
   MRCLASS = {60B15 (46L54 57R56 60B20 60D05 60K35 81T13 81T32)},
  MRNUMBER = {4672112},
       DOI = {10.2140/pmp.2023.4.849},
       URL = {https://doi.org/10.2140/pmp.2023.4.849},
}

@misc{Dah26,
      title={Large N limit of Wilson Loops on orientable closed surfaces in the light of Koike-Schur-Weyl duality and Spin Networks}, 
      author={Antoine Dahlqvist},
      year={2026},
      eprint={2603.11374},
      archivePrefix={arXiv},
      primaryClass={math.PR},
      url={https://arxiv.org/abs/2603.11374}, 
}

@incollection {Dij,
    AUTHOR = {Dijkgraaf, Robbert},
     TITLE = {Mirror symmetry and elliptic curves},
 BOOKTITLE = {The moduli space of curves ({T}exel {I}sland, 1994)},
    SERIES = {Progr. Math.},
    VOLUME = {129},
     PAGES = {149--163},
 PUBLISHER = {Birkh\"{a}user Boston, Boston, MA},
      YEAR = {1995},
   MRCLASS = {14N10 (11F03 14H52 32G81)},
  MRNUMBER = {1363055},
MRREVIEWER = {Bruce Hunt},
}

@incollection {Dou,
    AUTHOR = {Douglas, Michael R.},
     TITLE = {Conformal field theory techniques in large {$N$}
              {Y}ang-{M}ills theory},
 BOOKTITLE = {Quantum field theory and string theory ({C}arg\`ese, 1993)},
    SERIES = {NATO Adv. Sci. Inst. Ser. B: Phys.},
    VOLUME = {328},
     PAGES = {119--135},
 PUBLISHER = {Plenum, New York},
      YEAR = {1995},
   MRCLASS = {81T40},
  MRNUMBER = {1322520},
}

@book {Dur70,
    AUTHOR = {Duren, Peter L.},
     TITLE = {Theory of {$H\sp{p}$} spaces},
    SERIES = {Pure and Applied Mathematics, Vol. 38},
 PUBLISHER = {Academic Press, New York-London},
      YEAR = {1970},
     PAGES = {xii+258},
   MRCLASS = {46.30 (30.00)},
  MRNUMBER = {268655},
MRREVIEWER = {D.\ Sarason},
}

@article {EO,
    AUTHOR = {Eskin, Alex and Okounkov, Andrei},
     TITLE = {Asymptotics of numbers of branched coverings of a torus and
              volumes of moduli spaces of holomorphic differentials},
   JOURNAL = {Invent. Math.},
  FJOURNAL = {Inventiones Mathematicae},
    VOLUME = {145},
      YEAR = {2001},
    NUMBER = {1},
     PAGES = {59--103},
      ISSN = {0020-9910},
   MRCLASS = {32G15 (05A17 11F23 37A25 57M12)},
  MRNUMBER = {1839286},
MRREVIEWER = {Christopher M. Judge},
       DOI = {10.1007/s002220100142},
}

@book {Eyn,
    AUTHOR = {Eynard, Bertrand},
     TITLE = {Counting surfaces},
    SERIES = {Progress in Mathematical Physics},
    VOLUME = {70},
      NOTE = {CRM Aisenstadt chair lectures},
 PUBLISHER = {Birkh\"{a}user/Springer, [Cham]},
      YEAR = {2016},
     PAGES = {xvii+414},
      ISBN = {978-3-7643-8796-9; 978-3-7643-8797-6},
   MRCLASS = {81-02 (05C10 05C30 14H60 30Fxx 57M50 81T20 81T30)},
  MRNUMBER = {3468847},
MRREVIEWER = {Daniel David Moskovich},
       DOI = {10.1007/978-3-7643-8797-6},
}

@article {EMS11,
    AUTHOR = {Eynard, Bertrand and Mulase, Motohico and Safnuk, Bradley},
     TITLE = {The {L}aplace transform of the cut-and-join equation and the
              {B}ouchard-{M}ari\~{n}o conjecture on {H}urwitz numbers},
   JOURNAL = {Publ. Res. Inst. Math. Sci.},
  FJOURNAL = {Publications of the Research Institute for Mathematical
              Sciences},
    VOLUME = {47},
      YEAR = {2011},
    NUMBER = {2},
     PAGES = {629--670},
      ISSN = {0034-5318,1663-4926},
   MRCLASS = {14N10 (14H10 14H30 14N35 81T45)},
  MRNUMBER = {2849645},
MRREVIEWER = {Hsian-Hua\ Tseng},
       DOI = {10.2977/PRIMS/47},
       URL = {https://doi.org/10.2977/PRIMS/47},
}

@article {Feg,
    AUTHOR = {Fegan, H. D.},
     TITLE = {The heat equation and modular forms},
   JOURNAL = {J. Differential Geometry},
  FJOURNAL = {Journal of Differential Geometry},
    VOLUME = {13},
      YEAR = {1978},
    NUMBER = {4},
     PAGES = {589--602 (1979)},
      ISSN = {0022-040X},
   MRCLASS = {22E30 (10A20 10D15 58G11 58G25)},
  MRNUMBER = {570220},
MRREVIEWER = {M. F. Atiyah},
}

@book {FegBook,
    AUTHOR = {Fegan, Howard D.},
     TITLE = {Introduction to compact {L}ie groups},
    SERIES = {Series in Pure Mathematics},
    VOLUME = {13},
 PUBLISHER = {World Scientific Publishing Co., Inc., River Edge, NJ},
      YEAR = {1991},
     PAGES = {xiv+131},
      ISBN = {981-02-0702-6},
   MRCLASS = {22-01 (22E30)},
  MRNUMBER = {1134781},
MRREVIEWER = {A.\ U.\ Klimyk},
       DOI = {10.1142/1436},
       URL = {https://doi.org/10.1142/1436},
}

@book{FellerII,
  author    = {Feller, William},
  title     = {An Introduction to Probability Theory and Its Applications. Vol. II},
  edition   = {Second},
  publisher = {John Wiley \& Sons},
  address   = {New York},
  year      = {1971}
}

@book {FlaSed09,
    AUTHOR = {Flajolet, Philippe and Sedgewick, Robert},
     TITLE = {Analytic combinatorics},
 PUBLISHER = {Cambridge University Press, Cambridge},
      YEAR = {2009},
     PAGES = {xiv+810},
      ISBN = {978-0-521-89806-5},
   MRCLASS = {05-02 (05A15 05A16 60C05 60E10 82-01)},
  MRNUMBER = {2483235},
       DOI = {10.1017/CBO9780511801655},
       URL = {https://doi.org/10.1017/CBO9780511801655},
}

@article {FGKST16,
    AUTHOR = {Folsom, Amanda and Garthwaite, Sharon and Kang, Soon-Yi and
              Swisher, Holly and Treneer, Stephanie},
     TITLE = {Quantum mock modular forms arising from eta-theta functions},
   JOURNAL = {Res. Number Theory},
  FJOURNAL = {Research in Number Theory},
    VOLUME = {2},
      YEAR = {2016},
     PAGES = {Paper No. 14, 41},
      ISSN = {2522-0160,2363-9555},
   MRCLASS = {11F27 (33D15)},
  MRNUMBER = {3534170},
MRREVIEWER = {\c{C}etin\ \"{U}rti\c{s}},
       DOI = {10.1007/s40993-016-0045-7},
       URL = {https://doi.org/10.1007/s40993-016-0045-7},
}

@article{Gou94,
  author  = {Goulden, I. P.},
  title   = {A differential operator for symmetric functions and the combinatorics of multiplying transpositions},
  journal = {Transactions of the American Mathematical Society},
  volume  = {344},
  number  = {1},
  pages   = {421--440},
  year    = {1994},
  doi     = {10.1090/S0002-9947-1994-1249468-3}
}

@article {Gro,
    AUTHOR = {Gross, David J.},
     TITLE = {Two-dimensional {QCD} as a string theory},
   JOURNAL = {Nuclear Phys. B},
  FJOURNAL = {Nuclear Physics. B. Theoretical, Phenomenological, and
              Experimental High Energy Physics. Quantum Field Theory and
              Statistical Systems},
    VOLUME = {400},
      YEAR = {1993},
    NUMBER = {1-3},
     PAGES = {161--180},
      ISSN = {0550-3213},
   MRCLASS = {81T40 (81T30)},
  MRNUMBER = {1227259},
MRREVIEWER = {Emil Vin\c{t}eler},
       DOI = {10.1016/0550-3213(93)90402-B},
}

@article {GT,
    AUTHOR = {Gross, David J. and Taylor, IV, Washington},
     TITLE = {Two-dimensional {QCD} is a string theory},
   JOURNAL = {Nuclear Phys. B},
  FJOURNAL = {Nuclear Physics. B. Theoretical, Phenomenological, and
              Experimental High Energy Physics. Quantum Field Theory and
              Statistical Systems},
    VOLUME = {400},
      YEAR = {1993},
    NUMBER = {1-3},
     PAGES = {181--208},
      ISSN = {0550-3213},
   MRCLASS = {81T40 (81T30)},
  MRNUMBER = {1227260},
MRREVIEWER = {Emil Vin\c{t}eler},
       DOI = {10.1016/0550-3213(93)90403-C},
}

@article {GT2,
    AUTHOR = {Gross, David J. and Taylor, IV, Washington},
     TITLE = {Twists and {W}ilson loops in the string theory of
              two-dimensional {QCD}},
   JOURNAL = {Nuclear Phys. B},
  FJOURNAL = {Nuclear Physics. B. Theoretical, Phenomenological, and
              Experimental High Energy Physics. Quantum Field Theory and
              Statistical Systems},
    VOLUME = {403},
      YEAR = {1993},
    NUMBER = {1-2},
     PAGES = {395--449},
      ISSN = {0550-3213},
   MRCLASS = {81T40},
  MRNUMBER = {1232625},
MRREVIEWER = {Emil Vin\c{t}eler},
       DOI = {10.1016/0550-3213(93)90042-N},
}

@article {HM23,
    AUTHOR = {Hide, Will and Magee, Michael},
     TITLE = {Near optimal spectral gaps for hyperbolic surfaces},
   JOURNAL = {Ann. of Math. (2)},
  FJOURNAL = {Annals of Mathematics. Second Series},
    VOLUME = {198},
      YEAR = {2023},
    NUMBER = {2},
     PAGES = {791--824},
      ISSN = {0003-486X},
   MRCLASS = {58J50 (05C50 05C80)},
  MRNUMBER = {4635304},
MRREVIEWER = {J\'{o}zef Dodziuk},
       DOI = {10.4007/annals.2023.198.2.6},
}

@misc{HMT,
      title={Spectral gap with polynomial rate for random covering surfaces}, 
      author={Will Hide and Davide Macera and Joe Thomas},
      year={2025},
      eprint={2505.08479},
      archivePrefix={arXiv},
      primaryClass={math.SP},
      url={https://arxiv.org/abs/2505.08479}, 
}

@article{HMT2,
      title={Spectral gap with polynomial rate for {W}eil-{P}etersson random surfaces}, 
      author={Will Hide and Davide Macera and Joe Thomas},
      year={2025},
      eprint={2508.14874},
      archivePrefix={arXiv},
      primaryClass={math.SP},
      url={https://arxiv.org/abs/2508.14874}, 
}

@article{Hoo74,
    author = "'t Hooft, Gerard",
    editor = "Taylor, J. C.",
    title = "{A Planar Diagram Theory for Strong Interactions}",
    reportNumber = "CERN-TH-1786",
    doi = "10.1016/0550-3213(74)90154-0",
    journal = "Nucl. Phys. B",
    volume = "72",
    pages = "461",
    year = "1974"
}

@book {HKKPRV,
    AUTHOR = {Hori, Kentaro and Katz, Sheldon and Klemm, Albrecht and
              Pandharipande, Rahul and Thomas, Richard and Vafa, Cumrun and
              Vakil, Ravi and Zaslow, Eric},
     TITLE = {Mirror symmetry},
    SERIES = {Clay Mathematics Monographs},
    VOLUME = {1},
      NOTE = {With a preface by Vafa},
 PUBLISHER = {American Mathematical Society, Providence, RI; Clay
              Mathematics Institute, Cambridge, MA},
      YEAR = {2003},
     PAGES = {xx+929},
      ISBN = {0-8218-2955-6},
   MRCLASS = {14J32 (14N35 32Q25 81T30 81T60)},
  MRNUMBER = {2003030},
MRREVIEWER = {Marcos Mari\~{n}o},
}

@article {Ing41,
    AUTHOR = {Ingham, A. E.},
     TITLE = {A {T}auberian theorem for partitions},
   JOURNAL = {Ann. of Math. (2)},
  FJOURNAL = {Annals of Mathematics. Second Series},
    VOLUME = {42},
      YEAR = {1941},
     PAGES = {1075--1090},
      ISSN = {0003-486X},
   MRCLASS = {10.0X},
  MRNUMBER = {5522},
MRREVIEWER = {H.\ S.\ Zuckerman},
       DOI = {10.2307/1970462},
       URL = {https://doi.org/10.2307/1970462},
}

@misc{Jaf16,
      title={Wilson loop expectations in {$SU(N)$} lattice gauge theory}, 
      author={Jafar Jafarov},
      year={2016},
      eprint={1610.03821},
      archivePrefix={arXiv},
      primaryClass={math.PR},
      url={https://arxiv.org/abs/1610.03821},
      note={arxiv:1610.03821}
}

@article {KW,
    AUTHOR = {Kapustin, Anton and Witten, Edward},
     TITLE = {Electric-magnetic duality and the geometric {L}anglands
              program},
   JOURNAL = {Commun. Number Theory Phys.},
  FJOURNAL = {Communications in Number Theory and Physics},
    VOLUME = {1},
      YEAR = {2007},
    NUMBER = {1},
     PAGES = {1--236},
      ISSN = {1931-4523},
   MRCLASS = {14D21 (22E46 32G13 81T45 81T60)},
  MRNUMBER = {2306566},
MRREVIEWER = {Siye Wu},
       DOI = {10.4310/CNTP.2007.v1.n1.a1},
}

@article {Koi,
    AUTHOR = {Koike, Kazuhiko},
     TITLE = {On the decomposition of tensor products of the representations
              of the classical groups: by means of the universal characters},
   JOURNAL = {Adv. Math.},
  FJOURNAL = {Advances in Mathematics},
    VOLUME = {74},
      YEAR = {1989},
    NUMBER = {1},
     PAGES = {57--86},
      ISSN = {0001-8708},
   MRCLASS = {22E46},
  MRNUMBER = {991410},
MRREVIEWER = {Ronald C. King},
       DOI = {10.1016/0001-8708(89)90004-2},
}

@incollection {Kon95,
    AUTHOR = {Kontsevich, Maxim},
     TITLE = {Enumeration of rational curves via torus actions},
 BOOKTITLE = {The moduli space of curves ({T}exel {I}sland, 1994)},
    SERIES = {Progr. Math.},
    VOLUME = {129},
     PAGES = {335--368},
 PUBLISHER = {Birkh\"{a}user Boston, Boston, MA},
      YEAR = {1995},
   MRCLASS = {14N10 (14D22 14L30)},
  MRNUMBER = {1363062},
MRREVIEWER = {Anatoly Libgober},
       DOI = {10.1007/978-1-4612-4264-2\_12},
}

@book {LZ,
    AUTHOR = {Lando, Sergei K. and Zvonkin, Alexander K.},
     TITLE = {Graphs on surfaces and their applications},
    SERIES = {Encyclopaedia of Mathematical Sciences},
    VOLUME = {141},
      NOTE = {With an appendix by Don B. Zagier,
              Low-Dimensional Topology, II},
 PUBLISHER = {Springer-Verlag, Berlin},
      YEAR = {2004},
     PAGES = {xvi+455},
      ISBN = {3-540-00203-0},
   MRCLASS = {14H55 (05-02 05C10 05C30 05C50 14H10 14H30 32G15)},
  MRNUMBER = {2036721},
MRREVIEWER = {Athanase Papadopoulos},
       DOI = {10.1007/978-3-540-38361-1},
}

@article {Lem,
    AUTHOR = {Lemoine, Thibaut},
     TITLE = {Large {$N$} behaviour of the two-dimensional {Y}ang-{M}ills
              partition function},
   JOURNAL = {Combin. Probab. Comput.},
  FJOURNAL = {Combinatorics, Probability and Computing},
    VOLUME = {31},
      YEAR = {2022},
    NUMBER = {1},
     PAGES = {144--165},
      ISSN = {0963-5483},
   MRCLASS = {81T13 (43A75)},
  MRNUMBER = {4356462},
       DOI = {10.1017/s0963548321000262},
}

@article{Lem3,
    author = {Lemoine, Thibaut},
    title = "Almost flat highest weights and application to {W}ilson loops on compact surfaces",
    YEAR = {2025},
    primaryClass={math-ph},
    JOURNAL = {Probab. Theory Related Fields},
    FJOURNAL = {Probability Theory and Related Fields},
}

@article{Lem4,
      title={Two-dimensional {Y}ang-{M}ills theory via integrable probability}, 
      author={Thibaut Lemoine},
      year={2026},
      journal={Bull. Amer. Soc.},
      note={to appear},
      eprint={2508.16162},
      archivePrefix={arXiv},
      primaryClass={math.CO},
      url={https://arxiv.org/abs/2508.16162}, 
}

@misc{Lem26a,
      title={Universal dualities for {W}ilson loops in lattice {Y}ang-{M}ills}, 
      author={Thibaut Lemoine},
      year={2026},
      eprint={2604.16252},
      archivePrefix={arXiv},
      primaryClass={math-ph},
      url={https://arxiv.org/abs/2604.16252},
      note={arxiv:2604.16252}
}

@article{Lem26mf,
    author = {Lemoine, Thibaut},
    title = {The heat-kernel master field on $\mathbb{Z}^d$ at strong coupling},
    year = {2026},
    note = {in preparation}
}

@article {LM2,
    AUTHOR = {Lemoine, Thibaut and Ma\"{\i}da, Myl\`ene},
     TITLE = {Gaussian measure on the dual of {${\rm U}(N)$}, random
              partitions and topological expansion of the partition
              function},
   JOURNAL = {Ann. Probab.},
  FJOURNAL = {The Annals of Probability},
    VOLUME = {53},
      YEAR = {2025},
    NUMBER = {5},
     PAGES = {1738--1763},
      ISSN = {0091-1798,2168-894X},
   MRCLASS = {60B15 (05A17 43A75 81T13 81T35)},
  MRNUMBER = {4962730},
       DOI = {10.1214/24-aop1749},
       URL = {https://doi.org/10.1214/24-aop1749},
}

@article {Lev10,
    AUTHOR = {L{\'{e}}vy, Thierry},
     TITLE = {Two-dimensional {M}arkovian holonomy fields},
   JOURNAL = {Ast\'{e}risque},
  FJOURNAL = {Ast\'{e}risque},
    NUMBER = {329},
      YEAR = {2010},
     PAGES = {172},
      ISSN = {0303-1179},
      ISBN = {978-2-85629-283-9},
   MRCLASS = {58J65 (53C29 57M15 57R56 58D20 60J60)},
  MRNUMBER = {2667871},
MRREVIEWER = {Ambar N. Sengupta},
}

@incollection {LevMai2,
    AUTHOR = {L\'{e}vy, Thierry and Ma\"{\i}da, Myl\`ene},
     TITLE = {On the {D}ouglas-{K}azakov phase transition. {W}eighted
              potential theory under constraint for probabilists},
 BOOKTITLE = {Mod\'{e}lisation {A}l\'{e}atoire et {S}tatistique---{J}ourn\'{e}es {MAS}
              2014},
    SERIES = {ESAIM Proc. Surveys},
    VOLUME = {51},
     PAGES = {89--121},
 PUBLISHER = {EDP Sci., Les Ulis},
      YEAR = {2015},
   MRCLASS = {60K40 (60A10 60B15 60B20 81T13 82B26)},
  MRNUMBER = {3440793},
       DOI = {10.1051/proc/201551006},
}

@article {Mag,
    AUTHOR = {Magee, Michael},
     TITLE = {Random unitary representations of surface groups {I}:
              asymptotic expansions},
   JOURNAL = {Comm. Math. Phys.},
  FJOURNAL = {Communications in Mathematical Physics},
    VOLUME = {391},
      YEAR = {2022},
    NUMBER = {1},
     PAGES = {119--171},
      ISSN = {0010-3616},
   MRCLASS = {57M05 (32G15 57R57)},
  MRNUMBER = {4393965},
       DOI = {10.1007/s00220-021-04295-5},
}

@article {Mag2,
    AUTHOR = {Magee, Michael},
     TITLE = {Random unitary representations of surface groups, {II} : {T}he
              large n limit},
   JOURNAL = {Geom. Topol.},
  FJOURNAL = {Geometry \& Topology},
    VOLUME = {29},
      YEAR = {2025},
    NUMBER = {3},
     PAGES = {1237--1281},
      ISSN = {1465-3060},
   MRCLASS = {Prelim},
  MRNUMBER = {4918107},
       DOI = {10.2140/gt.2025.29.1237},
}

@article {MagPud,
    AUTHOR = {Magee, Michael and Puder, Doron},
     TITLE = {The asymptotic statistics of random covering surfaces},
   JOURNAL = {Forum Math. Pi},
  FJOURNAL = {Forum of Mathematics. Pi},
    VOLUME = {11},
      YEAR = {2023},
     PAGES = {Paper No. e15, 51},
   MRCLASS = {20C15 (20B30 20F34 20F65 20F70 20P05 57K20 60B15)},
  MRNUMBER = {4591384},
MRREVIEWER = {Jianchun Wu},
       DOI = {10.1017/fmp.2023.13},
}

@article {MN,
    AUTHOR = {Magee, Michael and Naud, Fr\'{e}d\'{e}ric},
     TITLE = {Explicit spectral gaps for random covers of {R}iemann
              surfaces},
   JOURNAL = {Publ. Math. Inst. Hautes \'{E}tudes Sci.},
  FJOURNAL = {Publications Math\'{e}matiques. Institut de Hautes \'{E}tudes
              Scientifiques},
    VOLUME = {132},
      YEAR = {2020},
     PAGES = {137--179},
      ISSN = {0073-8301},
   MRCLASS = {58J50 (05C50 11M36)},
  MRNUMBER = {4179833},
MRREVIEWER = {Anton Deitmar},
       DOI = {10.1007/s10240-020-00118-w},
}

@article {MNP,
    AUTHOR = {Magee, Michael and Naud, Fr\'{e}d\'{e}ric and Puder, Doron},
     TITLE = {A random cover of a compact hyperbolic surface has relative
              spectral gap {$\frac{3}{16}-\varepsilon$}},
   JOURNAL = {Geom. Funct. Anal.},
  FJOURNAL = {Geometric and Functional Analysis},
    VOLUME = {32},
      YEAR = {2022},
    NUMBER = {3},
     PAGES = {595--661},
      ISSN = {1016-443X},
   MRCLASS = {58J50 (05C50 32G15)},
  MRNUMBER = {4431124},
MRREVIEWER = {Sugata Mondal},
       DOI = {10.1007/s00039-022-00602-x},
}

@incollection {Mal,
    AUTHOR = {Maldacena, Juan},
     TITLE = {The large-{$N$} limit of superconformal field theories and
              supergravity},
      NOTE = {Quantum gravity in the southern cone (Bariloche, 1998)},
   JOURNAL = {Internat. J. Theoret. Phys.},
  FJOURNAL = {International Journal of Theoretical Physics},
    VOLUME = {38},
      YEAR = {1999},
    NUMBER = {4},
     PAGES = {1113--1133},
      ISSN = {0020-7748},
   MRCLASS = {81T60 (81T30 81T40 83E50)},
  MRNUMBER = {1705508},
       DOI = {10.1023/A:1026654312961},
}

@article {Mon22,
    AUTHOR = {Monk, Laura},
     TITLE = {Benjamini-{S}chramm convergence and spectra of random
              hyperbolic surfaces of high genus},
   JOURNAL = {Anal. PDE},
  FJOURNAL = {Analysis \& PDE},
    VOLUME = {15},
      YEAR = {2022},
    NUMBER = {3},
     PAGES = {727--752},
      ISSN = {2157-5045,1948-206X},
   MRCLASS = {58J50 (32G15)},
  MRNUMBER = {4442839},
MRREVIEWER = {Yuhao\ Xue},
       DOI = {10.2140/apde.2022.15.727},
       URL = {https://doi.org/10.2140/apde.2022.15.727},
}

@incollection {NO,
    AUTHOR = {Nekrasov, Nikita A. and Okounkov, Andrei},
     TITLE = {Seiberg-{W}itten theory and random partitions},
 BOOKTITLE = {The unity of mathematics},
    SERIES = {Progr. Math.},
    VOLUME = {244},
     PAGES = {525--596},
 PUBLISHER = {Birkh\"{a}user Boston, Boston, MA},
      YEAR = {2006},
   MRCLASS = {81T60 (05E10 11Z05 14D21 60C05 81T45)},
  MRNUMBER = {2181816},
MRREVIEWER = {Johan A. Martens},
       DOI = {10.1007/0-8176-4467-9\_15},
}

@misc{Nov2024,
      title={On the 2D {Y}ang-{M}ills/{H}urwitz Correspondence}, 
      author={Jonathan Novak},
      year={2024},
      eprint={2401.00628},
      archivePrefix={arXiv},
      primaryClass={math.CO}
}

@article {Oko,
    AUTHOR = {Okounkov, Andrei},
     TITLE = {Infinite wedge and random partitions},
   JOURNAL = {Selecta Math. (N.S.)},
  FJOURNAL = {Selecta Mathematica. New Series},
    VOLUME = {7},
      YEAR = {2001},
    NUMBER = {1},
     PAGES = {57--81},
      ISSN = {1022-1824},
   MRCLASS = {60C05},
  MRNUMBER = {1856553},
MRREVIEWER = {Hajime Yamato},
       DOI = {10.1007/PL00001398},
}

@article {OP2,
    AUTHOR = {Okounkov, A. and Pandharipande, R.},
     TITLE = {Gromov-{W}itten theory, {H}urwitz theory, and completed
              cycles},
   JOURNAL = {Ann. of Math. (2)},
  FJOURNAL = {Annals of Mathematics. Second Series},
    VOLUME = {163},
      YEAR = {2006},
    NUMBER = {2},
     PAGES = {517--560},
      ISSN = {0003-486X},
   MRCLASS = {14N35 (14N10 37K20 53D45)},
  MRNUMBER = {2199225},
MRREVIEWER = {Hsian-Hua Tseng},
       DOI = {10.4007/annals.2006.163.517},
}

@article {OO97,
    AUTHOR = {Okounkov, Andrei and Olshanski, Grigori},
     TITLE = {Shifted {S}chur functions},
   JOURNAL = {Algebra i Analiz},
  FJOURNAL = {Rossi\u{\i}skaya Akademiya Nauk. Algebra i Analiz},
    VOLUME = {9},
      YEAR = {1997},
    NUMBER = {2},
     PAGES = {73--146},
      ISSN = {0234-0852},
   MRCLASS = {05E05 (05E10 17B35 20G05)},
  MRNUMBER = {1468548},
MRREVIEWER = {Witold\ Kra\'{s}kiewicz},
}

@incollection {OO98,
    AUTHOR = {Okounkov, Andrei and Olshanski, Grigori},
     TITLE = {Shifted {S}chur functions. {II}. {T}he binomial formula for
              characters of classical groups and its applications},
 BOOKTITLE = {Kirillov's seminar on representation theory},
    SERIES = {Amer. Math. Soc. Transl. Ser. 2},
    VOLUME = {181},
     PAGES = {245--271},
 PUBLISHER = {Amer. Math. Soc., Providence, RI},
      YEAR = {1998},
      ISBN = {0-8218-0669-6},
   MRCLASS = {17B35 (05E05 20G05)},
  MRNUMBER = {1618763},
MRREVIEWER = {Witold\ Kra\'{s}kiewicz},
       DOI = {10.1090/trans2/181/08},
       URL = {https://doi.org/10.1090/trans2/181/08},
}

@article {PS,
    AUTHOR = {Pantev, Tony and Sharpe, Eric},
     TITLE = {Decomposition and the {G}ross-{T}aylor string theory},
   JOURNAL = {Internat. J. Modern Phys. A},
  FJOURNAL = {International Journal of Modern Physics A. Particles and
              Fields. Gravitation. Cosmology},
    VOLUME = {38},
      YEAR = {2023},
    NUMBER = {29-30},
     PAGES = {Paper No. 2350156, 81},
      ISSN = {0217-751X},
   MRCLASS = {81T13 (81T30)},
  MRNUMBER = {4683255},
       DOI = {10.1142/S0217751X23501567},
}

@article{PPSY,
      title={Wilson loop expectations as sums over surfaces on the plane}, 
      author={Minjae Park and Joshua Pfeffer and Scott Sheffield and Pu Yu},
      journal={Probab. Math. Phys.},
      year={2025},
      note = {to appear}
}

@article {Ram2,
    AUTHOR = {Ramgoolam, Sanjaye},
     TITLE = {Comment on two-dimensional {${\rm O}(N)$} and {${\rm Sp}(N)$}
              {Y}ang-{M}ills theories as string theories},
   JOURNAL = {Nuclear Phys. B},
  FJOURNAL = {Nuclear Physics. B. Theoretical, Phenomenological, and
              Experimental High Energy Physics. Quantum Field Theory and
              Statistical Systems},
    VOLUME = {418},
      YEAR = {1994},
    NUMBER = {1-2},
     PAGES = {30--44},
      ISSN = {0550-3213},
   MRCLASS = {81T40 (81T30)},
  MRNUMBER = {1281351},
MRREVIEWER = {C. Gauthier},
       DOI = {10.1016/0550-3213(94)90237-2},
       URL = {https://doi.org/10.1016/0550-3213(94)90237-2},
}

@article {Rus,
    AUTHOR = {Rusakov, B. Ye.},
     TITLE = {Loop averages and partition functions in {$\mathrm{U}(N)$} gauge
              theory on two-dimensional manifolds},
   JOURNAL = {Modern Phys. Lett. A},
  FJOURNAL = {Modern Physics Letters A. Particles and Fields, Gravitation,
              Cosmology, Astrophysics, Nuclear Physics, Accelerator Physics,
              Quantum Information},
    VOLUME = {5},
      YEAR = {1990},
    NUMBER = {9},
     PAGES = {693--703},
      ISSN = {0217-7323},
   MRCLASS = {81T13 (81T20 81T25)},
  MRNUMBER = {1051372},
MRREVIEWER = {E. Wieczorek},
       DOI = {10.1142/S0217732390000780},
       URL = {https://doi.org/10.1142/S0217732390000780},
}

@article {Sta,
    AUTHOR = {Stanley, Richard P.},
     TITLE = {The stable behavior of some characters of {${\rm SL}(n,{\bf
              C})$}},
   JOURNAL = {Linear and Multilinear Algebra},
  FJOURNAL = {Linear and Multilinear Algebra},
    VOLUME = {16},
      YEAR = {1984},
    NUMBER = {1-4},
     PAGES = {3--27},
      ISSN = {0308-1087},
   MRCLASS = {22E45 (05A19 17B10 20C15)},
  MRNUMBER = {768993},
MRREVIEWER = {George E. Andrews},
       DOI = {10.1080/03081088408817606},
       URL = {https://doi.org/10.1080/03081088408817606},
}

@article {Wit,
    AUTHOR = {Witten, Edward},
     TITLE = {On quantum gauge theories in two dimensions},
   JOURNAL = {Comm. Math. Phys.},
  FJOURNAL = {Communications in Mathematical Physics},
    VOLUME = {141},
      YEAR = {1991},
    NUMBER = {1},
     PAGES = {153--209},
      ISSN = {0010-3616},
   MRCLASS = {58G26 (14D20 32G13 58D27 81T13)},
  MRNUMBER = {1133264},
MRREVIEWER = {Dana S. Fine},
       URL = {http://projecteuclid.org/euclid.cmp/1104248198},
}

@article {Wit88,
    AUTHOR = {Witten, Edward},
     TITLE = {Topological sigma models},
   JOURNAL = {Comm. Math. Phys.},
  FJOURNAL = {Communications in Mathematical Physics},
    VOLUME = {118},
      YEAR = {1988},
    NUMBER = {3},
     PAGES = {411--449},
      ISSN = {0010-3616},
   MRCLASS = {81E13 (57R99 58E99 81E30 81E40 83C45)},
  MRNUMBER = {958805},
MRREVIEWER = {Kotik K. Lee},
       URL = {http://projecteuclid.org/euclid.cmp/1104162092},
}

@incollection {Wit91b,
    AUTHOR = {Witten, Edward},
     TITLE = {Two-dimensional gravity and intersection theory on moduli
              space},
 BOOKTITLE = {Surveys in differential geometry ({C}ambridge, {MA}, 1990)},
     PAGES = {243--310},
 PUBLISHER = {Lehigh Univ., Bethlehem, PA},
      YEAR = {1991},
   MRCLASS = {32G15 (14C17 14H15 32G81 58F07 81T40)},
  MRNUMBER = {1144529},
MRREVIEWER = {Steven Rosenberg},
}

\end{document}